%% file: RCWJMP.tex
\documentclass[aps,longbibliography]{revtex4-1}
\usepackage{amsmath}
\usepackage{amsfonts} \usepackage{amssymb} \usepackage{bbm}
\usepackage{graphicx}  
\usepackage{psfrag}    
\usepackage{natbib,hyperref}
\usepackage{subfigure}

\usepackage[draft]{fixme}

\input{RCWmJMP}
\begin{document}

\title{Asymptotic evolution of quantum walks with random coin}
\author{A. Ahlbrecht}
\email{andre.ahlbrecht@itp.uni-hannover.de}
\author{H. Vogts}
\email{holger.vogts@gmx.de}
\author{A.H. Werner}
\email{albert.werner@itp.uni-hannover.de}
\author{R.F. Werner}
\email{Reinhard.Werner@itp.uni-hannover.de}
\affiliation{Inst. f. Theoretical Physics, Leibniz Universit\"{a}t Hannover,
Appelstr. 2, 30167 Hannover, Germany}

\begin{abstract}
We study the asymptotic position distribution of general quantum walks on a lattice, including walks with a random coin, which is chosen from step to step by a general Markov chain. In the unitary (i.e., non-random) case, we allow any unitary operator, which commutes with translations, and couples only sites at a finite distance from each other. For example, a single step of the walk could be composed of any finite succession of different shift and coin operations in the usual sense, with any lattice dimension and coin dimension. We find ballistic scaling, and establish a direct method for computing the asymptotic distribution of position divided by time, namely as the distribution of the discrete time analog of the group velocity. In the random case, we let a Markov chain (control process) pick in each step one of finitely many unitary walks, in the sense described above. In ballistic order we find a non-random drift, which depends only on the mean of the control process and not on the initial state. In diffusive scaling the limiting distribution is asymptotically Gaussian, with a covariance matrix (diffusion matrix) depending on momentum. The diffusion matrix depends not only on the mean but also on the transition rates of the control process. In the non-random limit, i.e., when the coins chosen are all very close, or the transition rates of the control process are small, leading to long intervals of ballistic evolution, the diffusion matrix diverges. Our method is based on spatial Fourier transforms, and the first and second order perturbation theory of the eigenvalue 1 of the transition operator for each value of the momentum.
\end{abstract}
\maketitle

\input{RCWsJMP}
\input{RCWaJMP}
\input{RCWdJMP}

\input{RCWeJMP}

\section*{Acknowledgements}
We gratefully acknowledge the support of the DFG (Forschergruppe 635) and the EU projects CORNER, QUICS and CoQuit.

\bibliography{walklitSV} 

\end{document}

%% file: RCWmJMP.tex
\def\idty{{\leavevmode\rm 1\mkern -5.4mu I}} 
\def\id{{\rm id}}
\def\Rl{{\mathbb R}}\def\Cx{{\mathbb C}}
\def\Ir{{\mathbb Z}}\def\Nl{{\mathbb N}}
\def\Rt{{\mathbb Q}}
\def\norm #1{\Vert #1\Vert}

\def\ind{{\mathop{\rm ind}\nolimits}\,}
\def\bra #1{\langle #1\vert}
\def\ket #1{\vert #1\rangle}
\def\braket #1#2{\langle #1 \vert #2\rangle}
\def\ketbra #1#2{\vert #1\rangle \langle #2\vert}

\def\tr{\mathop{\rm tr}\nolimits}
\def\abs#1{\vert#1\vert}

\def\order{{\mathbf o}}
\def\Order{{\mathbf O}}

\mathchardef\ree="023C \mathchardef\imm="023D  

\def\BB{{\mathcal B}}

\def\DD{{\mathcal D}}
\def\HH{{\mathcal H}}
\def\NN{{\mathcal N}}\def\KK{{\mathcal K}}

\def\Fou{{\mathcal F}}

\def\L#1{{\mathcal L}^{#1}} 

\def\VV{{\mathcal V}}

\def\Wa{{\mathbf W}} 
\def\Wt{\widetilde\Wa} 
\def\Va{{\mathbf V}} 
\def\Vt{\widetilde\Va}
\def\Ma{{\mathbf M}} 
\def\ma{{\mathbf m}} 
\def\invrho{{\overline\rho\,}}
\def\invm{{\overline\ma}}

\let\veps\varepsilon
\def\Pb{\mathbf P}  
\def\Db{\mathbf D}

\def\fatasterisk{\lower 4pt\hbox {\Large\bf\char42}}

\def\Xh{\widehat X} 

\def\lamv{\lambda{\cdot}v} 
\def\ordert#1{\order(t^{-#1})}

\newtheorem{thm}{Theorem}
\newtheorem{lem}[thm]{Lemma}
\newtheorem{prop}[thm]{Proposition}

\newtheorem{cor}[thm]{Corollary}
\newtheorem{ass}{Assumption}
\def\assref#1{Assumption~\ref{ass:#1}}

\def\eqref#1{(\ref{#1})}
\def\secref#1{Sect.~\ref{sec:#1}}

\newenvironment{proof}[1][]{\par\vskip12pt\noindent%
     {\it Proof\ \ifx#1==\else{ #1\/:}\newline\fi\ }\bgroup}
     {\egroup\quad\strut\hfill$\blacksquare$\vskip12pt}

%% file: RCWsJMP.tex
\section{Introduction}
A quantum walk, the counterpart of a classical random walk, models a quantum particle moving randomly in discrete time steps on a lattice, but contrary to the classical case one has to consider a quantum particle with an internal degree of freedom, usually called the {\bf coin space}, in order to get nondeterministic behavior \cite{Ambainis,Meyer}. Frequently, the dynamics of a quantum walk is decomposed into a unitary operation acting on the internal degree of freedom, called the {\bf coin operation}, followed by a state dependent spatial {\bf shift operation}. This constructive definition permits many interesting examples of quantum walks and moreover it yields a decomposition into experimentally realizable operations, which already have been implemented \cite{Bonn}. However, in general there is no need to restrict to dynamics that can be decomposed into coin operations and conditional shifts. Indeed, we will consider a more general axiomatic definition of quantum walks without any intrinsic connection between the lattice dimension and the number of internal states of the particle, which is suggested by the constructive approach to quantum walks. Without adhering to such a connection, quantum walks with memory as introduced by \citet{McGettrick} can easily be formulated in our language and their analysis, in particular the calculation of their asymptotic behavior, is much simplified. The definition of quantum walks we are going to incorporate imposes two axioms on the time evolution. The first one is {\bf translation invariance}, i.e. the time evolution is homogeneous in space, meaning it commutes with lattice translations. The second axiom is a {\bf locality condition}, we want to exclude infinite propagation speed of the particle and therefore we assume that the maximal step size of the particle is bounded. Interestingly, it turns out that for one dimensional lattices and unitary time evolution this axiomatic definition exactly matches the constructive approach \cite{Holgerthesis}, whereas in higher dimensions this is not clear.

In this paper we consider quantum walks on lattices in any space dimension. Their internal degree of freedom is described in any finite dimensional Hilbert space. According to our axiomatic definition the dynamical step will be translation invariant and the step size can be many lattice constants long but finite. Another assumption we may impose on the time evolution is unitarity. Although we also discuss general decoherent quantum walks compatible with our definition, we will later consider a particular class of decoherent quantum walks emerging from a set of unitary ones. What we want to study is disorder in time, rather than in space: We assume that there is some collection of possible steps, from which a classical Markov process chooses one instance. Thus successive steps are not independent, although steps separated by a long time will be practically independent.

For this class of models we aim to compute the asymptotic probability distribution for the position variable $Q(t)$ after $t$ steps. Such analysis has been carried out before, but to the best of our knowledge not in full generality. Most of the analysis so far is concerned with unitary quantum walks \cite{Ambainis,Carteret,Grimmett} or under certain constraints, e.g. in one lattice dimension and with step size one \cite{Konno,KonnoB,Bressler}. A numerical study of certain examples of unitary quantum walks in two lattice dimensions was performed by \citet{Sanders} and \citet{Kollar}. Recently \citet{Baryshnikov} and \citet{BresslerB} computed the limiting probability distributions for various one and two dimensional unitary quantum walks. Also, the study of quantum walks subjected to decoherence is mostly concerned with particular models or subclasses of quantum walks, see e.g. the review article by \citet{KendonRev} and references therein for an overview on decoherent quantum walks on one dimensional lattices, cycles or hypercubes. The kinds of decoherence considered can be roughly classified as measurement induced decoherence \cite{Romanelli,Romanelli2,Chand,Zhang,Shikano} and decoherence caused by sloppy control of quantum operations \cite{Abal,Biham,Chand2,Brun2}. A subclass of the second decoherence mechanism which we are going to consider is the case when in each time step a unitary quantum walk is chosen from a given set with a certain probability, for example the coin operation of a one dimensional quantum walk may be chosen from a fixed set of unitaries at each time step \cite{Brun,BrunTrans,KonnoManyCoins,Buzek}.
A more general kind of decoherence compatible with our definition of a quantum walk was introduced by \citet{Iran} and \citet{Iran2}, and an analysis of the asymptotic behavior was performed based on the calculation of the first and second moment of the distribution.

\begin{figure}[htb]

\includegraphics[width=0.5\textwidth]{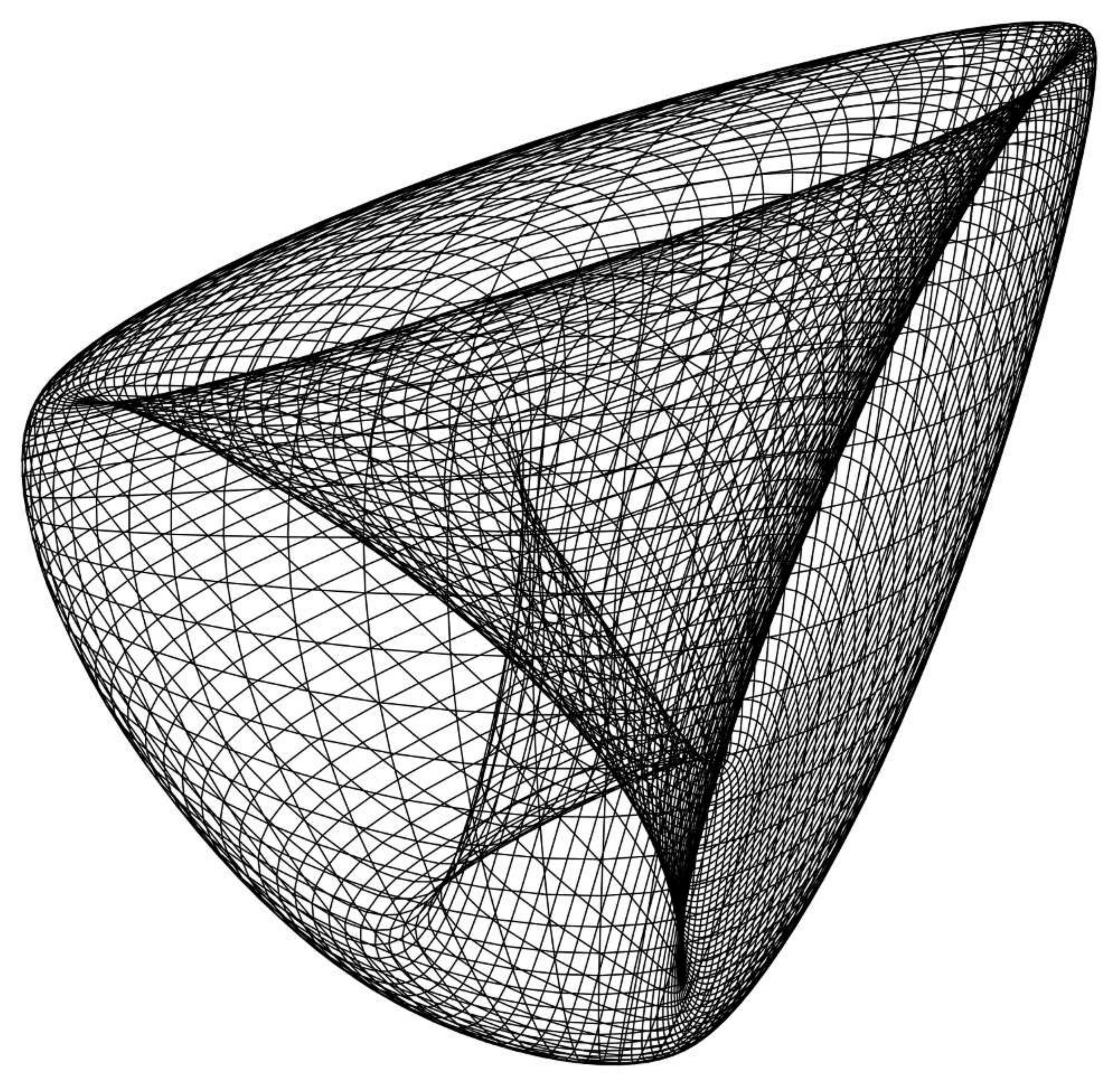}
    \caption{The figure shows a caustic of a unitary quantum walk on a two-dimensional lattice. Closed curves correspond to the image of the map $p\mapsto \nabla \omega (p)$ applied to a discrete set of coordinate lines in momentum space $[0,2\pi)^2$. The caustic, i.e. the region where the line density diverges, is exactly the region where the asymptotic probability distribution of the quantum walk exhibits peaks.}
\label{Fig:Caustics}
\end{figure}

One key question we address is whether the spreading is {\bf ballistic}, i.e. whether $Q(t)/t$ converges in distribution, or whether it is {\bf diffusive}, i.e. $Q(t)/\sqrt t$ has a meaningful limit as $t\to\infty$. Now without randomness, i.e. when we use always the same step, it is known that the spreading is ballistic. In this case, as we show in Sec. \ref{sec:unitary}, the asymptotic distribution of the particles position is determined by the {\bf dispersion relation} of the unitary quantum walk, more precisely, its derivative with respect to momentum, the {\bf group velocity}, dictates the asymptotic behavior. Commonly, the asymptotic distribution shows characteristic peaks, which can be understood as {\bf caustics} of the dispersion relation, see Fig. \ref{Fig:Caustics}. Such caustics have also been observed by \citet{Baryshnikov}, we give a detailed description of this concept in Sec. \ref{sec:unitary}. Since a classical random walk satisfies our description it is also clear that for stringent randomness we expect diffusive scaling. In fact, diffusive behavior is a common feature of decoherent quantum walks, e.g. for quantum walks with multiple coins \cite{Brun,BrunTrans,Biham,KonnoManyCoins} or dynamic gaps in the lattice \cite{Romanelli2,Abal,KendonDynGaps} this has been observed. On the basis of computations of the second moments, it has been shown recently \cite{Iran} that diffusive scaling holds even for very low randomness, i.e. when either all the steps used are nearly the same, or if one step is chosen almost always. Our results support this conclusion. In addition, we compute the asymptotic distribution of $Q(t)/\sqrt t$ for every initial state. It is Gaussian in every momentum component (in a sense specified below). As is to be expected, the diffusion constant (the limit of $\langle Q^2(t)\rangle/t$) diverges at low randomness. In such cases the system will initially evolve ballistically, and then exhibit a crossover to diffusive scaling, see Fig. \ref{Fig:Introduction1D} and \ref{Fig:Variance}.

\begin{figure}[htb]

  \centering
  \subfigure[]{\includegraphics[width=0.45\textwidth]{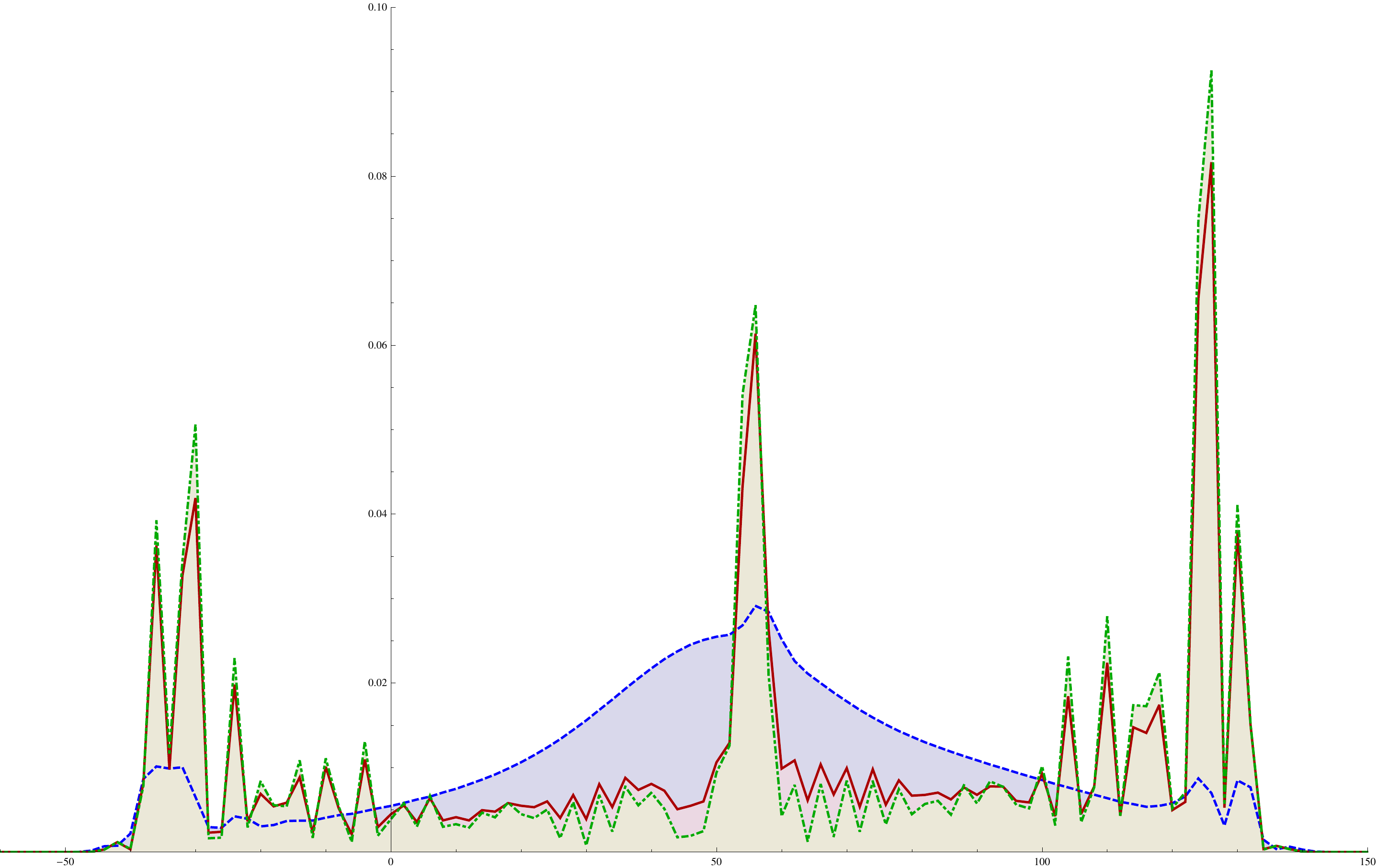}}
  \hspace{1cm}
  \subfigure[]{\includegraphics[width=0.45\textwidth]{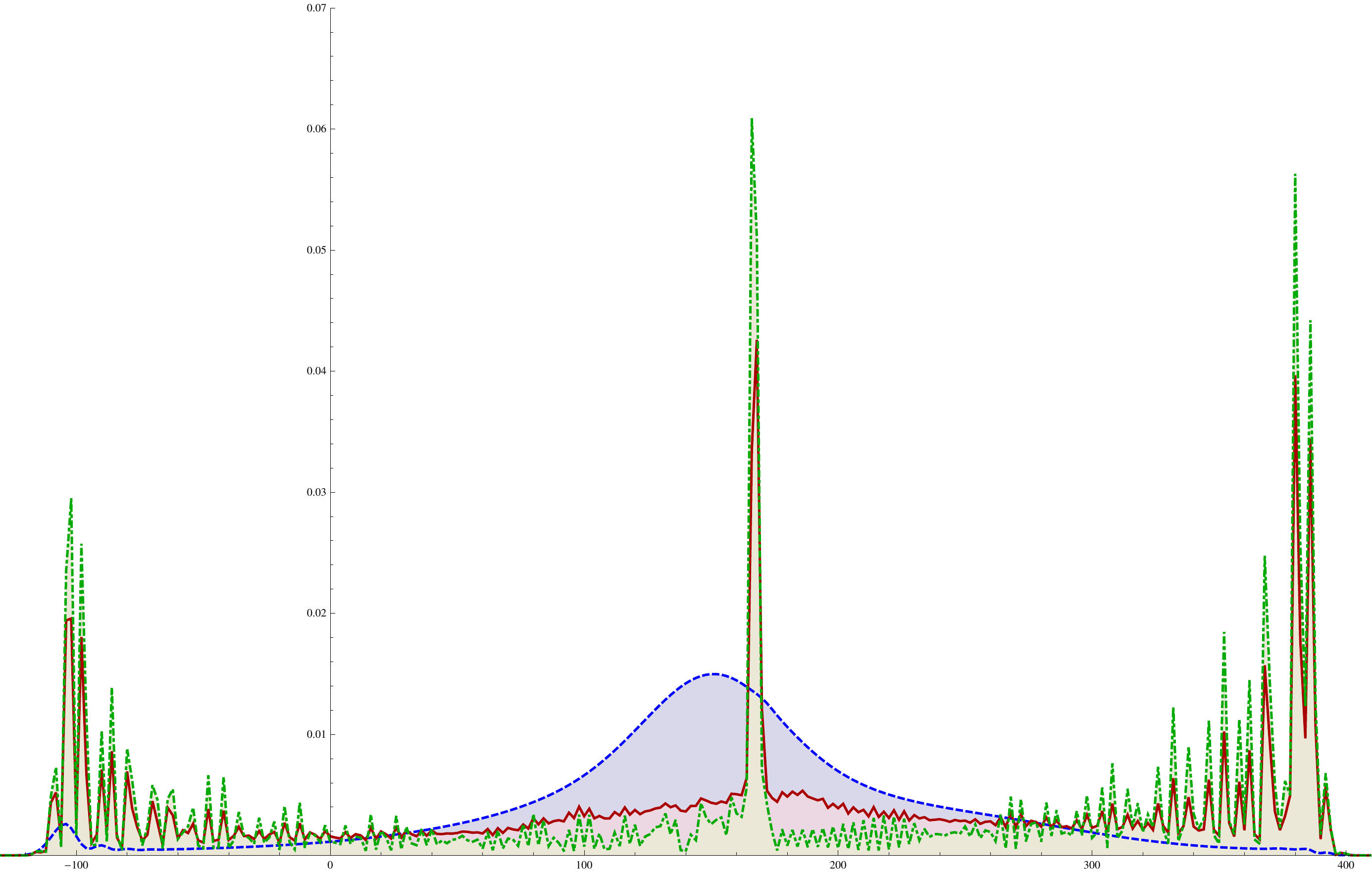}}
\label{Fig:Introduction1D}
\newline

  \centering
  \subfigure[]{\includegraphics[width=0.45\textwidth]{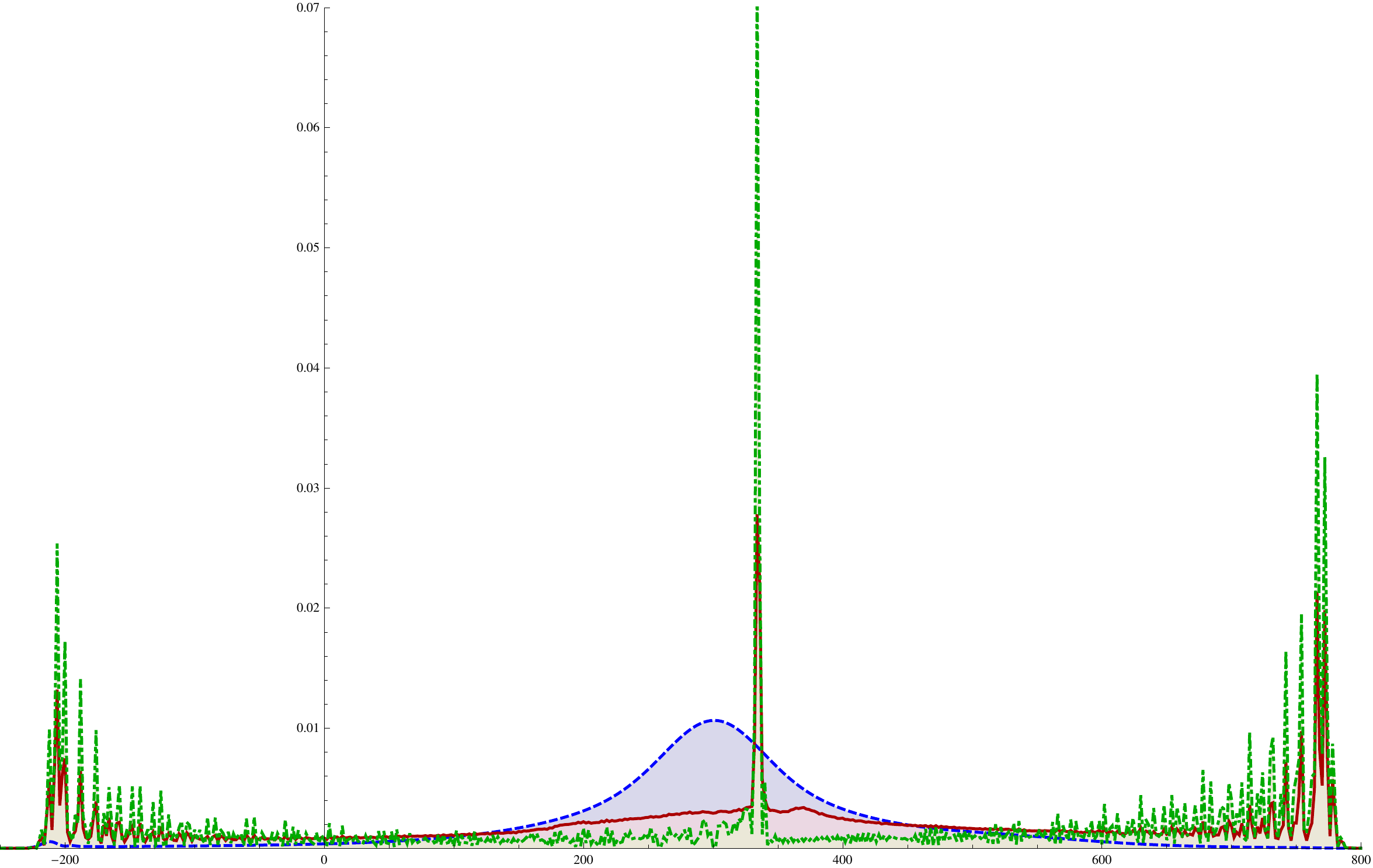}}
  \hspace{1cm}
  \subfigure[]{\includegraphics[width=0.45\textwidth]{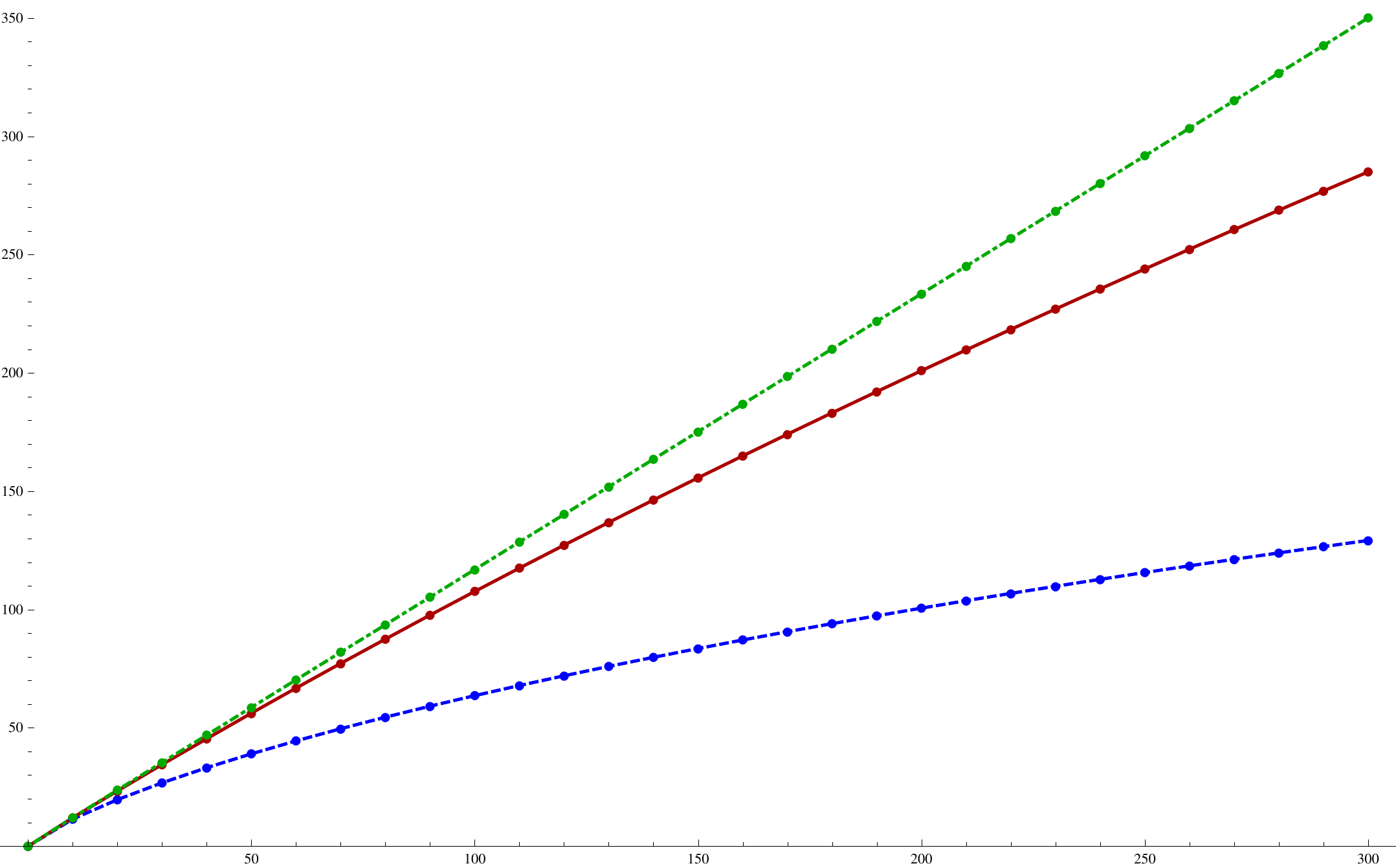}}
    \caption{(color online) Figures (a)-(c) show the time evolution of a quantum walk with index three and three internal states for $50$, $150$ and $300$ time steps. The green/dot-dashed curve corresponds to the undisturbed walk with a coin that is constant in time, whereas the red/solid and blue/dashed curves correspond to a Bernoulli process where in every time step instead of the coin of the green/dot-dashed quantum walk a different coin is applied with probability $0.01$ or $0.2$, respectively. Graph (d) shows the standard deviation of the position distributions with a visible crossover from ballistic to diffusive spreading behavior for the perturbed quantum walks.}
\label{Fig:Variance}
\end{figure}

The methods used so far to obtain the asymptotic position distribution can be summarized as Fourier methods, see e.g. \citet{Ambainis,Grimmett}, and combinatorial calculations as used by \citet{Konno,KonnoCombinatorics}. Another approach, closely related to Fourier methods is the generating function formalism \cite{Bressler,Zhang}. Our method involves Fourier methods as a way of describing a translation invariant system and its time evolution. The second main ingredient of our formalism is perturbation theory. We phrase the problem of finding the asymptotic distribution of the scaled random variables $Q(t)/t$ and $Q(t)/\sqrt t$ as the problem of applying a perturbed operator infinitely many times to an eigenvector of the unperturbed operator in the limit of vanishing perturbation. This becomes a nontrivial problem because both limits do not commute. The advantage of our method is that it gives us full information about the limiting position distribution of a large class of quantum walks, whereas other approaches often only give partial information, e.g. about the first two moments of the limiting position distribution.

Our paper is organized as follows: In the next section we will introduce notation and basic concepts. We close by describing the most general translation invariant quantum evolution with strictly finite step sizes, extending a result prematurely claimed to be exhaustive by \citet{Iran}. The type we actually consider in most of the paper is a special case and essentially the same as the one considered by \citet{Iran}. We then briefly review the non-random case, which is a prerequisite to understanding the general case. In Sect. \ref{sec:ballistic}, \ref{sec:diffusive} and \ref{sec:Bernoulli} we then state our result with all assumptions spelled out, and with a description of the procedure to compute the limit distributions. In the final section we gather some examples of quantum walks which violate our assumptions and analyze their asymptotic behavior.

\section{The systems}
The underlying {\bf lattice} of the system will be denoted by $X$, and is a subset of $\Rl^s$. We will always take $X=\Ir^s$, but we would like to stress that this covers also much more complex periodic structures, like triangular lattices, honeycombs, Kagome lattices, and so on in higher dimensions. Indeed, for us the role of the lattice is mainly that of the abstract symmetry group. For a general periodic pattern, the translation symmetries are by vectors of the form $\sum_{i=1}^sx_i\vec a_i$ for some basis of lattice vectors $\vec a_i$ and integer coefficients $x_i$. One can then choose a unit cell, so that every point in $\Rl^s$ is uniquely obtained as a point in that cell plus a lattice vector. Then we consider all lattice points inside the unit cell as one super-site with internal structure, and take the integer coefficients $x_i$ themselves to label the lattice translations. This brings us back to the choice $X=\Ir^s$.
The position degree of freedom of the walking quantum particle is thus described in the Hilbert space $\ell^2(X)$, the square-summable functions $\psi:X\to\Cx$.

The {\bf internal states} of the walking particle are described by a finite dimensional Hilbert space $\KK$. In many papers on quantum walks this is called the {\bf coin space}. We will follow this terminology although the dynamics we consider usually does not have a simple decomposition into coin operation and conditional shifts. The Hilbert space of the system is thus $\ell^2(X)\otimes\KK$, and it will be convenient to identify this with $\ell^2(X,\KK)$, the set of functions $\psi:X\to\KK$ such that $\norm\psi^2=\sum_{x\in X}\norm{\psi(x)}^2<\infty$. Translations now act as shifts of the argument, i.e. $(U_x\psi)(y)=\psi(y-x)$.

Since we look at translation invariant systems, the analysis can be simplified considerably by taking Fourier transforms. Let $\Xh$ denote the dual group, in our case concretely parameterized as the space of momentum vectors $p\in(-\pi,\pi]^s\subset\Rl^s$. The Fourier transform is then the unitary operator
\begin{equation}\label{Fou}
    (\Fou\psi)(p)=(2\pi)^{-s/2}\sum_{x\in X}e^{ip\cdot x}\psi(x)\, .
\end{equation}
Thus $\Fou$ is a map from $\ell^2(X,\KK)$ to $\L2(\Xh,\KK)$. We can also think of this equation as the representation of a general function on $\Rl^s$, which is periodic with respect to the shifts in $(2\pi\Ir)^s$. Of particular importance will be the {\bf trigonometric polynomials} given by such a series with just finitely many non-zero terms.

A {\bf unitary quantum walk} is given by a unitary operator $W$, which commutes with translations and has the property that $(W\psi)(x)$ depends only on the values of $\psi(y)$ such that $x-y\in\NN$, where $\NN$ is a fixed finite set called the neighborhood scheme of $W$. Because $W$ commutes with translations it can be diagonalized jointly with the translation operators, i.e. it becomes a multiplication operator in momentum space:
\begin{equation}\label{Wp}
    (\Fou W\psi)(p)=W(p)(\Fou\psi)(p),
\end{equation}
where, for each $p$, $W(p)$ is a unitary operator on $\KK$, and each matrix element of $W(p)$ is according to the locality condition a trigonometric polynomial. More precisely, the only monomials $e^{ix\cdot p}$ appearing with non-zero coefficients are those with $x\in\NN$. In electrical engineering such operators are called para-unitaries. One can show \cite{Bratteli,Doganata,Nguyen} that, at least for $s=1$, every such operator can be decomposed into a finite product of $p$-independent unitaries (``coin tosses'') and diagonal unitaries which only have elements of the form $e^{ix\cdot p}$ on the diagonal (``conditional shifts'').

In order to characterize unitary quantum walks it is useful to consider the determinant of the walk operator $W(p)$, which is also a trigonometric polynomial. Since $1/\det W(p)=\det(W(p)^*)$ is also a polynomial, it must be a monomial $e^{ix\cdot p}$ for some $x\in X$. We call this $x$ the {\bf index} \cite{index} of $W$ and by definition we have
\begin{equation}\label{index}
\det W(p) = \det W(0)\cdot e^{\ind W\cdot p}\,.
\end{equation}

Turning now to irreversible processes, let us first introduce the kind of randomness, which is analyzed in our main result. We consider unitary walks, such that the particular unitary applied at time $t$ is chosen randomly. In this way we want to model experiments in which the randomness comes from fluctuations in some external parameters controlling the quantum operations. If the time scale of these fluctuations covers several steps, it is unreasonable to assume that the unitaries in successive time steps are independent. Therefore we allow the external parameters to be given by a Markov process, which we call the {\bf control process}. We fix some dependence $\gamma\mapsto W_\gamma$ of the walk unitaries on the value $\gamma$ of the control process. Then the following two steps are iterated:  At time step $t$ the walk unitary $W_\gamma(t)$ acts on the quantum system, where $\gamma(t)$ is the current value of the control process.  Then the next value $\gamma(t+1)$ is drawn according to the transition probability law of the control process. Our aim is to derive the long time behavior of this scheme, particularly the distribution of the particle's position at large times.

Depending on the transition probabilities, we can describe systems in which the walk unitaries of successive steps are either strongly correlated, or nearly independent. We assume that in the long run the process goes to an equilibrium probability density. This stationary distribution will be taken as the {\bf initial distribution}. In this way we express the condition that there are no correlations between successive runs of the walk, as the statistical data are collected. Otherwise, non-trivial correlations would exist between successive runs, and the experimental results would depend on how quickly the next run is initiated. Of course, this kind of correlations is possible, and may actually occur in experiments, but it would mean a deviation from the paradigm of statistical data collection. Theoretically the appropriate response would be to describe not single runs, but batches of runs with controlled relative timing. This is a complication we do not want to consider in this paper. So, as already stated, we will assume that the initial distribution is the stationary one.

Of course, further generalizations are possible. Firstly, we may have decoherence in the individual steps, so that the single step is not given by a random unitary, but by a completely positive map, whose Kraus operators are not multiples of unitaries. We will analyze the general form of such processes in the next section. We will also point out (Sect.~\ref{sec:ballistic}) the asymptotic results which can be obtained for this more general class of decoherent walks using the method of the main theorem together with an assumption concerning the general form of the Kraus operators. However, in the whole paper we will keep translation invariance. When this is broken, perhaps randomly, a whole new range of phenomena appears. The most interesting is the analogue of Anderson localization: Whereas the decoherence studied in this paper slows the spreading of the walk from ballistic scaling ($\sim t$) to diffusive scaling ($\sim t^{1/2}$), random space dependent coins typically stop the spreading altogether ($\sim t^0$) in the sense that an initially localized state will remain finitely localized for all times with arbitrarily high probability. Here the localization region depends on the initial state and the allowed error probability, but not on time. This has been demonstrated by \citet{andeloc} using methods adapted from the much better studied case of continuous time Anderson localization. Eventually, it will be possible to combine randomness in space and time, but this is definitely beyond the scope of the present work.

\subsection{General form of decoherent, translation invariant quantum walks}
\label{sec:geneSys}
Turning now to decoherent dynamics we will no longer consider pure quantum states $\psi \in \ell^2(X,\KK)$ but density operators $\rho$, i.e. bounded and positive operators with unit trace. Hence, $\rho$ is an element of the space $\BB(\HH)$ of bounded operators on the Hilbert space $\HH=\ell^2(X,\KK)$. Generically, a density operator $\rho$ will be non-translation invariant, e.g. $\rho$ may be supported on finitely many sites of the underlying lattice. Hence, the Fourier transform of $\rho$ is not given by a multiplication operator as in \eqref{Wp}, but an operator depending on two variables $p$ and $p'$. An example of such an operator is a state which is a mixture of finitely many pure states $\psi $ that are all supported on finitely many sites of the lattice:
\[
\rho =\sum_i \lambda_i\ketbra{\psi_i}{\psi_i}\quad \Rightarrow \quad \rho(p,p')=\Fou \rho \Fou^*=\sum_i \lambda_i\ketbra{\psi_i(p)}{\psi_i(p')}
\]
If we consider $\rho(p,p')$ as a linear map its action and the trace are given by the formulas
\[
(\Fou \rho \psi)(p)=\int dp'\,\rho(p,p')\cdot\psi(p')\quad \mathrm{and}\quad\tr \rho =\int dp\,\tr \rho (p,p)=\int dp\,\tr \rho (p)\,,
\]
where we abbreviated $ \rho (p,p)=\rho (p)$. A general non-unitary time evolution of a quantum system is described by the concept of a quantum channel $\Wa$. In the following we will describe the dynamics in the Heisenberg picture, i.e. we will evolve the observables of the system rather than the states. Hence, a quantum channel $\Wa:\BB(\HH_2)\rightarrow \BB(\HH_1)$ is formally given by a linear map which is completely positive and unital, i.e. for $n\in\Nl$ and $A\in \BB(\HH_2\otimes \Cx^n)$
\[
A\geq 0 \,\Rightarrow\,\Wa\otimes \id_{\Cx^n}(A)\geq 0 \quad \mathrm{and} \quad \Wa(\idty_{\HH_2})=\idty_{\HH_1}\,.
\]
Here, $\HH_1$ is the Hilbert space of the initial system, and $\HH_2$ describes the system after the time evolution. The representation theorem by Stinespring \cite{Stinespring,Paulsen} states that a completely positive and unital map $\Wa$ can be written explicitly as
\[
\Wa(A)=\VV^*(A\otimes \idty_{\DD})\VV\quad ,\forall A \in \BB(\HH_2)\,,
\]
where $\VV:\HH_1\rightarrow \HH_2\otimes \DD$ is an isometry, i.e. $\VV^*\VV=\idty_{\HH_1}$, and $\DD$ is called the dilation space. By choosing an orthonormal basis $e_i$ in $\DD$ and writing $\idty_\DD=\sum_i \ketbra{e_i}{e_i}$ one obtains the Kraus representation \cite{Kraus,Paulsen} of a quantum channel $\Wa$ which reads
\[
\Wa(A)=\sum_i K_i^*AK_i\quad \forall A\in \BB(\HH_2)\,.
\]
The relation between the $K_i$ and the isometry $\VV$ is given by
\[
\braket{\phi}{K_i\psi}=\braket{\phi\otimes e_i}{\VV\psi}\quad \phi\in \HH_2\, ,\, \psi\in\HH_1\,.
\]
The question we are going to address now is which isometries $\VV$, respectively Kraus operators $K_i$, represent a translation invariant quantum walk $\Wa$. We denote translations by lattice vectors $x\in X$ by $\tau_x$, that is, $\tau_x \in \BB(\ell^2(X)\otimes \KK)$ is defined via
\[
\tau_x (\ketbra{y}{z}\otimes M )=\ketbra{y+x}{z+x}\otimes M\, ,\quad y,z\in X\, ,\,M\in\BB(\KK) \,.
\]
With this definition translation invariance of the quantum walk $\Wa$ is expressed by
\begin{equation}\label{TransInv}
\Wa(\tau_x (A))=\Wa (A) \, ,\quad \forall \,x\in X\, ,\,A\in \BB(\ell^2(X)\otimes \KK) \,.
\end{equation}
In order to exclude infinite propagation speed we also impose the following locality condition on the quantum walk $\Wa$, which by \eqref{TransInv} can be chosen translation invariant. We assume there exists a finite neighborhood scheme $\mathcal{N}\subset X$ such that for arbitrary internal states $\phi,\psi \in \KK$ and $M\in\BB(\KK)$
\begin{equation}
\label{Locality}
\left\langle   k\otimes \phi \vert \Wa (\ketbra{x}{y}\otimes M)\vert l\otimes \psi\right\rangle =0 \quad \text{if}\quad k-x \notin \mathcal{N}\,\text{ or }\,y-l\notin \mathcal{N}\,.
\end{equation}
The following theorem characterizes all translation invariant quantum walks with finite propagation speed.
\begin{thm}
\label{GenForm}
Let $\Wa$ be a translation invariant quantum walk on $\HH=\ell^2(X)\otimes \KK$, that is, $\Wa\,:\BB(\HH)\rightarrow \BB(\HH)$ is a completely positive map respecting \eqref{TransInv} and \eqref{Locality}. Then, there exists a dilation space $\DD$ and a unitary representation $\{\tilde U_x\}_{x\in X}$ of $X$ on $\DD$ together with operators $v_y\,:\,\KK\rightarrow \KK\otimes \DD$, $y\in\mathcal{N}$, such that an isometry $\VV\,:\HH\rightarrow \HH\otimes \DD$ representing $\Wa$ is given by
\begin{equation}
\label{IsometryDef}
\VV\ket{x\otimes\phi} = \sum_{ y\in\mathcal{N} }\ket{x+y}\otimes \ket{ (\idty_{\KK}\otimes \tilde U_x) v_y\,\phi}\,,\quad  x\in X\, ,\phi \in \KK\,.
\end{equation}
Conversely, a unitary representation $\{U_x\}_{x\in X}$ of $X$ on $\DD$ together with operators $v_y\,:\,\KK\rightarrow\KK\otimes\DD$, $y\in\mathcal{N}$, satisfying the normalization condition
\begin{equation}
\label{Normalization}
\sum_{y\in\mathcal{N}\cap(\mathcal{N}-x) }v_{x+y}^*(\idty_{\KK}\otimes \tilde U_{x})v_y=\idty_{\KK}\cdot \delta_{x0}\,,
\end{equation}
defines a translation invariant quantum walk $\Wa$ via $\eqref{IsometryDef}$.
\end{thm}
\begin{cor}
By choosing orthonormal bases $e_j\in \DD$ and $e_\alpha\in \KK$ we get the Kraus operators $K_j \in \BB(\ell^2(X)\otimes \KK)$ corresponding to the isometry $\VV$ of Theorem \ref{GenForm} via
\[
\bra{z\otimes e_\beta}K_j\ket{x\otimes e_\alpha}=\braket{e_\beta\otimes e_j}{(\idty\otimes \tilde U_x) v_{z-x}\vert e_\alpha}\,.
\]
\end{cor}
\begin{proof}[of Theorem \ref{GenForm}]
A general isometry $\VV$ is given by the relation
\[
\VV\ket{x\otimes \phi}= \sum_{y\in\Ir}\ket{x+y}\otimes \ket {v_y(x)\,\phi}\, ,
\]
with operators $v_y(x)\,:\, \KK\rightarrow\KK\otimes \DD$. Invariance under translation by $x\in X$ requires $\VV U_x = U_x\otimes \tilde U_x \VV$, where $U_x$ is the translation operator on $\HH$ and the $\tilde U_x$ form a representation of $X$ on $\DD$. The intertwining relation of $\VV$ leads to $v_y(x)=\tilde U_x v_y(0)=:\tilde U_x v_y$. The locality condition \eqref{Locality} assures $v_y=0$ if $y\notin \mathcal{N}$ and the normalization condition is a consequence of the isometry condition $\delta_{xy}\cdot\braket{\phi}{\psi}=\braket{x\otimes \phi}{\VV^*\VV \vert y\otimes\psi}$.
\end{proof}
If the spectrum of the operators $\tilde U_x$ is only pure point it is easy to define Fourier transformed versions of $\VV$ and $K_j$. By choosing an orthonormal basis $e_j \in\DD$, $j\in J$, of eigenvectors of the $\tilde U_x$ and writing $\tilde U_x = \sum\limits_{j\in J} e^{iq_j\cdot x}\ketbra{e_j}{e_j}$ with $q_j\in \Rl$ we get the following corollary.
\begin{cor}
\label{cor:KrausMomShift}
If the spectrum of the $\tilde U_x$ in Theorem \ref{GenForm} consists only of pure point spectrum and $e_j$ labels a common eigenbasis of the $U_x$ with eigenvalues $e^{iq_j\cdot x}$, then the isometry $\VV$ and the Kraus operators $ K_j$ in momentum space are given by
\[
\braket{\phi\otimes e_j}{(\VV\psi)(p)}=\sum_{ y\in \mathcal{N}}e^{ip\cdot y}\braket{\phi}{v_y \psi(p+q_j)}\,,\quad \phi,\,\psi \in\KK
\]
\[
\braket{\phi}{(K_j\psi)(p)}=\sum_{y\in\mathcal{N}}e^{ip\cdot y} \braket{\phi}{v_y \psi(p+q_j)}\,,\quad \phi,\,\psi \in\KK\,.
\]
\end{cor}
In general, the operators $\tilde U_x$ may also exhibit continuous spectrum. A simple example of such a translation invariant quantum walk is the following. Consider $\HH=\ell^2(\Ir)$, that is, a particle with no internal degree of freedom moving on a one dimensional lattice. As a further simplification we will assume $\mathcal{N}=\{0\}$, i.e. the particle is not moving at all. We choose a representation $\{\tilde U_x\}_{x\in\Ir}$ of $\Ir$ on some infinite dimensional dilation space $\DD$ and an operator $v_0\,:\,\Cx\rightarrow \ell^2(\Ir)$. The operator $\VV$ defines a translation invariant quantum walk if the normalization condition \eqref{Normalization}, which reads
\[
v_x^*(\idty_\KK\otimes \tilde U_x)v_0=\delta_{x0} ,\quad \forall x\in\Ir\, ,
\]
is fulfilled. By definition $v_x$ is zero if $x\neq 0$, hence, $v_0$ is an isometry. The action of the operator $\Wa$ corresponding to $\VV$ is
\[
 \Wa(\ketbra{x}{y})  = \braket{v_0}{\tilde U_{y-x}\vert v_0}\,\cdot\,\ketbra{x}{y}\, .
\]
Obviously, $\Wa$ satisfies the locality condition \eqref{Locality} with $\mathcal{N}=\{0\}$, it leaves diagonal elements invariant and off-diagonal elements are damped exponentially in the number of time steps. A possible choice for $\tilde U_x$ with continuous spectrum is $\DD=\ell_2(\Ir)$, any normalized vector $v_0\in \ell_2(\Ir)$ and the operators $\tilde U_x$ as shift by $x$ lattice sites.

General translation invariant quantum walks with momentum transfer are difficult to handle. We must assume in the following that there is no momentum transfer. An example where we drop this assumption can be found in section \ref{subsec:QwMK}.
\begin{ass}\label{ass:nomom}
There is no momentum transfer, i.e. $\tilde U_x=\idty$, $\forall x\in \Ir^s$. The isometry $ \VV$ and Kraus operators $ K_j$ are given by
\[
(\VV\psi)(p)=\sum_{ y\in \mathcal{N}}e^{ip\cdot y}v_y \psi(p)\,,\quad \psi \in\KK\,,
\]
\[
\braket{\phi}{(K_j\psi)(p)}=\sum_{y\in\mathcal{N}}e^{ip\cdot y} \braket{\phi\otimes e_j}{v_y \psi(p)}\,,\quad \phi,\,\psi \in\KK\,,
\]
where $e_j$ labels an orthonormal basis of $\DD$.
\end{ass} 

%% file: RCWaJMP.tex
\section{Asymptotic position by the perturbation method}\label{sec:AsymPosPertMeth}

From the beginning of quantum walk theory the question of asymptotic behavior of the position $Q(t)$ at large times $t$ has been one of the main themes. Early papers were to some extent misguided by the analogy with random walks and long combinatorial computations of matrix elements were done to evaluate just the special case of a Hadamard walk starting from the origin. The physicists in the community quickly brought to bear Fourier methods, and these emphasized the analogy not with classical random walks, but with the free particle under a continuous time Schr\"{o}dinger time evolution. In particular, this brought in dispersion relations $\omega(p)$, and group velocities $\nabla\omega(p)$ as the relevant quantities in the unitary case. Thus, for a general unitary walk and arbitrary initial state, the computation of the asymptotic distribution of $Q(t)/t$ became a straightforward evaluation of expectation values (see below). Initial studies on decoherent walks were often limited to a very special noise model, and have almost exclusively considered the first and second moments of position. While this is already good enough to distinguish ballistic from diffusive transport, it usually remained open how to compute the asymptotic distribution of $Q(t)/\sqrt t$ or, indeed, how to decide whether this quantity had a limit distribution.

We will therefore begin by showing how to focus on the entire distribution of $Q$ from the outset. This will establish the perturbation theory of the eigenvalue $1$ of the transition operator as the key tool in the further analysis. In that introductory section we ignore the control process in order to keep the notation simple.
We then look at the unitary case (\secref{unitary}), where the first order perturbation theory of a degenerate eigenvalue determines the group velocity operator and hence the ballistic scaling. This is then extended to more general processes, including externally controlled ones (\secref{ballistic}).
We then come to our main result, the asymptotic formulas for Markov controlled coined walks, in diffusive scaling (\secref{diffusive}). Finally, we consider the simplified case where the Markov process is of Bernoulli type (\secref{Bernoulli}), i.e. its transition rates are independent of previous time steps, which means the quantum operations are drawn independent and identically distributed in each time step.

A compact way to characterize a probability distribution of a real vector valued random variable $Q\in\Rl^s$ is in terms of its characteristic function
\begin{equation}\label{charf}
    C_Q(\lambda)=\left\langle e^{i\lambda\cdot Q}\right\rangle.
\end{equation}
Here the bracket denotes expectation, $\lambda$ is a real vector of the same dimension $s$ as $Q$, and the product in the exponent is the scalar product in $\Rl^s$. When $Q$ has a probability density, the characteristic function is just its Fourier transform. The derivatives of $C$ at the origin (if they happen to exist) are the moments of $Q$ (if they happen to exist). Very helpful for our purpose is that it is easy to express the characteristic function  for a scaled variable, say $\mu Q$, with a fixed factor $\mu$. Then we just have $C_{\mu Q}(\lambda)=C_Q(\mu\lambda)$. In our case $Q=Q(t)$ will be the position after $t$ time steps of a quantum walk $\Wa$, starting from some initial state $\rho_0$, i.e.,
\begin{equation}\label{charw}
    C_{Q(t)}(\lambda)=\tr \rho_0 \Wa^t\Bigl(e^{i\lambda\cdot Q}\Bigr).
\end{equation}
When the walk is controlled by an external Markov process, we also have to take the expectation of the right hand side with respect to the stationary distribution of the control process. Now we want to look at a scaled position distribution. For example, in ballistic scaling we get
\begin{equation}\label{charwballist}
    C_{Q(t)/t}(\lambda)=\tr \rho_0 \Wa^t\Bigl(e^{i\lambda\cdot Q/t}\Bigr).
\end{equation}
In the limit $t\to\infty$ the unitary operator $\exp(i\lambda\cdot Q/t)$ approaches the identity, which is invariant under $\Wa$. On the other hand, we act on this nearly invariant element with a high power of $\Wa$. The basic idea of our asymptotic evaluation is to look instead at the high powers of a slightly modified operator $\Wt_\veps$, defined by
\begin{equation}\label{Wtwiddle}
    \Wt_\veps(X)=\Wa\Bigl(Xe^{i\veps\lambda\cdot Q}\Bigr)e^{-i\veps\lambda\cdot Q}
\end{equation}
where $\veps=1/t$, or $\veps=1/\sqrt t$ for diffusive scaling, is now a small parameter.
$\Wt_\veps$ is similar to $\Wa$ via the invertible linear (but quite non-positive) operator $X\mapsto\exp(i\veps\lambda\cdot Q)$.
This means that
\begin{equation}\label{Wpowers}
    \Wa^t(e^{i\veps\lambda\cdot Q})=\Wt^t_\veps(\idty)e^{i\veps\lambda\cdot Q}
\end{equation}
for all $t$. For initial states $\rho_0$ supported on a finite region the exponential factor on the right hand side will be close to the identity in the scalings we consider, hence can be neglected when substituting this expression into (\ref{charwballist}).

A crucial observation is that although $\exp(i\veps\lambda\cdot Q)$ is not a translation invariant operator, $\Wt_\veps$, like $\Wa$ commutes with translations. This is because if we apply a translation by $x\in\Ir^s$ to (\ref{Wtwiddle}), we get two phase factors $\exp(\pm i\veps\lambda\cdot x)$, which cancel. In particular, if we apply $\Wt_\veps$ to a translation invariant operator like $\idty$ we again get a translation invariant operator. However, we will consider $\Wa$ and $\Wt_\veps$ as maps on the space of translation invariant operators, i.e. as multiplication operators in momentum space, and apply perturbation theory in those subspaces. Now, the restriction of $\Wt_\veps$ to the translation invariant operators is {\it not} similar to the restriction of $\Wa$. This is because the similarity transform on the whole space, i.e. right multiplication by $\exp(\pm i\veps\lambda\cdot Q)$ does not respect translation invariance. Hence, the eigenvalue of $\Wt_\veps$ which goes to $1$ as $\veps\to0$ may differ from $1$. Indeed, the perturbation theory of this eigenvalue is the core of our method.

An operator $A\in\BB(\HH)$ is translation invariant iff it is a function of momentum, i.e. $(\Fou A\psi)(p)=A(p)(\Fou\psi)(p)$. Under the action of $\Wa$ or $\Wt_\veps$ such an $A$ is transformed into a function $A'(p)$. In general, $A'(p)$ might depend on values $A(p')$ at points $p'\neq p$. This is where \assref{nomom} comes into play: if there are no momentum transfers, i.e. each Kraus operator is itself a function of $p$, then $(\Wa A)(p)=\Wa(p)A(p)$ for a suitable operator $\Wa(p)$ on $\BB(\KK)$ depending on $p$. Indeed, from
\begin{equation}\label{Walkp}
    (\Wa A)(p)=\sum_\alpha K_\alpha(p)^*A(p)K_\alpha(p)
\end{equation}
and (\ref{Wtwiddle}) we get
\begin{equation}\label{Wtp}
    (\Wt_\veps A)(p)=\sum_\alpha K_\alpha(p)^*A(p)K_\alpha(p+\lambda \veps).
\end{equation}
The finite range condition makes each $K_\alpha$ a trigonometric polynomial, so this operator is an analytic function of $\veps$ and we can apply perturbation theory \cite{Kato}.
Let us denote the Jordan decomposition of the above operator \cite[Sect.I\S5.4]{Kato} by
\begin{equation}\label{jordan}
    (\Wt_\veps A)(p)=\sum_i \Bigl(\mu_i(\veps)\Pb_i(\veps)+\Db_i(\veps)\Bigr)(A(p)),
\end{equation}
where $\Pb_i\Pb_j=\delta_{ij}\Pb_i$ are the eigenprojections, and $\Db_i$ eigennilpotent operators with $\Db_i\Pb_j=\Pb_j\Db_i=\delta_{ij}\Db_i$ for the eigenvalue $\mu_i(\veps)$.
Let us assume for the moment that the eigenvalue $\mu_0(0)=1$ of the unperturbed operator is simple, so that the unique eigenvector is $\idty$, and the other eigenvalues satisfy $|\mu_i(0)|<1$ for $i\neq 0$. (This will be the standing assumption in \secref{ballistic} but not in \secref{unitary}). Then $\Db_0=0$, and we get
\begin{equation}\label{xxx}
    \Wt_\veps^t (\idty)= \mu_0(\veps)^t \Pb_0(\veps)(\idty)+\ldots\,,
\end{equation}
where the dots stand for terms with $i\neq0$. Since $\idty$ is in the eigenspace for $\mu_0$, and the $\veps$-dependent operators can be chosen analytic, all these contributions vanish as $\veps\to0$, and $\Pb_0(\veps)\idty\to\idty$. Note that the crucial point here is the assumption $|\mu_i(0)|<1$ for $i\neq 0$, cf. the discussion in Sect. \ref{sec:HighOrdWOBall}. So everything depends on the eigenvalue term $\mu_0(\veps)^t$. For {\bf ballistic scaling}, i.e. $\veps=1/t$ and $\mu_0(\veps)=1+iv\cdot\lambda\veps+\Order(\veps^2)$, for some vector $v\in\Rl^s$, we find
\begin{equation}\label{v-example}
    \mu_0(\veps)^t=\Bigl(1+\frac{iv\cdot\lambda}t+\Order(t^{-2})\Bigr)^t
        \quad\longrightarrow e^{iv\cdot\lambda}.
\end{equation}
Hence, the probability distribution of $Q/t$ converges to a point measure at a deterministic (but possibly $p$-dependent) velocity $v$. A more detailed study of this case is given in \secref{unitary} and \secref{ballistic}, where we also include an external control process. Consider on the other hand the special case $v=0$, then the leading order contribution to $\mu_0$ is of the form $\mu_0(\veps)=1-\frac12 \veps^2\lambda\cdot M\cdot\lambda+\Order(\veps^3)$ for some matrix $M$. Then, in {\bf diffusive scaling} $\veps=t^{-1/2}$ we get
\begin{equation}\label{diff-example}
    \mu_0(\veps)^t=\Bigl(1-\frac{\lambda\cdot M\cdot\lambda}{2t}+\Order(t^{-3/2})\Bigr)^t
        \quad\longrightarrow e^{-\frac12\lambda \cdot M\cdot\lambda}.
\end{equation}
This is the characteristic function of a Gaussian with covariance matrix $M$. Hence, the asymptotic distribution of $Q/\sqrt t$ is Gaussian. A closer analysis of this case, again including control processes, will be given in \secref{diffusive}.

\section{Unitary quantum walks}
\subsection{Ballistic order}
\label{sec:unitary}
The unitary case has been the subject of various papers \cite{Ambainis,Carteret,Konno,Bressler,Ambainis2,Kempe,Grimmett}. Here we allow a general walk, as given by a unitary matrix $W(p)$ \eqref{Wp}. As described in the previous section, we need to study the perturbation theory of the eigenvalue $1$ of the family of operators $\Wt_\veps$:
\begin{equation}\label{unitary}
   \Wt_\veps(A)=W(p)^*AW(p+\lambda\veps),
\end{equation}
where we treat $p$ as a fixed parameter. Clearly, $\idty$ is an eigenvector of eigenvalue $1$ for this operator, but the eigenvalue $1$ is actually quite degenerate: any operator $X$ commuting with $W(p)$ is also in this eigenspace. The eigenspace is thus at least $\dim\KK$-dimensional, but if some of the eigenvalues of $W(p)$ are degenerate at $p$, the degeneracy can be even higher. So let
\begin{equation}\label{specWp}
    W(p)=\sum_ke^{i\omega_k(p)}P_k(p)
\end{equation}
be the spectral resolution of $W(p)$. In the case of degeneracies, i.e. when several $\omega_k(p)$ coincide, this is not unique. But the perturbation theory of the one-parameter analytic family $\veps\mapsto W(p+\veps\lambda)$ tells us that we can choose the operators $P_k(p+\veps\lambda)$ such that in the neighborhood of $\veps=0$ they and the corresponding eigenvalues $\omega(p+\veps\lambda)$ depend analytically on $\veps$. In the sequel we assume such a choice has been made in \eqref{specWp}.

Not surprisingly, this leads to an analytic perturbation expression for $\Wt_\veps$. Indeed, let $\{P_k\}_{k=1}^d$ and $\{R_\ell\}_{\ell=1}^d$ be families of orthogonal projections in a Hilbert space $\KK$. Then we can consider the operators $E_{k\ell}(X)=P_kXR_\ell$ on $\BB(\KK)$ one easily checks that each $E_{k\ell}$ is hermitian with respect to the Hilbert Schmidt scalar product $\braket YX=\tr(Y^*X)$, and the $E_{k\ell}$ are themselves a family of orthogonal projections. Now setting $P_k=P_k(p)$ and $R_k=P_k(p+\veps\lambda)$ we find
\begin{equation}\label{specWt}
    \Wt_\veps=\sum_{k\ell}e^{i(\omega_\ell(p+\veps\lambda)-\omega_k(p))}E_{k\ell},
\end{equation}
which is clearly a spectral decomposition in terms of eigenvalues and eigenprojections, which are all analytic in $\veps$. Therefore, the expression
\begin{equation}\label{specWtIdty}
    \Wt_\veps^t(\idty)=\sum_{kl}e^{it(\omega_\ell(p+\veps\lambda)-\omega_k(p))}P_k(p)P_\ell(p+\veps\lambda)
\end{equation}
is correct to all orders. As $\veps\to0$, we have $P_\ell(p+\veps\lambda)\to P_\ell(p)$, and since $P_k(p)$ and $P_\ell(p)$ are orthogonal, only the terms with $k=\ell$ survive in the limit. Moreover, with ballistic scaling $\veps=1/t$ the exponent converges to the derivative of $\omega_\ell$, which exists even at degeneracy points, because we have chosen \eqref{specWp} analytically. Hence
\begin{equation}\label{limWt1}
    \lim_{t\to\infty}\Wt_{1/t}^t(\idty)
     =\sum_{k}\left.\exp\Bigl(i \frac {d\omega_k(p+\veps\lambda)}{d\veps}\Bigr)\right|_{\veps=0} P_k .
\end{equation}
Note, however, that the choice of projections $P_k$ at a degenerate point may well depend on the direction $\lambda$, in which $p$ is varied. Moreover, the derivatives of $W(p)$, compressed to the degenerate eigenspace need not commute, so there is no analytic choice of branches $\omega_k$. We call a point $p$ a {\bf regular} momentum for $W(p)$ if in \eqref{specWp} we can choose $P_k$ and $\omega_k$ to be analytic functions of the vector near $p$. Of course, when the eigenvalues of $W(p)$ are all non-degenerate, $p$ is regular, and this will almost always be the case.

For all regular $p$, we can write \eqref{limWt1} as the exponential of the operator $i\lambda\cdot V(p)$, where $V$ is the $p$-dependent
vector operator with components
\begin{equation}\label{groupV}
    V_\alpha(p)=\sum_k\frac {\partial\omega_k(p)}{\partial p_\alpha}P_k.
\end{equation}
This is the operator of {\bf group velocity}. Note that, for all regular $p$, all of its components commute with $W(p)$, since they are linear combinations of eigenprojections of $W(p)$. Therefore, the components of $V$ are jointly measurable in the sense of standard quantum mechanics. In any initial state $\rho$ this gives a probability measure on velocity space. This measure is the asymptotic position distribution starting from $\rho$.

\begin{thm}Let $p\mapsto W(p)$ be the defining unitary of a quantum walk on $\Rl^s$. Let $Q(t)$ denote the position observable, evolved for $t$ steps.
Suppose that almost all $p$ are regular, so the group velocity operator $V(p)$ is defined almost everywhere. Then
$$ \lim_{t\to\infty}\frac{Q(t)}t=V$$
in the sense that for all bounded continuous functions $f:\Rl^s\to\Cx$ going to zero at infinity we have
the weak operator limit of $f(Q(t)/t$, evaluated in the functional calculus, is $f(V)$. This means, for any initial state $\rho$ the distribution of $Q(t)/t$ goes weakly (in the sense of probability measures) to the distribution of $V$ in $\rho$.
\end{thm}

Let us conclude this section with two short examples of unitary quantum walks without temporal disorder. We start with the generic one dimensional quantum walk with a two dimensional coin space, given by a shift operation $S$ that moves the walker left or right depending on the internal degree of freedom and some $SU(2)$ coin operation.
In that case the walk operator in momentum space can be written in the familiar coin and shift decomposition
\begin{align}\label{eqn:1dwalk}
  W(p) = C\cdot S=\left(
               \begin{array}{cc}
                 \cos(\alpha)\mathrm e^{i\beta} & \sin(\alpha)\mathrm e^{i\gamma} \\
                 -\sin(\alpha)\mathrm e^{-i\gamma} & \cos(\alpha)\mathrm e^{-i\beta} \\
               \end{array}
             \right)
   \left(
               \begin{array}{cc}
                 \mathrm e^{i p} & 0 \\
                 0 & \mathrm e^{-i p} \\
               \end{array}
             \right)
\end{align}
with some unitary matrix $C$. By diagonalizing this matrix we get the $p$ dependent phases $\omega_\pm$ of the eigenvalues of the walk operator and the group velocities $\frac{\partial\omega_\pm}{\partial p}$

\begin{align*}
  w_\pm(p) =& \pm\arccos(\cos(p+\alpha)\cos(\beta)) \\ \frac{\partial\omega_\pm}{\partial p} =& \pm\frac{\cos(\beta)\sin(p+\alpha)}{\sqrt{1-\cos(\beta)^2\sin(p+\alpha)^2}}\;.
\end{align*}

In the case of the well known Hadamard walk, where $C$ is chosen to be the Hadamard matrix, the dispersion relations are given by $\omega_\pm(p)=\frac{\pi}{2}\pm\arccos\left(\frac{\sin{p}}{\sqrt{2}}\right)$, where the $\frac{\pi}{2}$ compensates a factor of $i$ due to our parametrization of the $2\times 2$ unitaries. In Fig. \ref{gra:Had:disp} a plot of the $p$ dependence of both eigenvalue branches is shown. From these dispersion relations we can compute the group velocities $\partial_{p}\omega_\pm(p)$, which are given in Fig. \ref{gra:Had:vel}. Looking at the asymptotic $t^{-1}$ scaled position distribution $\mathrm{P}(x)$ of the initial state $\psi(p) =
\genfrac(){0pt}{}{1}{0}$,
we see that at the points $\pm\frac{1}{\sqrt{2}}$ the density diverges (Fig. \ref{gra:Had:asym}. Comparing this with the graph of the group velocities we find that this corresponds to points, where the group velocity has an extremal point. We call these points {\bf caustics} and they are precisely the points causing the peaks in the asymptotic distributions.

\begin{figure}[htb]
  \centering
  \subfigure[Dispersion relation \label{gra:Had:disp}]{\includegraphics[width=0.32\textwidth]{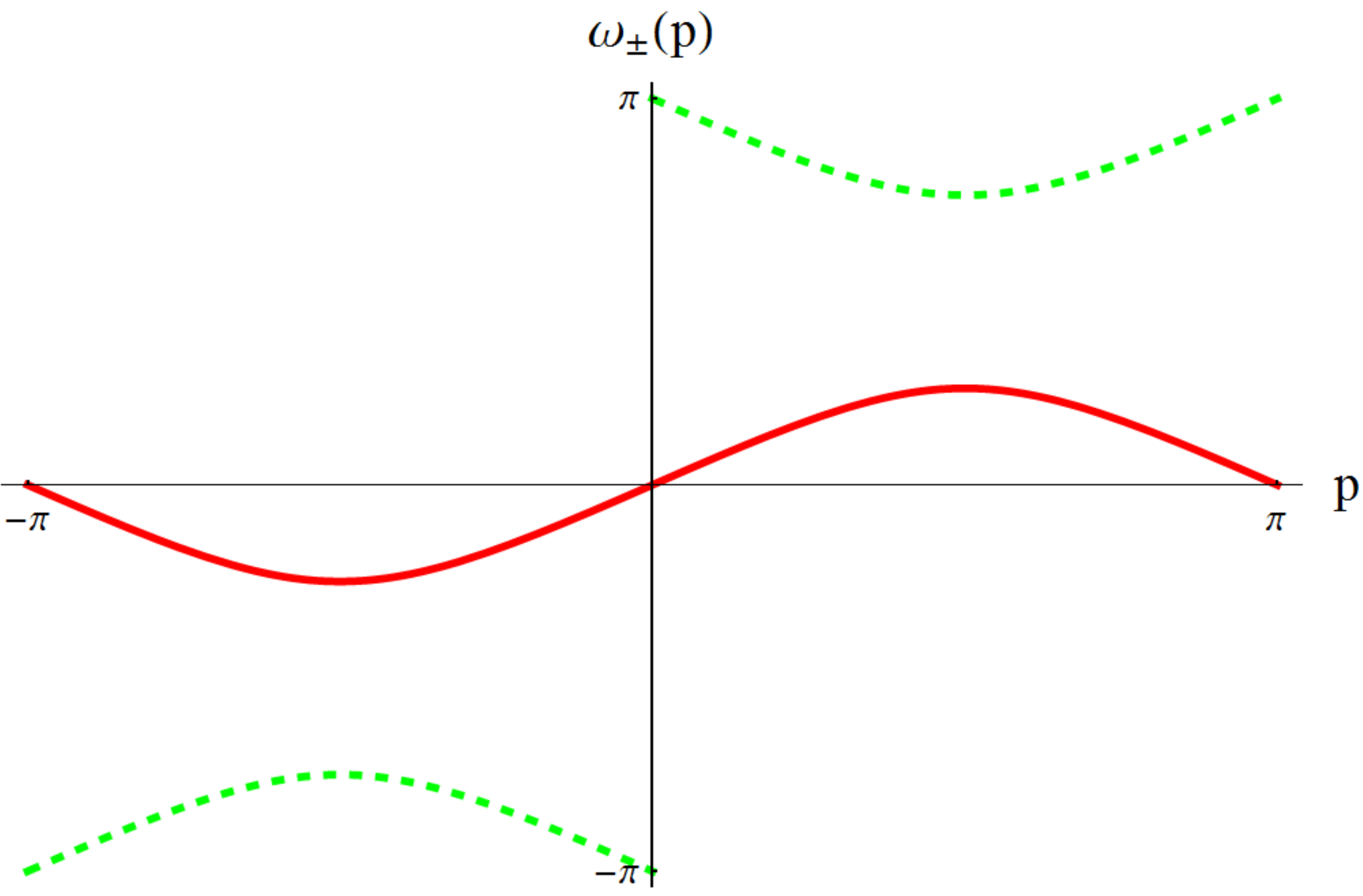}}
   \subfigure[Velocities \label{gra:Had:vel}]{\includegraphics[width=0.32\textwidth]{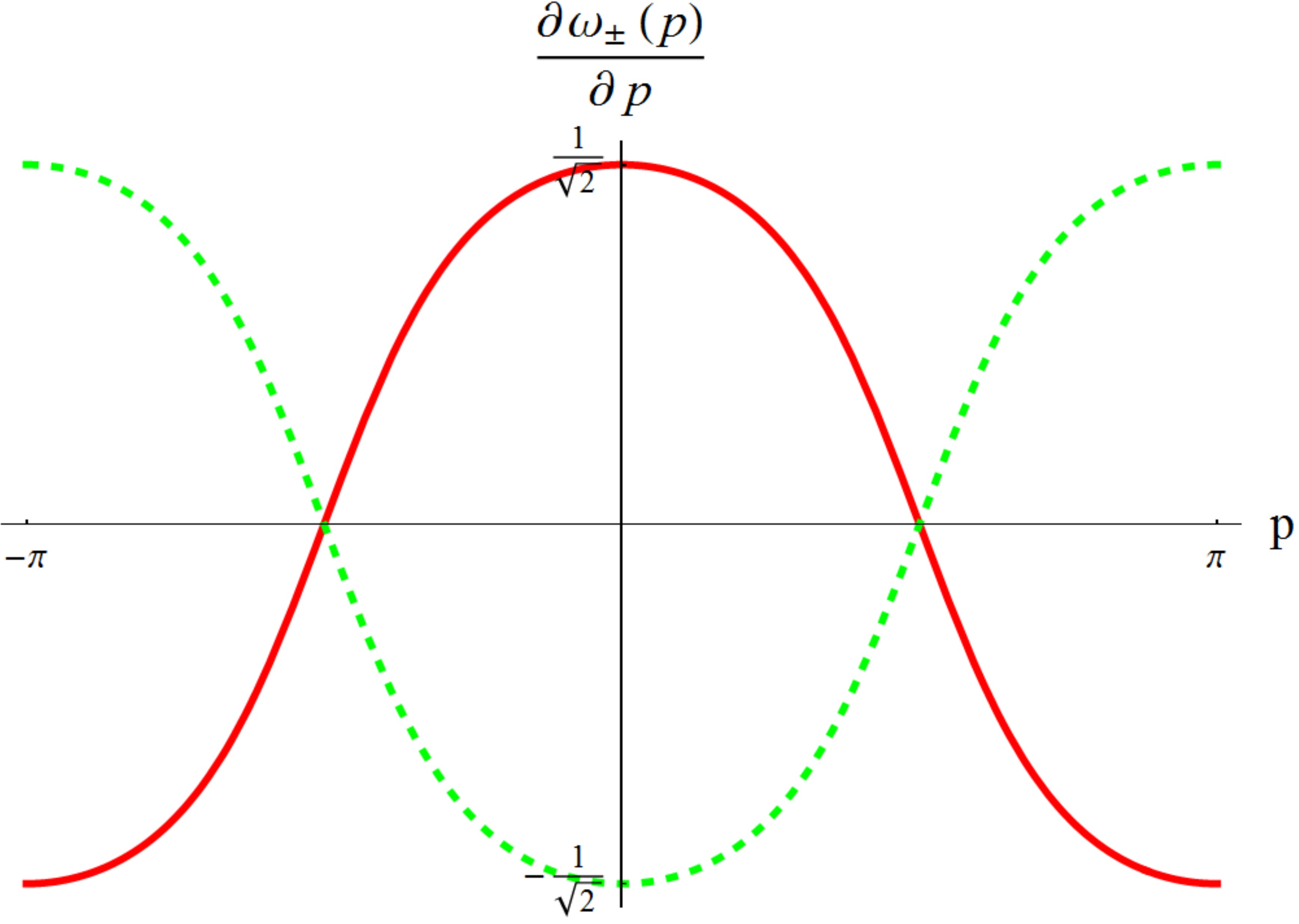}}
    \subfigure[Asymptotic distribution for initial state {$\psi(p)= \genfrac(){0pt}{}{1}{0}$ } \label{gra:Had:asym}]{\includegraphics[width=0.32\textwidth]{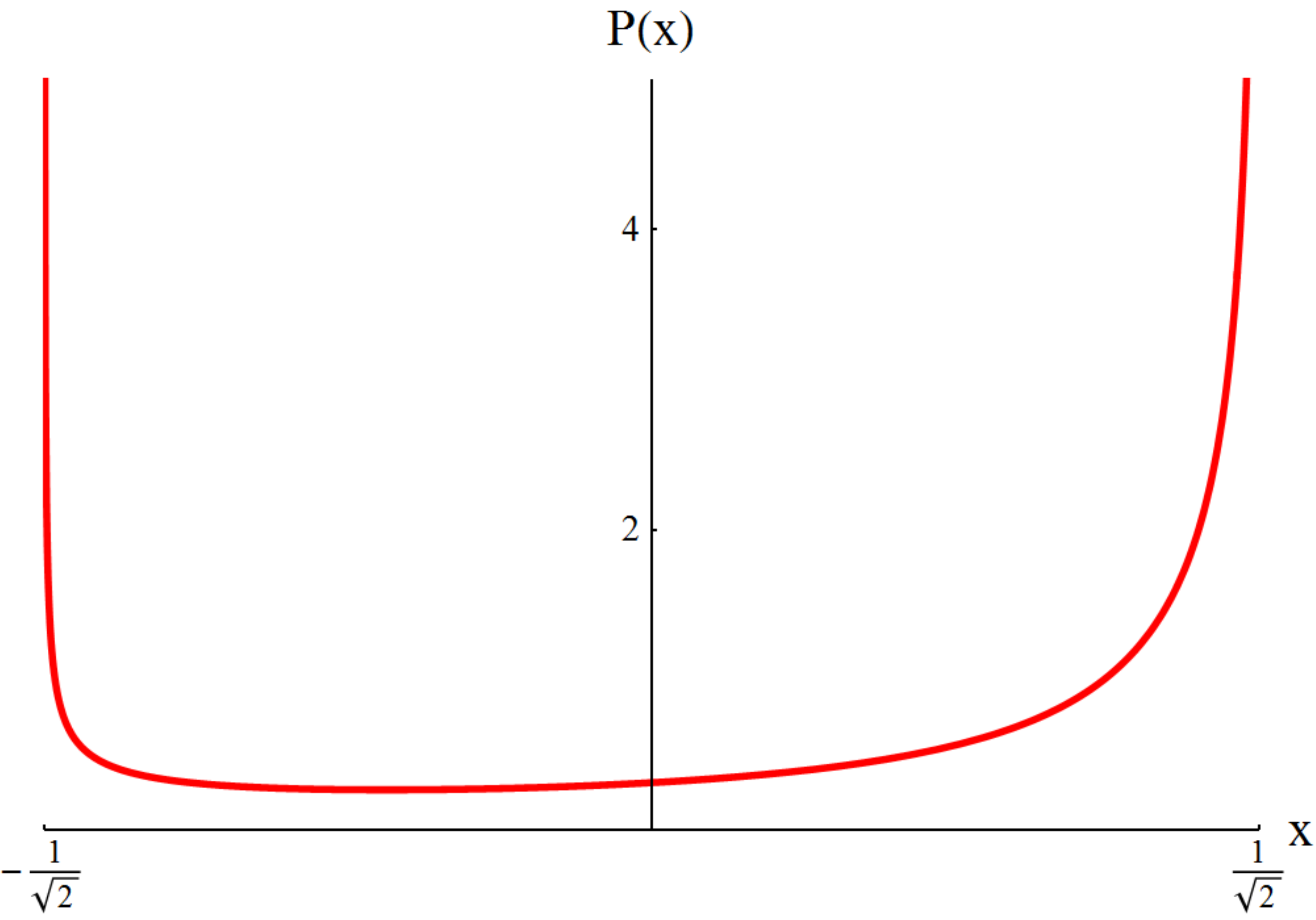}}
    \caption{(color online) For the quantum walk with Hadamard coin, the figures show (a) the dispersion relations $\omega_\pm$, (b) the velocities $v_\pm$, and (c) the asymptotic probability distribution. The initial state is chosen to be $\psi(p)=\genfrac(){0pt}{}{1}{0}$. The extremal points of the functions $v_\pm(p)$, hence the inflection points of $\omega_\pm$, are responsible for the peaks in (c) at $\pm 1/\sqrt{2}$.}
\end{figure}

Such caustics are also a generic behavior of quantum walks on $\Ir^2$. In accordance with the one dimensional case we define a caustic a bit more formally as a point in momentum space where the Jacobi matrix of the group velocity or equivalently the Hessian of $\omega(p_1,p_2)$ is singular. This in turn implies that the density of the $t^{-1}$ scaled asymptotic position distribution will diverge at such points, cf. Corollary \ref{locinv}.

As an example we will study a quantum walk on $\Ir^2$ with a two dimensional coin space. The shifting in the $p_1$ and $p_2$ direction is done separately and in between two unitary coin operations are alternated. The overall walk operator is therefore given by
\begin{align}\label{2Dwalk}
  W(p_1,p_2) = U_2\cdot S_2\cdot U_1\cdot S_1=U_2\left(
               \begin{array}{cc}
                 e^{i p_2} & 0 \\
                  0 & e^{-i p_2} \\
               \end{array}
             \right)\cdot U_1\cdot
   \left(
               \begin{array}{cc}
                 e^{i p_1} & 0 \\
                 0 & e^{-i p_1} \\
               \end{array}
             \right)\;.
\end{align}
The resulting dispersion relations are given by
\begin{eqnarray}
  \omega_{\pm}(p_1,p_2) &=& \pm\arccos(\cos(p_1 + p_2 + \theta_1+\theta_2) \cos(\phi_1) \cos(\phi_2) \nonumber\\
&&- \cos(p_1 - p_2 - \chi_1 + \chi_2) \sin(\phi_1) \sin(\phi_2)\;,
\end{eqnarray}
where greek lower case letters correspond the parametrization of a $2\times 2$ unitary matrix as given in \eqref{eqn:1dwalk} and the subscripts refer to the unitaries $U_i$. From this relation we could also calculate the velocities and the points in the asymptotic distribution where caustics can be observed. We will do this for an explicit example, see Fig. \ref{Fig:2DW}, where the parameters of the two unitaries are chosen to be $\theta_1 = -\theta_2 = \frac{\pi}{3}$, $\chi_1=-\chi_2=\frac{\pi}{4}$,
$\phi_1 = \frac{\pi}{4}$ and $\phi_2 = \frac{\pi}{3}$

The resulting band structure is depicted in Fig. \ref{gra:2DW:disp}. In principle, the bands could cross or intersect each other at a line or a single point. In this example, however, we can observe the generic behavior of an avoided crossing of the two branches of the dispersion relation. The red lines in the graph indicate momentum values for which we find caustics, i.e. where the determinant of the Jacobi matrix of $\omega(p_1,p_2)$ vanishes.
\begin{figure}[htb]
  \centering
  \subfigure[Dispersion relation \label{gra:2DW:disp}]{\includegraphics[width=0.37\textwidth]{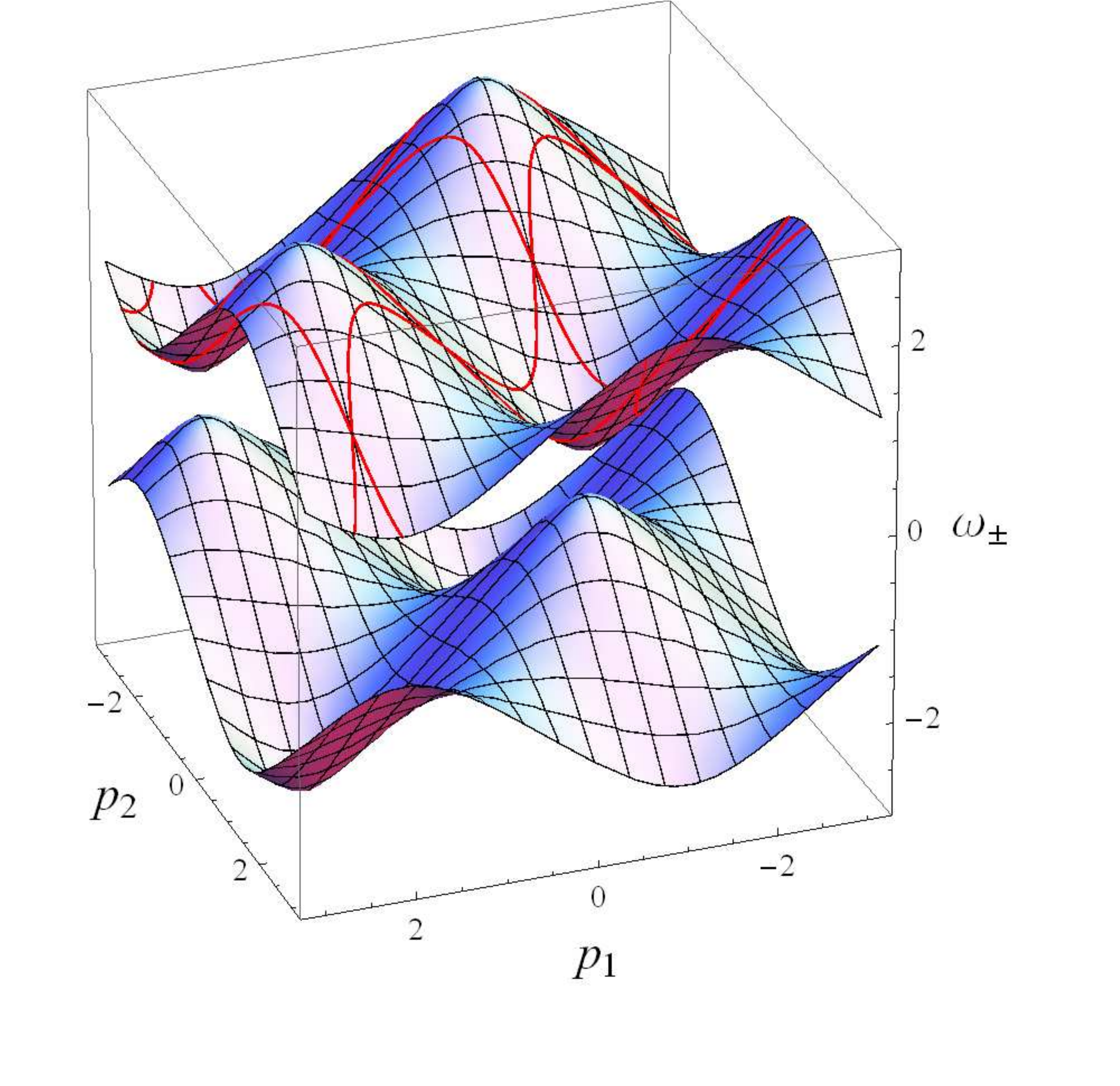}}
   \subfigure[Probability density of $v_+$ \label{gra:2DW:vel}]{\includegraphics[width=0.24\textwidth]{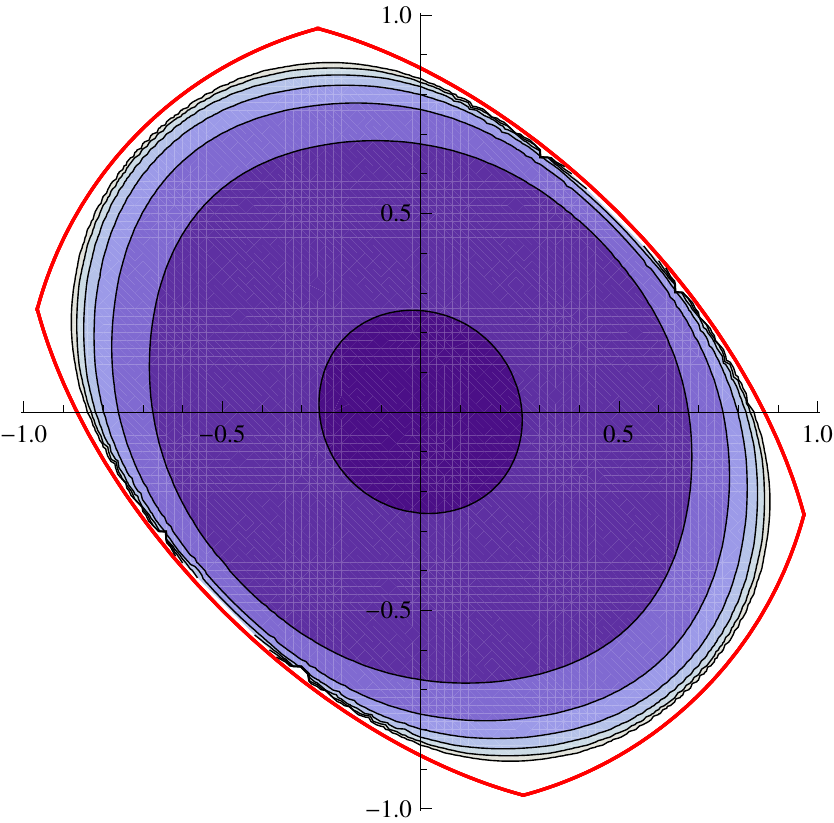}}
   \subfigure[Probability density of $v_+$ \label{gra:2DW:hess}]{\includegraphics[width=0.37\textwidth]{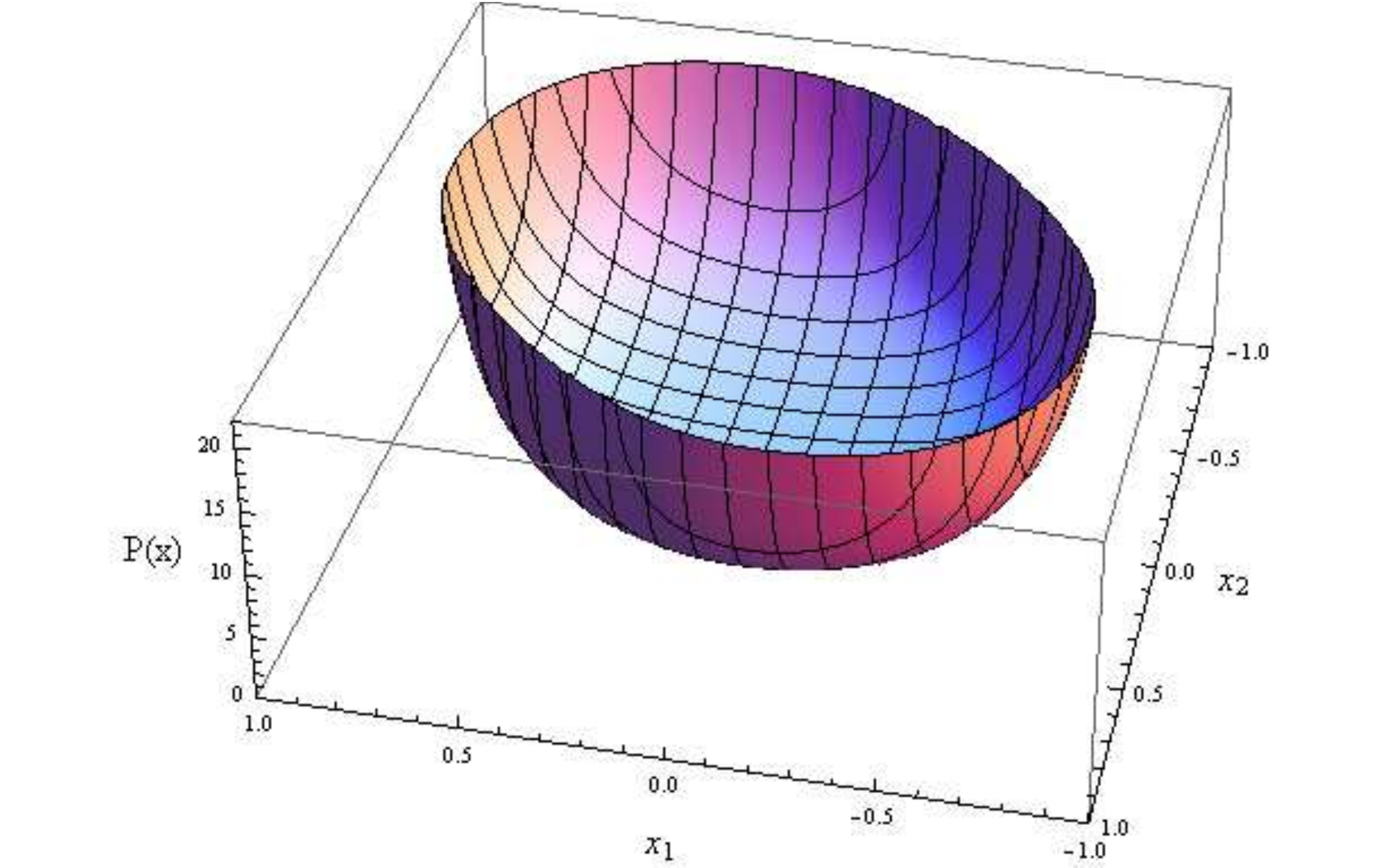}}
    \caption{(color online) For the two-dimensional walk \eqref{2Dwalk} the plots show (a) the dispersion relation $\omega_\pm$, (b) a contour plot, and (c) a 3D plot of the contribution of the $\omega_+$-branch to the asymptotic position density, for a particle starting at the origin. The red lines in (a) are curves of vanishing curvature. At these points the velocity density (i.e., the inverse of the Jacobian of the transformation $p\mapsto v_+(p)$) diverges. This produces the enclosing red line in (b), and infinitely high values in (c). For more complicated walks such lines also appear in the interior of the velocity region.}
\label{Fig:2DW}
\end{figure}
The second graph (Fig. \ref{gra:2DW:vel}) shows a contour plot of the possible pairs of velocities $(\partial_{p_1} \omega, \partial_{p_2} \omega)$ in the upper branch of the dispersion relation $\omega_+$.
 The red line at the border of the possible velocities, where the value of the determinant drops to zero, corresponds to the red lines in the $p_1,p_2$-dependent graph in Fig. \ref{gra:2DW:disp} of the dispersion relations. As in the one dimensional case, these are the points where the probability density of the asymptotic position in ballistic scaling will exhibit peaks. In the last graph (Fig. \ref{gra:2DW:hess}) the value of $|\det (\partial_p^2 \omega_+ )|^{-1}$ is plotted dependent on the velocities.

\subsection{Higher orders}\label{sec:higherUnitary}
We have evaluated asymptotic characteristic functions by employing a perturbation series. It is therefore natural to ask, whether one cannot use the higher terms in the series to get better approximations to the position density for large but not infinite $t$. Of course, this is possible, but one has to be careful in the interpretation of the results. The main problem is that the partial sums of an expansion of the characteristic function in powers of $1/t$ is {\it not} a characteristic function of any probability measure. Indeed, in the expansion we use we typically get combinations of $\lambda/t$, so the higher orders of the expansion will be polynomials in $\lambda$ times an oscillating factor. This is clearly not integrable, so the inverse Fourier transform to get the probability density is ill-defined, and gives, at best, a rather singular distribution. However, if we only look for the expectations of sufficiently smooth functions of velocity, say $f(Q(t)/t)$, we get the integral of the characteristic function with the Fourier transform of $f$ which decays rapidly enough to absorb all polynomial factors. Thus for a fixed smooth test function the expansion makes sense. Another way to look at this is to multiply the expanded characteristic function with a suitable cutoff-function (enforcing sufficient decay in $\lambda$) before transforming back to velocity space, resulting in a smoothed out probability density. Then ``removing the cutoff'' and the series expansion do not commute, and the choice of cutoff is effectively the choice of a smooth family of test functions.

We can go back directly to Eq.~\eqref{specWtIdty}, which is correct to all orders. That is, the characteristic function of the position distribution in ballistic scaling at time $t$ is
\begin{eqnarray}\label{Ctunitaryplus}
    C_t(\lambda)&=&\tr\rho \Wt_{1/t}^t(\idty)e^{i\lambda\cdot Q/t}\nonumber\\
      &=& \int dp\ \sum_{k\ell} e^{\textstyle it(\omega_\ell(p+\lambda/t)-\omega_k(p))}
              \tr\Bigl(\rho(p+\lambda/t,p)P_k(p)P_\ell(p+\lambda/t)\Bigr)
\end{eqnarray}
Here, as in the previous subsection $\rho(p_1,p_2)$ denotes the integral kernel of the initial density, and the trace is over the internal degrees of freedom.
In leading order we could neglect the shift by $\lambda/t$ in $P_\ell$, so only terms with $k=\ell$ remain. Looking now at the first order term resulting from the expansion of $P_\ell(p+\lambda/t)$ and $k\neq\ell$ we find an oscillatory integral with a regular integrand and rapidly oscillating exponential $\exp it(\omega_\ell(p)-i\omega_k(p))$. Assuming that this phase is not constant on sets of positive measure (as a function of $p$), we conclude that the integral goes to zero, so that with the factor $1/t$ from the expansion of $P_\ell$ such terms are $\ordert1$ and can be neglected. Of course, in higher order corrections one will have to extract the leading orders of the oscillatory integral by a stationary phase analysis.

For the expansion to first order we need the expansion of the dispersion relation to second order:
\begin{equation}\label{unifirst}
    \omega_k(p+\lambda/t)=\omega_k(p)+\frac\lambda t\cdot v_k(p)
             +\frac{1}{2t^2}\omega''(p,\lambda)+\ordert2,
\end{equation}
where $\omega''$ is a quadratic from in $\lambda$ containing the Hessian of the branch $\omega_k$ of the dispersion relation. This approximation eliminates one of the most prominent features of the finite $t$ probability distributions, namely their rapid oscillations. Indeed, from inspecting such distributions it is clear that these oscillations have a frequency of order $1/t$ in ballistic scaling, i.e., they are really at the scale of the underlying lattice. This is reflected in the exact formula  \eqref{Ctunitaryplus} by the fact that all expressions are $2\pi$-periodic in $\lambda/t$. So the Fourier transform of $C_t$ is a sum of $\delta$-functions at the lattice points. The approximation \eqref{unifirst} destroys this feature, resulting in a rather smooth function inside the allowed region of velocities.
The Hessian can be determined from standard second order perturbation theory of $W(p)$, but we found it convenient to eliminate it by partial integration of the $1/t$-term containing $\omega''$. Since the integrand is periodic in $p$, this gives no boundary terms, and the resulting differentiations of the trace can be combined with the other Taylor expansions to give the first order term
\begin{equation}\label{cfunitaryfirst}
    C_t(\lambda)=C_\infty(\lambda)+\frac1{2t}\int dp\sum_ke^{i\lambda\cdot v_k(p)}\tr
                   \Bigl((Q\cdot\lambda)\rho+\rho(Q\cdot\lambda))(p,p) + i[P_k(p),\lambda\cdot P'_k(p)]\Bigr)
                   +\ordert1.
\end{equation}
In the example of a Hadamard walk one can substitute the variable $v(p)$ and find a Bessel function for the first order term. However, due to the factor $\lambda$ this is not integrable, and the inverse Fourier transform has to be taken in the distributional sense, more specifically under the integration with smooth test functions whose support in velocity space stays away from the caustics. This results in the expression
\begin{equation}\label{nextHadamard}
    \rho_t(u)=\frac{1}{\pi(1-u)\sqrt{1-2 u^2}}+\frac1t\ \frac{u}{\pi \sqrt{1-2 u^2}^{\,3}}
              +\ordert1
              \quad\mbox{ for}\ \abs u<\frac1{\sqrt2}
\end{equation}
for the Hadamard walk starting at the origin in state $\psi_0=(1,0)$. This is shown for $t=10$ in Fig.~\ref{fig:hadamardnext}.
\begin{figure}[htb]
\includegraphics[width=0.5\textwidth]{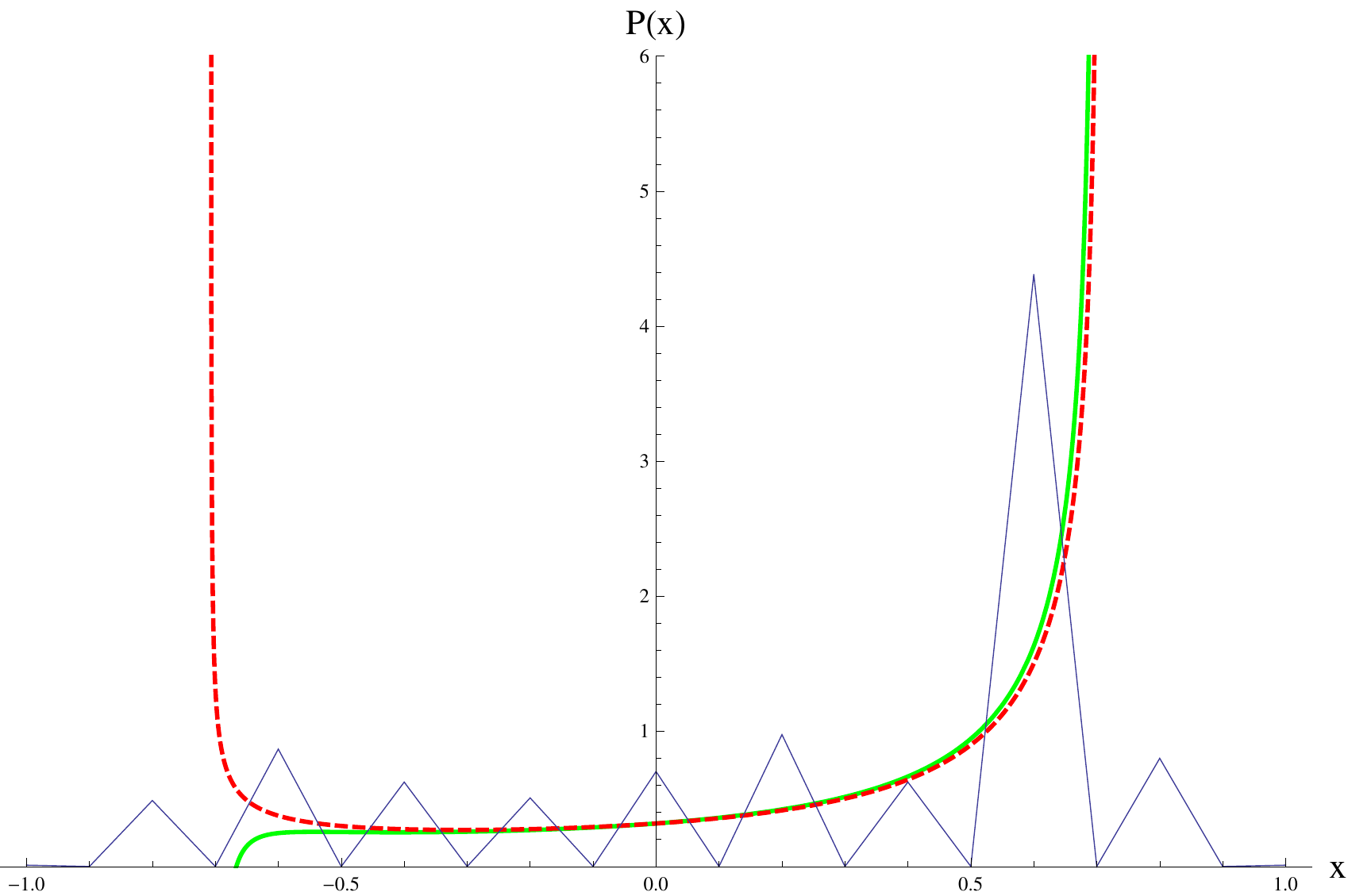}
    \caption{(color online) Correction of order $1/t$ to the asymptotic position distribution for the Hadamard walk from the initial state $\psi_0=(1,0)$ located at the origin. The polygon connects the exact values for $n=10$. The asymptotic distribution (red/dashed) overestimates the left peak and underestimates the right peak. The correction after \eqref{nextHadamard} for the same $t$ is shown by the green/solid curve.}
\label{fig:hadamardnext}
\end{figure} 

%% file: RCWdJMP.tex
\section{Time decoherent walks}
\subsection{Ballistic determinism for decoherent processes}
\label{sec:ballistic}
We now extend the analysis to decoherent models, as described in \secref{geneSys}. We will continue to work under \assref{nomom}, i.e. excluding momentum transfer. We can thus still apply the theory outlined in \secref{AsymPosPertMeth}, and find that the asymptotic distribution of position is given, as in \eqref{Wtp} by an operator depending on momentum. This is also true if we include a control space with a Markov chain dynamics. The observable space we work on is hence the tensor product of $\L\infty(\Gamma)$, i.e. the control space variables with the quantum observables $\BB(\KK)$. We will fix the value of $p$, and everything will depend on this parameter, but this dependence will, for the moment be the only thing left of the translation degree of freedom.

As before, we will consider the observables on $A\in \L\infty\otimes\BB(\KK)$ as measurable functions $\gamma\mapsto A(\gamma)\in\BB(\KK)$. The two ingredients of each time step will then be on the one hand the $\gamma$-dependent quantum operation
\begin{equation}\label{Vtpg}
    (\Va A)(\gamma)=\Va_\gamma(A(\gamma))
       =\sum_\alpha K^*_{\gamma,\alpha}(p)A(\gamma)K_{\gamma,\alpha}(p),
\end{equation}
where the Kraus operators are normalized such that the sum over $\alpha$ with $A(\gamma)=\idty$ gives $\idty$. The second step is the update of the control parameters by the Markovian evolution $\Ma$, i.e.
\begin{equation}\label{markov}
  (\Ma A)(\gamma)=\int \ma_\gamma(d\eta)A(\eta).
\end{equation}
The full evolution is then given by $\Wa=\Va\Ma$ or, written out more explicitly:
\begin{equation}\label{WVMlong}
    (\Wa A)(\gamma)=\int \ma_\gamma(d\eta)\Va_\gamma(A(\eta)),
\end{equation}
with the momentum dependence implicit in the dependence of $\Va_\gamma$ according to \eqref{Vtpg}.

Obviously, this form contains the unitary case, e.g., when $\Va_\gamma$ is the same walk unitary for all $\gamma$, and the control process is irrelevant for the walking particle. We are now interested in the opposite end, i.e. the case of generic randomness. Therefore we will assume the following

\begin{ass}
\label{ass:chaos}
For almost all $p$, the eigenvalue $\mu_0=1$ of $\Wa$ as an operator on $\L\infty(\Gamma)\otimes\BB(\KK)$ is simple and isolated and for $i\neq 0$ the eigenvalues satisfy $|\mu_i|<1$. Moreover, there exists an invariant faithful state for $\Wa$, i.e. a density operator $\invrho$ with nonzero eigenvalues on $\KK$ such that $\tr\invrho\Va_\gamma(A)=\tr\invrho A$ for all $\gamma$ and all $A$.
\end{ass}

An example of commuting Kraus operators, where the non-degeneracy condition of the eigenvalue $1$ is violated given in \ref{ex:commKraus}. 
To see that this assumption is not overly restrictive we consider the following scenario.
\begin{prop}
\label{prob:simpev}
\assref{chaos} is valid if the quantum walk $\Wa$ has the following properties
\begin{enumerate}
\item $\Gamma$ is finite, so the transition probabilities can be written as a matrix $\ma_\gamma(\eta)$.
    For some power of this matrix all entries are strictly positive.
\item There is a density operator $\invrho$ with nonzero eigenvalues on $\KK$ such that $\tr\invrho\Va_\gamma(A)=\tr\invrho A$ for all $\gamma$ and all $A$.
\item For almost all $p$, and some $r\in\Nl$ the set of operators
   $K_{\gamma_1\alpha_1}(p)\cdots K_{\gamma_r\alpha_r}(p)$ is irreducible, i.e. only multiples of $\idty$ commute with all of them, and its linear span contains the identity.
\end{enumerate}
\end{prop}

The last condition seems tricky to check, but it is generically satisfied. In fact, the operator products usually span the whole space of matrices for relatively small $r$.

\begin{proof}
Since $\L\infty(\Gamma)\otimes\BB(\KK)$ is now finite dimensional, any simple eigenvalue is isolated, so we only have to prove that $1$ is a simple eigenvalue, and the only one on the unit circle.
The transitivity assumption on the transition probabilities guarantees that there is a unique, strictly positive, invariant probability distribution $\invm$ on $\Gamma$. Hence we have an invariant state for  $\Wa$, namely $\overline\ma\otimes\overline\rho$, written out as
$$  (\overline\ma\otimes\overline\rho)(A)=\sum_\gamma\invm_\gamma\ \overline\rho\bigl(A(\gamma)\bigr).$$

We first show the simplicity of $1$, and since linearly independent eigenvectors of $\Wa$ will be independent eigenvectors of $\Wa^{n}$ we may do this for some power of $\Wa$. We choose some multiple of $r$, say $rn$, chosen sufficiently large so that all entries of the $nr$-step transition matrix are positive. Then condition 3 in the Proposition is also satisfied for $rn$, because the identity lies in the span of the Kraus operators for $\Wa^r$. Hence, for this step we can simplify the assumptions to $m_\gamma(\eta)>0$ for all $\gamma,\eta$ and the irreducibility of $\{K_{\gamma\alpha}\}$.

The basic technique for the proof is the decomposition
\begin{equation}\label{2posdec}
    \Wa(A^*A)-\Wa(A)^*\Wa(A)= \Va\Bigl(\Ma(A^*A)-\Ma(A)^*\Ma(A)\Bigr)+ \Va(B^*B)-\Va(B)^*\Va(B),
\end{equation}
where $B=\Ma(A)$. Both terms are positive by the ``2-positivity inequality'' \cite{Paulsen} for channels. But when we evaluate for $\Wa(A)=A$ in an invariant state of $\Wa$ the left hand side becomes zero. This will provide a lot of information from the vanishing of sums of positive terms on the right. Explicitly, we get
$$ \Bigl(\Ma(A^*A)-\Ma(A)^*\Ma(A)\Bigr)(\gamma)
   =\frac12\sum_{\eta,\chi}\ma_\gamma(\eta)\ma_\gamma(\chi)
     \Bigl(A(\eta)-A(\chi)\Bigr)^*\Bigl(A(\eta)-A(\chi)\Bigr),
$$
where we used that $\sum_{\eta} m_\gamma(\eta)=1$ for all $\gamma$. Applying the invariant state $\invm\otimes\invrho$, and using the invariance condition for each $\Va_\gamma$, we find
$$ \sum_{\gamma,\eta,\chi}\invm_\gamma\ma_\gamma(\eta)\ma_\gamma(\chi)
     \tr\invrho\bigl(A(\eta)-A(\chi)\bigr)^*\bigl(A(\eta)-A(\chi)\bigr)=0.$$
Since the probabilities $\ma$ and $\invm$ are all strictly positive, we find that each summand vanishes. Because $\invrho$ has no zero eigenvalue, this also implies that $\bigl(A(\eta)-A(\chi)\bigr)^*\bigl(A(\eta)-A(\chi)\bigr)=0$, and hence
$A(\eta)=A(\chi)$ for all $\eta,\chi$. Hence we can set $A(\gamma)=A$ for all $\gamma$.
This makes the first term in \eqref{2posdec} vanish for every $A$, and leads to $B=\Ma(A)=A$ in the second term. Moreover, from \eqref{WVMlong} we see that $\Wa(A)=A$ just means that $\Va_\gamma(A)=A$ for all $\gamma$.

Now consider, for each one of the operators $\Va_\gamma$, which has Kraus operators $K_{\gamma\alpha}$,
the expression
\begin{equation}\label{commutatorsum}
    \sum_\alpha [A,K_{\gamma\alpha}]^*[A,K_{\gamma\alpha}]
      =\Va_\gamma\bigl(A^*A\bigr)-\Va_\gamma(A)^*A-A^*\Va_\gamma(A)+A^*A
\end{equation}
If $A$ is invariant, this reduces to $\Va_\gamma\bigl(A^*A\bigr)-A^*A$, which is clearly zero under the
common invariant state $\invrho$. Hence the expectation of the positive terms on the left hand side must vanish also, and since $\invrho$ has no zero eigenvalues each $[A,K_{\gamma\alpha}]=0$ for all $\alpha$ and $\gamma$. By assumption this implies that $A$ is a multiple of the identity.

It remains to be shown that there are no further eigenvalues on the unit circle. Suppose to the contrary that $\Wa(X)=\omega X\neq0$ for some $\omega\neq1$ with $\abs\omega=1$. Then the operator
$\Wa(X^*X)-\Wa(X)^*\Wa(X)=\Wa(X^*X)-X^*X$, which is positive by the $2$-positivity, has vanishing expectation in the faithful invariant state, which implies that it is zero. But then $X^*X$ is a fixed point of $\Wa$, and we have already seen that this implies that $X^*X$ is a multiple of the identity. Since it cannot be zero, we can normalize $X$ so that it becomes unitary. But $2$-positivity also implies that if $\Wa(X^*X)-\Wa(X)^*\Wa(X)=0$ we must also have $\Wa(Y^*X)-\Wa(Y)^*\Wa(X)=0$, for all $Y$. Otherwise, the inequality could not be valid for linear combinations of $Y$ and $X$. But then, by induction on $n$, we find that $\Wa(X^n)=\omega^nX^n$. In other words, all powers of $\omega$ are eigenvalues. Since the dimension of the space is finite, this means that $\omega$ must be a root of unity, say $\omega^n=1$. But then $X$ is also an eigenvector of $\Wa^n$ with eigenvalue $1$. Since our arguments for the simplicity of $1$ also apply to powers of $\Wa$, this implies $X=\idty$, and a contradiction to $\omega\neq1$.
\end{proof}

Now we can apply the ideas of \secref{unitary} to extract the ballistic spreading of the position distribution. The only difference is that we can now use non-degenerate perturbation theory, so everything is much simpler. Consider again some parameter $\lambda$ as the argument in the characteristic function, and the operator $\Wt_\veps$, acting on $A\in\L\infty(\Gamma)\otimes\BB(\HH)$. We define operators $\Vt_{\gamma,\veps}$ via
\begin{equation}\label{Wtgamma}
    (\Wt_\veps A)(\gamma)
       =\int \ma_\gamma(d\eta)\
         \sum_\alpha K_{\gamma\alpha}(p)^*A(\eta)K_{\gamma\alpha}(p+\veps\lambda)
       =\int \ma_\gamma(d\eta)\ \Vt_{\gamma,\veps}(A(\eta))\, .
\end{equation}
Let $\Wt_\veps A_\veps=\mu_\veps A_\veps$ denote the branch of eigenvectors and eigenvalues with $\mu_0=1$.
Since the eigenvector is only determined up to a factor we are free to choose a normalization so that
\begin{equation}\label{normalizeA}
    \invm\otimes\invrho(A_\veps)=1.
\end{equation}
Expanding these objects to second order,  we get
\begin{eqnarray}
  \Vt_{\gamma,\veps} &=& \Va_\gamma+\veps\Vt'_\gamma+\frac{\veps^2}2 \Vt_\gamma''+\Order(\veps^3)
                         \nonumber\\
  \mu_\veps &=& 1 + \veps\mu' +\frac{\veps^2}2 \mu''+\Order(\veps^3)     \label{eveps}\\
  A_\veps &=& \idty+\veps A' + \frac{\veps^2}2 A'' + \Order(\veps^3)  \, . \nonumber
\end{eqnarray}
For now we will only use the first order, in which the eigenvalue equation reads
\begin{eqnarray}\label{Ev1nondeg}
    \Wa(A')(\gamma)-A'(\gamma)
      &=&\mu'\idty -\int \ma_\gamma(d\eta)\
         \Vt'_\gamma(\idty)=\mu'\idty - \Vt'_\gamma(\idty)
         \nonumber\\
      &=&\mu'\idty - \sum_\alpha K_{\gamma\alpha}(p)^* \left.
         \frac d{d\veps}K_{\gamma\alpha}(p+\veps\lambda)\right\vert_{\veps=0},
\end{eqnarray}
Where we evaluated the sum over $\eta$, on which the integrand does not depend. Taking the expectation with respect to the invariant state $\invm\otimes\invrho$ makes the left hand side zero, and leaves an explicit equation for $\mu'$, namely
\begin{eqnarray}\label{vdecoh}
    \mu'&=&i\lamv(p) \quad\mbox{with the vector}\nonumber\\
    v(p)&=&  -i\int \invm(d\gamma) \sum_\alpha \tr\invrho K_{\gamma\alpha}(p)^* \nabla K_{\gamma\alpha}(p).
\end{eqnarray}
This expression is real, because for all $\gamma$ and all $X\in\BB(\KK)$, the function
\[
p\mapsto\sum_\alpha\tr\invrho K_{\gamma\alpha}(p)^* X K_{\gamma\alpha}(p)=\tr\invrho X
\]
is constant, and hence has zero gradient. Summarizing we get the following result.

\begin{prop}
\label{asymdistr} Given \assref{nomom} and \assref{chaos}, the asymptotic distribution of $Q(t)/t$ is determined by the distribution of $v(p)$ as defined in \eqref{vdecoh}.
\end{prop}
\begin{proof}
According to the preceding reasoning the characteristic function of the random variable $Q(t)/t$ converges for $t\to \infty $ to $\int  dp\,e^{i\lambda \cdot v(p)}\tr \rho (p)$ and the asymptotic distribution of $Q(t)/t$ is given by inverse Fourier transform of this function.
\end{proof}
Two features of this formula are remarkable. Firstly, $v(p)$ is a scalar, i.e. a multiple of the identity with respect to the internal degrees of freedom from $\KK$. Therefore, in contrast to the unitary case discussed in \secref{unitary}, the asymptotic distribution is independent of the initial state of the coin. Secondly, the transition probabilities of the control process only enter through the invariant probability distribution $\invm$. This means that we get the same result for a Bernoulli process, in which we take independent coins in successive steps with probability distribution $\invm$. The control space therefore just contributes another index to the Kraus decomposition of a one-step channel, cf. Sect. \ref{sec:Bernoulli}. We will see, however, that the transition probabilities are not irrelevant for next order, see Sect. \ref{ex:misslinks}.

The following corollary gives a closed formula for the asymptotic distribution of $Q(t)/t$ if the function $v(p)$ is such that there is a decomposition of momentum space into a finite number of open sets on which the function $v(p)$ is invertible and a boundary which coincides with the set of caustic points of the quantum walk.

\begin{cor}
\label{locinv}
Let the characteristic function of the asymptotic distribution of $Q(t)/t$ be given by
\[
\lim\limits_{t\to\infty}C_{Q(t)/t}(\lambda )
    =\int\limits_{(-\pi,\pi]^s}\kern-10pt dp\ e^{i\lambda \cdot v(p)}\rho (p)\,
\]
and suppose there exist finitely many disjoint open sets $R_j$, $j\in J$ on which $v(p)$ is invertible and $R=\dot \cup R_j$ is the set of all points such that the determinant of the Jacobi matrix $\partial v(p)/\partial p$ is non-zero. Then, the asymptotic probability distribution of $Q(t)/t$ is given by
\begin{equation}
\label{asymdistreqn}
\mathrm{P}(x)= \sum_{j\in J} \left|\det \left(\frac{\partial v_j}{\partial p}\right)\right|^{-1}\kern-10pt(x)
    \ \rho ({v_j}^{-1}(x))\quad ,\, x\in\Rl^s\, ,
\end{equation}
with $v_j$ the restriction of $v$ to $R_j$ and the set of caustic point coincides with the complement of $R$ in $(-\pi,\pi]^s$.
\end{cor}
\begin{proof}
The asymptotic position distribution $\mathrm{P}(x)$ can be obtained from the characteristic function $C(\lambda)$ by applying the inverse Fourier transform
\[
\mathrm{P}(x)=\frac{1}{(2\pi)^s}\int\limits_{\Rl^s}d\lambda \kern-10pt\int\limits_{(-\pi,\pi]^s}\kern-10pt dp\
     e^{i\lambda \cdot(v(p)-x)}\rho (p)\,.
\]
Using the integral representation of the Dirac distribution
$\delta (x) = \frac{1}{(2\pi)^s}\int_{\Rl^s} d\lambda\ e^{i\lambda\cdot x}$ we get
\[
\mathrm{P}(x)= \sum_{j\in J} \int\limits_{R_j}\!\!dp\ \delta (v(p)-x)\rho (p)\,.
\]
Now, formula \eqref{asymdistreqn} follows from the fact that $v(p)$ is invertible on the $R_j$ and a substitution of the integration variable $p$ by $v$. The complement of $R$ in $(-\pi,\pi]^s$ is given by the set of points where the Jacobi matrix $\partial v(p)/\partial p$, i.e. the Hessian of the dispersion relation, is singular. This is the set of caustic points, and the asymptotic probability distribution diverges at the corresponding points.
\end{proof}

\subsection{Markov controlled walks in diffusive scaling}
\label{sec:diffusive}

We now go on to study the asymptotic position distribution on the $\sqrt t$ scale. If $Q(t)/t\to v$ goes to a sharp value the quantity of interest will be $(Q(t)-vt)/\sqrt t$, and we can hope that this has a well-defined limiting distribution. Indeed for this to make sense we need that the velocity is sharply defined. Otherwise, we would just see the limiting distribution of $v$ scaled to larger and larger variance. Therefore, we need to restrict to a setting, in which $v(p)$ is automatically constant. So from now on we assume that we have a Markov controlled random coin as described in the introduction:

\begin{ass}\label{ass:rcoin} For every $\gamma$, the operation $\Va_\gamma$ is unitarily implemented.
That is, for each $\gamma$ there is only one Kraus operator $K_\gamma=K_{\gamma1}$, which is a unitary walk operator, so that $\Va_\gamma(A)=K_\gamma^*AK_\gamma$.
\end{ass}

An example violating this assumption is given in section \ref{subs:NonUnitKO}.
The following proposition says that this indeed implies the required constancy of the ballistic velocity.

\begin{prop}\label{prob:velirred} \assref{chaos} and \assref{rcoin} together imply that the function $v(p)$, defined in \eqref{vdecoh}, is independent of $p$. Moreover,
\begin{eqnarray}\label{oneinv}
   \invrho&=&\frac1{\dim\KK}\ \idty  \quad\mbox{and}\\
         v&=&\int \ma(d\gamma)\, \frac{\ind K_\gamma}{\dim\KK}\ .
\label{meanindex}
\end{eqnarray}
\end{prop}

\begin{proof} Obviously, the density operator $\invrho\propto\idty$ is an invariant state for each $\Va_\gamma$. By definition of the index, $\det K_\gamma(p)=c\exp(ip\cdot\ind K_\gamma)$. Hence
\begin{eqnarray}
    i\lambda\cdot\ind K_\gamma
      &=&\frac1{\det K_\gamma(p)}\left.\frac d{d\veps}\det K_\gamma(p+\veps\lambda)\right|_{\veps=0}
       =\left.\frac d{d\veps}\det\bigl(K_\gamma(p)^*K_\gamma(p+\veps\lambda)\bigr)\right|_{\veps=0}
         \nonumber\\
      &=&\tr\bigl(K_\gamma(p)^*\frac d{d\veps}K_\gamma(p+\veps\lambda)\bigr)
       =\dim\KK \tr\invrho\Vt'_\gamma(\idty),   \nonumber
\end{eqnarray}
and the formula follows by dividing this equation by $\dim\KK$ and summing with respect to $\invm$.
\end{proof}
Note that for unitary quantum walks $W$ the index $\ind W $ is given by the trace of the group velocity operator $V$, see Sect. \ref{sec:unitary}. Therefore, if we naively apply \eqref{meanindex} to the unitary case, although Assumption \ref{ass:chaos} is violated, we find $v(p)=\tr V(p)$. Of course, the asymptotic distribution of a unitary quantum walk $W$ is not determined by $v(p)$ but by the whole operator $V(p)$.

When the ballistic order is completely defined by a deterministic velocity, i.e. $Q(t)/t\to v\idty$, then we can subtract the ballistic motion and look at how the probability distribution develops around it. In other words, we look at the deviation operator
\begin{equation}\label{deviate}
    D(t)=\frac1{\sqrt t}\bigl(Q(t)-vt\bigr).
\end{equation}
We will compute the limit of the characteristic function of $D$, i.e.
\begin{eqnarray}\label{chardeviate}
 \lim_{t\to\infty}C_{D(t)}(\lambda)
      &=& \lim_{t\to\infty} e^{-i\lamv\,\sqrt t}C_{Q(t)}\left(\frac\lambda{\sqrt t}\right)
           \nonumber\\
      &=&\lim_{t\to\infty}e^{-i\lamv\,\sqrt t}\int\! dp \int \invm(d\gamma)\tr\left( \rho(p)
           \Wt^t_{1/\sqrt t}(\idty)(p,\gamma)\right)
           \nonumber\\
      &=&\int\!\!dp \tr\bigr(\rho(p)\bigl)\lim_{t\to\infty}e^{-i\lamv/\sqrt t}
           \Bigl(1+i\frac{\lamv}{\sqrt t}+\frac{\mu''}{2t}+\Order(t^{-3/2})\Bigr)^t
           \nonumber\\
      &=& \int\!\!dp \tr\bigr(\rho(p)\bigl) \exp\left(\frac{1}2(\mu''+(\lamv)^2)\right)
\end{eqnarray}
Here, at the first equality we substituted \eqref{deviate}, at the second introduced $\Wt$ from
\eqref{Wtwiddle} in the form \eqref{Wtgamma}, and at the third introduced the perturbation expansion
\eqref{eveps} with $\veps=1/\sqrt t$. At the last equality we used again the asymptotic formula $\lim_{t\to\infty}(1+x/t+\ordert1)^t=\exp(x)$, which follows immediately from the Taylor expansion of the logarithm. In the application above there is a cancelation of large phases, so the formula must be applied with care, in the form
\begin{eqnarray}\label{LemmaExp}
    \lim_{t\to\infty}&&e^{-ia\sqrt t}\left(1+\frac{ia}{\sqrt t}+ \frac bt+ \Order(t^{-3/2})\right)^t
      =\lim_{t\to\infty}\left(e^{-ia/\sqrt t}\bigl(1+\frac{ia}{\sqrt t}+ \frac bt+\Order(t^{-3/2})\bigr)\right)^t \nonumber\\
      &&=\lim_{t\to\infty}\left(1+\frac bt+\frac{a^2}{2t}+\order(t^{-1})\right)^t
        = \exp\left(b+\frac{a^2}2\right).
\end{eqnarray}

Note that, in contrast to $\mu'=iv\cdot\lambda$, the second order perturbation coefficient $\mu''$ will depend on $p$. Moreover, since the perturbation is proportional to $\veps\lambda$, the second order perturbation coefficient $\mu''$ must be homogeneous quadratic in $\lambda$. Hence, according to \eqref{chardeviate}, the limiting distribution of $D(t)$ is a mixture of Gaussians with $p$-dependent covariance matrix $s(p)$ such that
\begin{equation}\label{st}
    \lambda\cdot s(p)\cdot\lambda=\sum_{jk}s_{jk}(p)\lambda_j\lambda_k=-\mu''(p)-(\lamv)^2.
\end{equation}
For the moment, we fix $\lambda$, and concentrate on computing the expression \eqref{st} for any set of given data and verifying its positivity.

Evaluating the eigenvalue equation $\Wt_\veps A_\veps=\mu_\veps A_\veps$ to first and second order in $\veps$, we now get
\begin{eqnarray}\label{Ev12nondeg}
    \Wa(A')(\gamma)-A'(\gamma)
       &=& i\lamv\idty - \Vt_\gamma'(\idty)
              \label{Ev1nondega}\\
    \Wa(A'')(\gamma)-A''(\gamma)
       &=&\mu''\idty +2\mu'A'(\gamma)-\int \ma_\gamma(d\eta)
            \bigl(\Vt_\gamma''(\idty)+2\Vt'_\gamma(A'(\eta))\bigr)
       \label{Ev2nondeg}
\end{eqnarray}
The first line is just a repetition of \eqref{Ev1nondeg}. From the second line we extract $\mu''$ by taking the expectation with respect to the invariant state, using also that by the convention \eqref{normalizeA} the second term on the right has zero expectation:
\begin{equation}\label{mupp}
    \mu''=\int \invm(d\gamma)\ \tr\bigl(\invrho\Vt_\gamma''(\idty)\bigr)
          +2\iint \invm(d\gamma)\ma_\gamma(d\eta)\
             \tr\bigl(\invrho\Vt'_\gamma(A'(\eta))\bigr)
\end{equation}
In this equation, the first order perturbation $A'$ of the eigenvector must be extracted from \eqref{Ev1nondega}. Indeed, this is uniquely possible: By \assref{chaos} the simple eigenvalue $1$ of $\Wa$ is isolated, so the rank of $\Wa-\id$ is exactly one less than maximal. The kernel is explicitly known: $\Wa-\id$ annihilates precisely the multiples of the identity and maps onto the elements with vanishing expectation under $\invm\otimes\invrho$. By definition of $v$, the right hand side hence lies in the range of $\Wa-\id$, so there is a unique solution $A'$ satisfying the normalization condition \eqref{normalizeA}. Note that since the solution is unique, and $\Vt_\gamma'$ is skew hermitian, so is $A'$. Moreover, since the right hand side of \eqref{Ev1nondega} is linear in $\lambda$, so is $A'$. We stress that Equations \eqref{Ev12nondeg}-\eqref{mupp} are also valid if Assumption \ref{ass:rcoin} is violated and for ballistic scaling $\veps=1/t$ with non-constant $v(p)$ they can also be used for a finer analysis of the large time behavior, incorporating second order perturbation theory. From now on we will adopt \assref{rcoin}, so that $\Vt_\gamma'(\idty)=K_\gamma^*K_\gamma'$. Since each $K_\gamma$ is unitary, we can use this to express the derivatives of $K_\gamma$ by $A'$:
\begin{equation}\label{elimKg}
    K_\gamma'=K_\gamma(A'(\gamma)-W(A')(\gamma)+i\lamv\idty)
\end{equation}

We come back to the determination of $\mu''$ from \eqref{mupp}. We can eliminate the second derivative in the first term by differentiating \eqref{vdecoh} with respect to $p$, or more precisely, by setting $p=p+\veps\lambda$ and differentiating with respect to $\veps$. This gives
\begin{eqnarray}
0&=&-i\int \invm(d\gamma)\,\frac d{d\veps}
             \tr\bigl(\invrho K_\gamma^*(p+\veps\lambda)K_\gamma'(p+\veps\lambda)\bigr) \nonumber\\
    &=&-i\int \invm(d\gamma)\,\tr\bigl(\invrho K_\gamma^{*\prime}K_\gamma'\bigr)+
    -i\int \invm(d\gamma)\,\tr\bigl(\invrho\Vt_\gamma''(\idty)\bigr)
\end{eqnarray}
and hence, using \eqref{elimKg} and the unitarity of $K_\gamma$:
\begin{equation}\label{Mvtpp}
    \int \invm(d\gamma)\ \tr\bigl(\invrho\Vt_\gamma''(\idty)\bigr)
    =-\int \invm(d\gamma)\,\tr\bigl(\invrho \left|A'(\gamma)-W(A')(\gamma)+i\lamv\idty\right|^2\bigr),
\end{equation}
where, for an operator $X$, we use the abbreviation $|X|^2=X^*X$. The second term in \eqref{mupp} can also be simplified by eliminating the derivative in $\Vt_\gamma'(X)=K_\gamma^*XK_\gamma'$ via \eqref{elimKg}. We get
$$  \int \ma_\gamma(d\eta)\ K_\gamma^*A'(\eta)K_\gamma'
      =W(A')(\gamma)\bigl(A'(\gamma)-W(A')(\gamma)+i\lamv\idty\bigr).$$
Note that in \eqref{mupp} we need the expectation of this expression in the invariant state $\invm\otimes\invrho$, just as \eqref{Mvtpp} is such an expectation. Bringing together the various terms of \eqref{mupp}, and using the skew hermiticity of $A'$, $W(A)'$, and $i\lamv$, we find for \eqref{st}:
\begin{eqnarray}\label{st1}
  -\mu''-(\lamv)^2&=&
     =\invm\otimes\invrho\Bigl(|A'-W(A')+i\lamv\idty|^2-2 W(A')(A'-W(A')+i\lamv\idty)\Bigr)
        -(\lamv)^2  \nonumber\\
     &=&\invm\otimes\invrho\Bigl( W(A')^2-(A')^2+[A',W(A')]-2i\lamv A'\Bigr)  \nonumber\\
     &=&\invm\otimes\invrho\Bigl(\abs{A'}^2-\abs{W(A')}^2\Bigr).
\end{eqnarray}
Since $\invm\otimes\invrho$ is invariant under $W$, we can also write the first term as the expectation of
$W(|A'|^2)$, so that this expression is non-negative by virtue of the Cauchy-Schwarz inequality for channels \cite{Paulsen}. With $A'$ a linear function of $\lambda$, the above expression becomes a quadratic form in $\lambda$, as claimed in \eqref{st}.

Intuitively, one would expect that the diffusion constant becomes very large if the considered quantum walk differs only little from a coherent quantum walk with ballistic behavior. For example, if the coin operations of a decoherent quantum walk would be all very similar or if the Markov process is such that it prefers one of the coins, we would expect reminiscences of a coherent quantum walk even for large times. This effect can be seen in the examples in Sect. \ref{ex:misslinks} and \ref{ex:1DQWEqDist}, where we derive an explicit formula for the variance $s(p)$, which diverges in the coherent limit of the the considered quantum walks. Similarly, if the Markov process converges to a deterministic Markov chain, i.e. the transition matrix approaches a permutation matrix, the diffusive order will diverge. This is because we could consider a full cycle of the permutation as one step of another quantum walk, which is also unitary, leading to ballistic transport. If on the other hand the transition rates of the Markov process are very small, i.e. there are long subsequences of coherent evolution, the overall evolution up to time step $t$ would be very similar to the average of a number of coherent quantum walks each evolved for $t$ time steps. If some of these coherent quantum walks show ballistic behavior we would again expect that the diffusion constant diverges.

\subsection{Higher orders without ballistic determinism}
\label{sec:HighOrdWOBall}
The general remarks about higher order expansions in \secref{higherUnitary} apply also in this case. In this section we focus on just the first (i.e., $\Order(1/t)$) correction to the ballistic scaling. When the ballistic velocity is independent of $p$, this is, in fact, best expressed by the diffusive scaling. However, almost all of the second order perturbation theory developed in the previous sections is independent of that \assref{rcoin}. Therefore, we have already done most of the work needed to get the first correction to ballistic scaling in the general case.

We adopt \assref{nomom} and \assref{chaos}, but not necessarily \assref{rcoin}. Let us neglect for a moment the Markov control, and consider the Jordan decomposition \eqref{jordan} for small $\veps$. The $t^{\rm th}$ power of this operator is built from the terms
\[\Bigl(\mu_i(\veps)\Pb_i(\veps)+\Db_i(\veps)\Bigr)^t
    =\Pb_i(\veps)\sum_{r=0}^{r_i}{t\choose r}\mu_i^{t-r}\Db_i(\veps)^r\]
Here $r_i$ is the order of nilpotency of $\Db_i(\veps)$, which is bounded by the algebraic multiplicity of $\mu_i(\veps)$. When we choose the label $i=0$ for the key eigenvalue (i.e., $\mu_0(0)=1$) we have $\abs{\mu_i}<1$, for all $i\neq0$ by \assref{chaos}. Hence expressions such as  $t^n\mu_i^{t-r}$ go to zero faster than any power of $t^{-n}$ ($n\in\Nl$). Consequently, in any expansion in such powers, all terms but the one with $i=0$ can be neglected. Moreover,  $\mu_0$ is simple, so $r_0=0$. Therefore, only the term $\Pb_0(\veps)\mu_0(\veps)^t$ remains. Since we were only interested in the leading order so far, it was enough to replace $\Pb_0(\veps)$ by $\Pb_0(0)$. However, for a systematic expansion we also have to expand the rank one projection $\Pb_0(\veps)$ in powers of $\veps$. This will give $\Pb_0(\veps)(X)=A(\veps)\tr(\invrho B(\veps)X)$, where $A(\veps)$ is the eigenvector of $\Wt_\veps(p)$ and $B(\veps)$ the corresponding eigenvector of the adjoint. Of course, eigenvectors are only defined up to a (possibly $\veps$-dependent) factor. We have already used the convention that $\tr\invrho A(\veps)=1$. This fixes the factor also for $B(\veps)$, since we must have $\Pb_0(\veps)^2=\Pb_0$, and hence $\tr\invrho A(\veps)B(\veps)=1$. Hence, expanding $B$ as in \eqref{eveps} we get $\tr\invrho A'=\tr\invrho B'=0$, so that for an expansion to order $t^{-1}$ we do not need $B'$. These ideas remain valid with Markov control, with the invariant state $\invm\otimes\invrho$ taking the role of $\invrho$. Then
\begin{equation}\label{p1first}
    \Pb_0(\idty)(\gamma)=\idty+\frac1t A'(\gamma) +\ordert1,
\end{equation}
where $A'$ is determined exactly as above. For the eigenvalue $\mu_0(\veps)^t=(1+\frac{\mu'}t+\frac{\mu''}{2t^2})^t$, we first expand the logarithm, giving $\exp(\mu'+\frac1{2t}(\mu''-(\mu')^2)+\ordert1$. With the $p$-dependent covariance matrix $s$ (see \eqref{st}) we can thus write
\begin{eqnarray}\label{mutfirst}
    \mu_0(\veps)^t&=&e^{\textstyle i\lambda\cdot v(p)-\frac1{2t}\lambda\cdot s(p)\cdot\lambda}+\ordert1\\
       &=&e^{\textstyle i\lambda\cdot v(p)}\Bigl(1-\frac1{2t}\lambda\cdot s(p)\cdot\lambda\Bigr)+\ordert1.\nonumber
\end{eqnarray}
From the point of view of a series expansion in inverse powers of $t$ these two forms are equivalent. However, the first form is preferable, because it is the characteristic function of a Gaussian distribution and hence has a probabilistic interpretation. This has to be taken with a grain of salt, however: The proof of positivity of $s(p)$ given in the previous section does depend on \assref{rcoin}. Indeed, we will see in an example (Sect.~\ref{subs:NonUnitKO}) that $s(p)$ may be complex, and only represents a probability after integration of the whole expression with respect to $p$.

Finally, we have to expand the second factor in \eqref{Wpowers}. This correction contains some information about the initial position distribution, although contracted by a factor $1/t$. When $\rho(p_1,p_2)$ denotes the integral kernel of the initial density, the effect of the factor $\exp(i\lambda\cdot Q/t)$ under the trace is to shift the first argument of $\rho$. We therefore introduce the function
\begin{equation}\label{initialWigner}
    C_0(\lambda,p)=\tr\rho(p+\lambda,p),
\end{equation}
where the trace is over the internal degrees of freedom. This notation is to suggest that this is some kind of characteristic function of the initial distribution. This is not literally true, since we keep the momentum variable, so at best it is a phase space distribution function. Indeed, a slightly more symmetric version
$\tr\rho(p+\lambda/2,p-\lambda/2)$ is the Fourier transform of the ``position distribution at fixed momentum'' according to the Wigner distribution function. Of course, this distribution is contracted by a factor $1/t$ in ballistic scaling. Together, we get
\begin{equation}\label{MarkovNextorder}
   C_t(\lambda)=\int\!\!dp\ e^{\textstyle i\lambda\cdot v(p)-\frac1{2t}\lambda\cdot s(p)\cdot\lambda}
                \Bigl(C_0(\lambda/t,p) + \frac1t\tr\rho(p,p)\int\invm(d\gamma)\ A'(\gamma)\Bigr)
                +\ordert1.
\end{equation}
Heuristically, the interpretation of the first term is the distribution of the sum of two independent quantities for each $p$: a Gaussian centered at $v(p)$, with decreasing variance, plus a scaled down version of the initial position distribution. This is then averaged over momentum. The second term does not allow such a simple interpretation. Two numerical examples are shown for quantum walks in one lattice dimension in Sect. \ref{ex:misslinks} and \ref{subs:NonUnitKO} below.

\subsection{Memoryless decoherence}\label{sec:Bernoulli}
In many applications there is no memory in the control process: the walks $\Va_\gamma$ are chosen independently in each step. The Markov process is then a Bernoulli process, and the probability  $ \ma_\gamma(\eta)$ to end up in $\eta$ is the same from any state $\gamma$. Obviously, this probability is then also the invariant distribution, i.e., we have
\begin{equation}\label{bernoulli}
    \ma_\gamma(\eta)=\invm(\eta)
\end{equation}
for all $\gamma,\eta$.
We can look at the resulting process in two ways: On the one hand, we could just specialize using \eqref{bernoulli}, which is the approach we will take below. On the other hand we could completely discard the control process, so there is only one value $\gamma=0$, say, and the corresponding $\Va_0$ operation is $\Va_0=\int \invm(d\gamma)\Va_\gamma$. This would violate Assumption 3 but, of course, this case can nevertheless be fully analyzed. In fact, it simplifies the computation of the operators $A^\prime$ and the diffusive order $\mu^{\prime\prime}$ considerably. We can reduce the dimension of the system of linear equations that determines $\mu^{\prime\prime}$ by a factor $n$, where $n$ is the number of Kraus operators from which the control process can choose from.

Substituting \eqref{bernoulli} into equation \eqref{mupp} for $\mu''$ we get
\begin{align}
\label{muppBernoulli}
   \mu''=\int \invm(d\gamma)\ \tr\bigl(\invrho\Vt_\gamma''(\idty)\bigr)
          +2\int \invm(d\gamma)\
             \tr\bigl(\invrho\Vt'_\gamma\bigl(\int \invm(d\eta) A'(\eta)\bigr)\bigr) \; .
\end{align}
So in order to determine  $\mu''$  we do not need to know the individual $A'(\eta)$, but it suffices to know the value of the average
\begin{equation}
  X= \overline{ A'}=\int \invm(d\eta) A'(\eta).
\end{equation}
Our aim is to set up an equation directly for this unknown matrix $X$. From \eqref{WVMlong} we find
$(\Wa A')(\gamma)=\Va_\gamma(X)$, which turns \eqref{Ev1nondega} into a definition of $A'(\gamma)$ in terms of $X$
\begin{equation}\label{Aprimebernoulli}
    A'(\gamma)=\Va_\gamma(X)-i\lamv\idty + \Vt_\gamma'(\idty)
\end{equation}
The equation for $X$ now follows by averaging:
\begin{equation}\label{eqBernoulli}
    \int\invm(d\gamma)\Va_\gamma(X)-X=i\lamv\idty -\int \invm(d\gamma) \Vt_\gamma'(\idty).
\end{equation}
Eliminating the control process from \eqref{eqBernoulli} and \eqref{muppBernoulli} via the definition of $\Va_0$ and $X$ we get the equations
\begin{equation}
\label{XWOcontrol}
\Va_0\bigl(X\bigr)-X=i\lambda\cdot v \idty -\Va_0'(\idty)
\end{equation}
and
\begin{equation}
\label{WOcontrol}
\mu'' = \tr\bigl(\invrho\Vt_0''(\idty)\bigr)
          +2\ \tr \bigl( \invrho \Vt'_0(X\bigr)\bigr)\,.
\end{equation}

%% file: RCWeJMP.tex
\section{Examples}
\label{sec:pictures}
\subsection{The Hadamard walk with reflections}
\label{ex:misslinks}
The aim of this section is to derive the asymptotic distribution of $Q(t)/\sqrt{t}$ of a quantum walk where the control process is of Bernoulli type and in each time step it picks one of two unitary quantum walks. We will highlight the simplification of the formulas \eqref{Ev12nondeg} and \eqref{Ev2nondeg}, accounting for general Markov processes, for the case of Bernoulli type decoherence by also calculating the diffusive order for general Markov processes, which then doubles the dimension of the computation. The unitary quantum walks, the control process chooses from, admit a decomposition into shift and coin operation. The shift operation is assumed to be the same for both walks, hence, the control process chooses only the coin operation for each time step. With probability $1-\varepsilon$ it will pick the Hadamard coin for each time step, which, on its own, would lead to ballistic behavior. In the remaining case the Pauli matrix $\sigma_1$ will be applied, this unitary quantum walk hinders the particle from moving at all, the particle will be reflected at each site. The Kraus operators of the quantum walk read
\begin{equation}\label{HadReflDef}
K_1(p) = \sqrt{\frac{1-\varepsilon}{2}}
\left(
\begin{array}{cc}
e^{ip} & e^{-ip} \\
e^{ip} & -e^{-ip}
\end{array}
\right)
\quad , \quad
K_2(p) = \sqrt{\varepsilon}
\left(
\begin{array}{cc}
0 & e^{-ip} \\
e^{ip} & 0
\end{array}
\right)\,.
\end{equation}
First, we note that the velocity $v(p)$ is zero, which follows from Eq.~\eqref{vdecoh} and $\Vt ' (\idty)=K_1^*K_1' + K_2^*K_2'=i\lambda\sigma_3$. Alternatively, one could also rephrase this Bernoulli type walk as a Markov controlled walk and then use formula \eqref{meanindex} to get $v(p)=0$. This means, in ballistic scaling the position distribution converges to a point measure at the origin.

In order to determine the diffusive scaling we make directly use of Eq.~\eqref{WOcontrol}. It is easily seen that
\[
\Vt'' (\idty )=K_1^*K_1'' + K_2^*K_2'' = -  \lambda^2 \idty
\]
and hence by \eqref{muppBernoulli} and $K_1'K_1^*=(1-\varepsilon) \sigma_3$ together with $K_2'K_2^*=\varepsilon\, \sigma_3$
\[
s(p)\lambda^2 = -\mu'' =\lambda^2 - \sum_{\alpha=1,2}\tr \left( K_\alpha^* A'K_\alpha '\right)=\lambda^2 -i\lambda \left ( (1-\varepsilon )\tr \left( H \sigma_3 H^* A'\right) +\varepsilon \tr \left( \sigma_1\sigma_3\sigma_1^* A'\right) \right)
\]
with the Pauli matrix $\sigma_1$ and the usual Hadamard matrix $H$. Using the rules $H^*\sigma_3 H=\sigma_1$ and $\sigma_1\sigma_3\sigma_1=-\sigma_3$ we get
\[
s(p)= \lambda^2 -i2\lambda \left( r(1-\varepsilon )a_1 -\varepsilon a_3\right)
\]
where we abbreviated $a_i=2^{-1}\tr \sigma_i A'$, which implies $A'=\sum_{i_0}^3a_i\sigma_i$. Now, $A'$ is determined by \eqref{XWOcontrol}, which reads
\[
\sum_{i=0}^3 a_i \left( K_1 ^* \sigma_i K_1 + K_2^* \sigma_i K_2 -\sigma_i\right) =-i\lambda \sigma_3 \,.
\]
By calculating the Hilbert-Schmidt scalar product of this equation with the Pauli matrices $\sigma_i$ we get the following system of equations
\[
\left(
\begin{array}{ccc}
\varepsilon \cos (2p)-1& \sin (2p) &(1-\varepsilon )\cos (2p)\\
\varepsilon \sin (2p) & -\cos (2p)-1 & (1-\varepsilon )\sin (2p)\\
1-\varepsilon & 0 & -1-\varepsilon
\end{array}
\right)
\cdot
\left(
\begin{array}{c}
a_1\\
a_2\\
a_3
\end{array}
\right)
=
\left(
\begin{array}{c}
0\\
0\\
-i\lambda
\end{array}
\right) \,,
\]
where we have chosen $a_0=0$ in order to satisfy $\invm \otimes \invrho (A')=0$. The solution to these equations is $a_1=a_3=i\lambda /2\varepsilon$ and $a_2=i\lambda \tan (p)/2\varepsilon$. This yields the following expression for the variance
\[
s_B(p)=\frac{1-\varepsilon}{\varepsilon}\,.
\]
By inverse Fourier transform of the characteristic function $\lim\limits_{t\to\infty}C_{D(t)}(\lambda)=e^{-\lambda^2(1-\varepsilon)/(2\varepsilon)}$ we obtain the asymptotic position distribution in diffusive scaling
\[
\mathrm{P}(x)=\sqrt{\frac{\varepsilon}{2\pi(1-\varepsilon)}}e^{-x^2\frac{\varepsilon}{2(1-\varepsilon)}}\,.
\]
The $1/t$ correction to the asymptotic position distribution in ballistic scaling can be inferred from \eqref{MarkovNextorder}, we omit the computation and just give the resulting probability distribution (cf. Fig. \ref{Plot:HadReflNextOrd}), which reads
\begin{equation}\label{Eq:HadReflNextOrd}
P_t(x,\veps) = \sqrt{\frac{2\pi t \veps}{1-\veps}}\Bigl(1+\frac{x}{2(1-\veps )}\Bigr)\exp \Bigl(-\frac{t x^2\veps}{2(1-\veps)} \Bigr)\,.
\end{equation}
\begin{figure}[htb]
  \centering
 \includegraphics[width=0.5\textwidth]{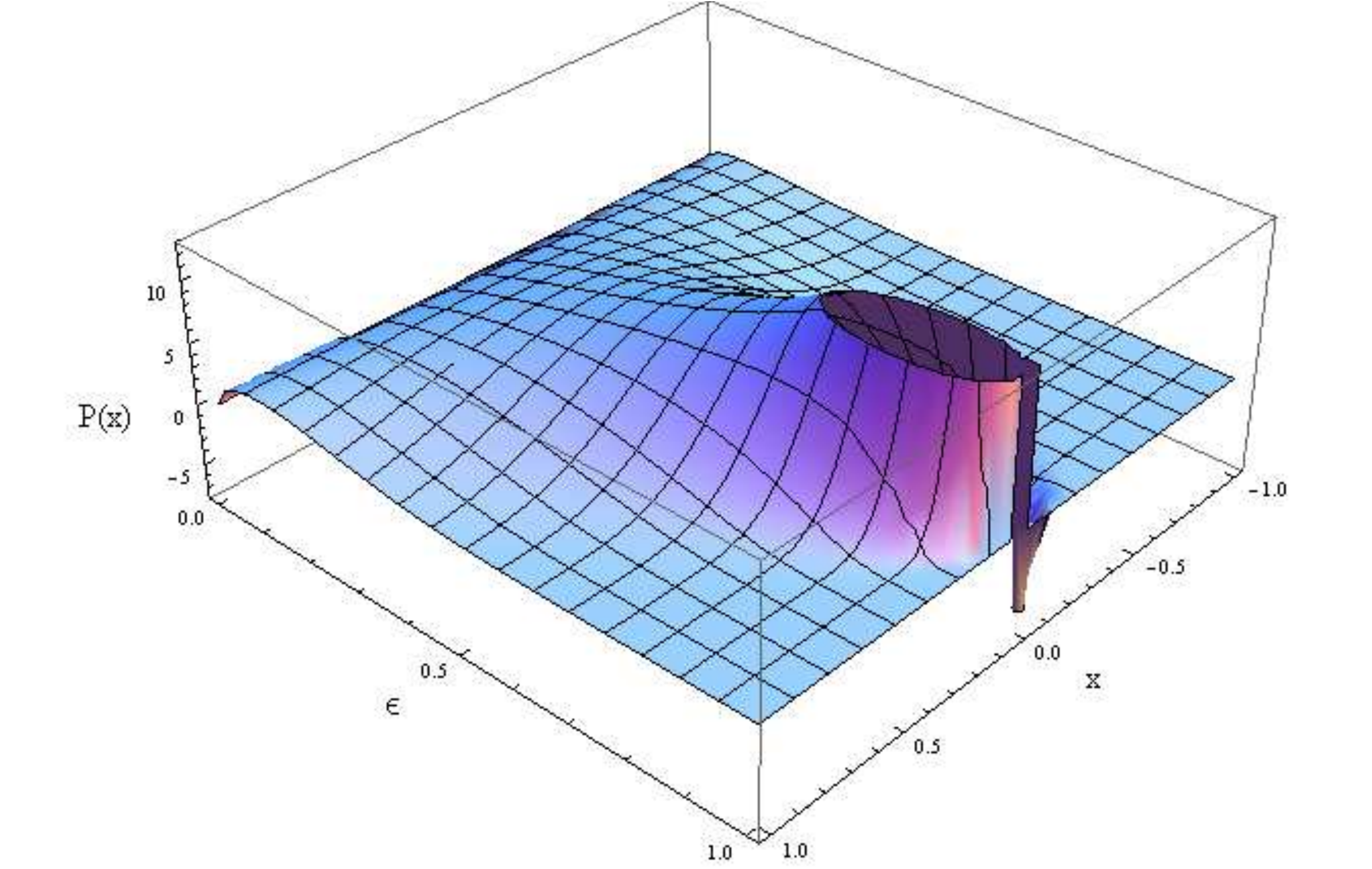}
    \caption{(color online) For the quantum walk according to \eqref{HadReflDef} the plot shows the $1/t$ correction to the asymptotic position distribution in ballistic scaling after $t=10$ depending on the parameter $\veps$.}\label{Plot:HadReflNextOrd}
\end{figure}

For a general Markov process $\Ma$, which chooses from the two unitary quantum walks, the time evolution reads
\[
\Wa (A) (\gamma ) = \sum_{\eta=1,2} \ma_\gamma (\eta ) K_\gamma^* A(\eta) K_\gamma\, , \quad\gamma =1,2
\]
now with Kraus operators and transition matrix
\[
K_1(p) = \frac{1}{\sqrt{2}}
\left(
\begin{array}{cc}
e^{ip} & e^{-ip} \\
e^{ip} & -e^{-ip}
\end{array}
\right)
\quad , \quad
K_2(p) =
\left(
\begin{array}{cc}
0 & e^{-ip} \\
e^{ip} & 0
\end{array}
\right)
\quad , \quad
\Ma =
\left(
\begin{array}{cc}
m_1 & 1-m_1 \\
1-m_2 & m_2
\end{array}
\right)\, .
\]
The Bernoulli control process corresponds to the choice $m_1=1-m_2=1-\varepsilon$. Since $K_1(p)$ and $K_2(p)$ are irreducible for almost all $p$ Proposition \ref{prob:velirred} applies. Observing that $\ind K_\gamma =0 $ we get $v(p)=0$ and hence $Q(t)/t$ converges to a point measure at zero. Before calculating the diffusive order of this quantum walk we will prove a useful lemma.
 \begin{lem}
 Let $ K_\gamma\, ,\gamma\in \Gamma$ be a finite collection of unitary walk operators on $\Ir$ with two dimensional coin. Assume the $K_\gamma$ admit a shift coin decomposition $ K_\gamma = U_\gamma \cdot  S$ with
 \[
  S(p)=\left(
 \begin{array}{cc}
 e^{ip} & 0\\
 0 & e^{-ip}
 \end{array}
 \right)\,
 \]
and $p$ independent unitaries $U_\gamma$ such that Assumption \ref{ass:chaos} is satisfied. Let the overall time evolution $\Wa$ be given by a Markov process $\Ma$ that chooses from the $ K_\gamma$, i.e.
 \[
 \Wa(A)(\gamma )=\sum_{\eta \in \Gamma} \ma_\gamma (\eta)  K_{\gamma }^* A(\eta)  K_{\gamma } \,.
 \]
 Then $v(p)=0$ and the diffusion constant is determined by
 \[
 \lambda^2\cdot s(p)=-\lambda^2-2\lambda i\sum_{\gamma\in\Gamma} \invm_\gamma a_3(\gamma)\, ,
 \]
 where $a_3(\gamma) = 2^{-1}\tr (\sigma_3 A'(\gamma)$ with $A'$ according to \eqref{Ev12nondeg} and the Pauli matrix $\sigma_3$.
 \end{lem}
 \begin{proof}
 Since $\ind W_\gamma =0$ for all $ \gamma$ it follows from \eqref{meanindex} that $v(p)=0$.

 According to \eqref{Ev12nondeg} and \eqref{st1}, we need to solve the equation
 \[
 \Wa(A')(\gamma)-A'(\gamma) = - \Vt_\gamma'(\idty)
 \]
 in order to determine the diffusive scaling of $\Wa$. To begin with, we decompose the components $A'(\gamma)$ into Pauli matrices $\sigma_i\, ,i=0,1,2,3$ with $\sigma_0=\idty_2$
 \[
 A'(\gamma) =\sum_{i=0}^3 a_i(\gamma) \sigma_i\,.
 \]
 Next, it is easily seen that for arbitrary two dimensional unitaries $U_\gamma$ we have $ K_\gamma ^* K_\gamma^\prime  = i\lambda \sigma_3$ and hence
 \[
 \Wa (A')(\gamma ) = (a_3 (\gamma) -i\lambda)\sigma_3+\sum_{i=0}^2 a_i(\gamma) \sigma_i \,.
 \]
 Now, the variance $s(p)$ can be computed from \eqref{st1}
 \[
  \lambda^2 \cdot s(p) = \invm\otimes\invrho\Bigl(\abs{A'}^2-\abs{W(A')}^2\Bigr) = \frac{1}{2}\sum_{\gamma \in \Gamma} \invm_\gamma \tr \left(\abs{A'(\gamma)}^2-\abs{W(A'(\gamma))}^2\right)
 \]
 and since $\tr |A'(\gamma)|^2 = 2\cdot \sum_{i=0}^3 \abs{a_i(\gamma)}^2$ it follows
\[
 \lambda^2\cdot s(p) = \sum_{\gamma\in\Gamma} \invm_\gamma(\abs{a_3(\gamma)}^2-\abs{a_3(\gamma )-i\lambda }^2)=-\lambda^2-2\lambda i\sum_{\gamma\in\Gamma} \invm_\gamma a_3(\gamma) \,,
 \]
 which is real since $A'(\gamma)$ is skew-hermitian.
 \end{proof}
By a straightforward calculation we get the following system of equations for $A'$
\[
\left(
\begin{array}{cccccc}
\scriptstyle
 -1 &\scriptstyle m_1\sin ( 2p) &\scriptstyle m_1\cos (2p) &\scriptstyle 0 &\scriptstyle (1-m_1)\sin (2p) &\scriptstyle (1-m_1)\cos (2p) \\
\scriptstyle 0 &\scriptstyle -1-m_1\cos (2p) &\scriptstyle m_1\sin (2p) &\scriptstyle 0 &\scriptstyle (m_1-1)\cos (2p) &\scriptstyle (1-m_1)\sin (2p) \\
\scriptstyle m_1 &\scriptstyle 0 &\scriptstyle -1 &\scriptstyle 1-m_1 &\scriptstyle 0 &\scriptstyle 0 \\
\scriptstyle (1-m_2) \cos (2p) &\scriptstyle (1-m_2)\sin (2p) &\scriptstyle 0 &\scriptstyle m_2\cos (2p)-1 &\scriptstyle m_2\sin (2p) &\scriptstyle 0 \\
\scriptstyle (1-m_2) \sin (2p) &\scriptstyle (m_2-1) \cos (2p) &\scriptstyle 0 &\scriptstyle m_2\sin (2p) &\scriptstyle -m_2\cos (2p)-1 &\scriptstyle 0 \\
\scriptstyle 0 &\scriptstyle 0 &\scriptstyle m_2-1 &\scriptstyle 0 &\scriptstyle 0 &\scriptstyle -m_2 -1
\end{array}
\right)\cdot {\bf a}
=\left(
\begin{array}{c}
0\\
0\\
-i\lambda \\
0\\
0\\
-i\lambda
\end{array}
\right)\, ,
\]
with coefficient vector ${\bf a}=(a_1(1),a_2(1),a_3(1),a_1(2),a_2(2),a_3(2))$. Actually, there are two more equations for the variables $a_0(1)$ and $a_0 (2)$, but these are already fixed to be zero in order to guarantee the condition $\invm\otimes\invrho (A')=0$. Solving these equations we get the result
 \[
 s_M(p)=\frac{1-m_2}{1-m_1}\cdot \frac{m_2+m_1}{2-m_1-m_2}
 \]
The comparison of $s_B(p)$ and $s_M(p)$ shows that the diffusive order of the quantum walk can distinguish between Bernoulli and Markov type decoherence, as opposed to the ballistic order which is the same for both control processes. In both cases, we see that in the coherent limit $\varepsilon \to 0$ and $m_1\to 1$ the variance $s_{B,M}(p)$ diverges independent of $m_2$. For the Markov process we can also consider the limit $m_2\to 1$ which gives us localization, expressed by $s_M(p)\to 0$.

\subsection{One dimensional Quantum Walks with equal position distribution}
\label{ex:1DQWEqDist}
In experimental implementations of one dimensional quantum walks one source of decoherence can be identified as dephasing of the internal degree of freedom of the walking particle \cite{Bonn}. Such a dephasing error can be modeled by introducing an additional z-Rotation $R(\theta)$ of angle $\theta$ before the application of the coin operation, such that a time step of the quantum walk $ W_\theta$ is then given by
\begin{align*}
   W_\theta(p) = C\cdot R(\theta) \cdot S = C \left(
               \begin{array}{cc}
                 e^{i\theta} & 0 \\
                0 & e^{-i\theta} \\
               \end{array}
             \right) S \,
\end{align*}
where $C$ and $S$ are the coin and shift operations as introduced in Eq.~\eqref{eqn:1dwalk}. In the following we will consider a quantum walk with Bernoulli type decoherence, the control process is assumed to chose a value $\theta$ in each time step independently according to which $W_\theta$ is applied subsequently. This quantum walk applied, up to special values of $\theta$, Assumptions \ref{ass:nomom}, \ref{ass:chaos} and \ref{ass:rcoin}, hence it is easy to compute the ballistic and diffusive scaling.

One problem faced in experiments is that the quantum walk $ W = C\cdot S$ can not easily be distinguished from its disturbed counterpart $ W_\theta = C\cdot R(\theta)\cdot S$, because of the following little lemma.
\begin{lem}
  For an initial state $ \psi(p) = \phi(z) e^{-izp}$ localized at lattice point $z$ the position distribution after any number of time steps $t$, generated by the quantum walks $ W = C\cdot S$ and $ W=C\cdot R(\theta)\cdot S$, are identical, where $R(\theta)$ is an arbitrary z-Rotation.
\end{lem}
\begin{proof}
  In the Fourier picture it is easy to see that $ W$ and $ W_\theta$ deviate only by a constant momentum shift, i.e. $ W_\theta(p) =  W(p+\theta)$. Computing the state at lattice site $x$ after $t$ time steps, starting with a state $\psi$ initially localized at lattice site z we find:
  \begin{align*}
   (W^t_\theta\psi)(x) &= \frac{1}{2\pi} \int dp\  W^t_\theta(p) \psi(p) e^{ipx} = \frac{1}{2\pi} \int dp\  W^t(p+\theta) \psi(z) e^{ip(x-z)}\\
   &=\frac{1}{2\pi}  e^{-i\theta(x-z)} \int dp\  W^t(p) \psi(z) e^{ip(x-z)} =  e^{-i\theta(x-z)} (W^t\psi)(x)
  \end{align*}

  Since both amplitudes deviate only by a global phase the probability to find the particle at a lattice point $x$, and therefore the position distribution, coincide.

\end{proof}

So if we start an experiment in a localized state we cannot distinguish between the two. On the other hand, however, implementing a quantum walk, that will in every time step apply the undisturbed walk $W$ with some probability $\veps$ and with probability $(1-\veps)$ apply the z-rotated quantum walk $W_\theta$, will for the most values of $\theta$ result in diffusive behavior.

As an explicit example consider as the undisturbed walk the one dimensional Hadamard walk as defined in Sect. \ref{sec:unitary}. The two Kraus operators in momentum space are then given by
\begin{align}\label{DephasedHadamard}
   K_1(p) = H S = \sqrt{\frac{\varepsilon}{2}}
\left(
\begin{array}{cc}
e^{ip} & e^{-ip} \\
e^{ip} & -e^{-ip}
\end{array}
\right)&  & K_2(p) = H R(\theta) S =  \sqrt{\frac{1-\varepsilon}{2}}
\left(
\begin{array}{cc}
e^{i(p+\theta)} & e^{-i(p+\theta)} \\
e^{i(p+\theta)} & -e^{-i(p+\theta)}
\end{array}
\right)
\end{align}
One can check that only the cases $\theta\in\{0,\pi,2\pi\}$ lead to commuting Kraus operators that violate part three of Assumption \ref{ass:chaos} and are treated in example \ref{ex:commKraus}. In all other cases we get an asymptotic deterministic velocity $v(p)=0$ by Proposition \ref{asymdistr}.

\begin{figure}[htb]
  \centering
 \includegraphics[width=0.45\textwidth]{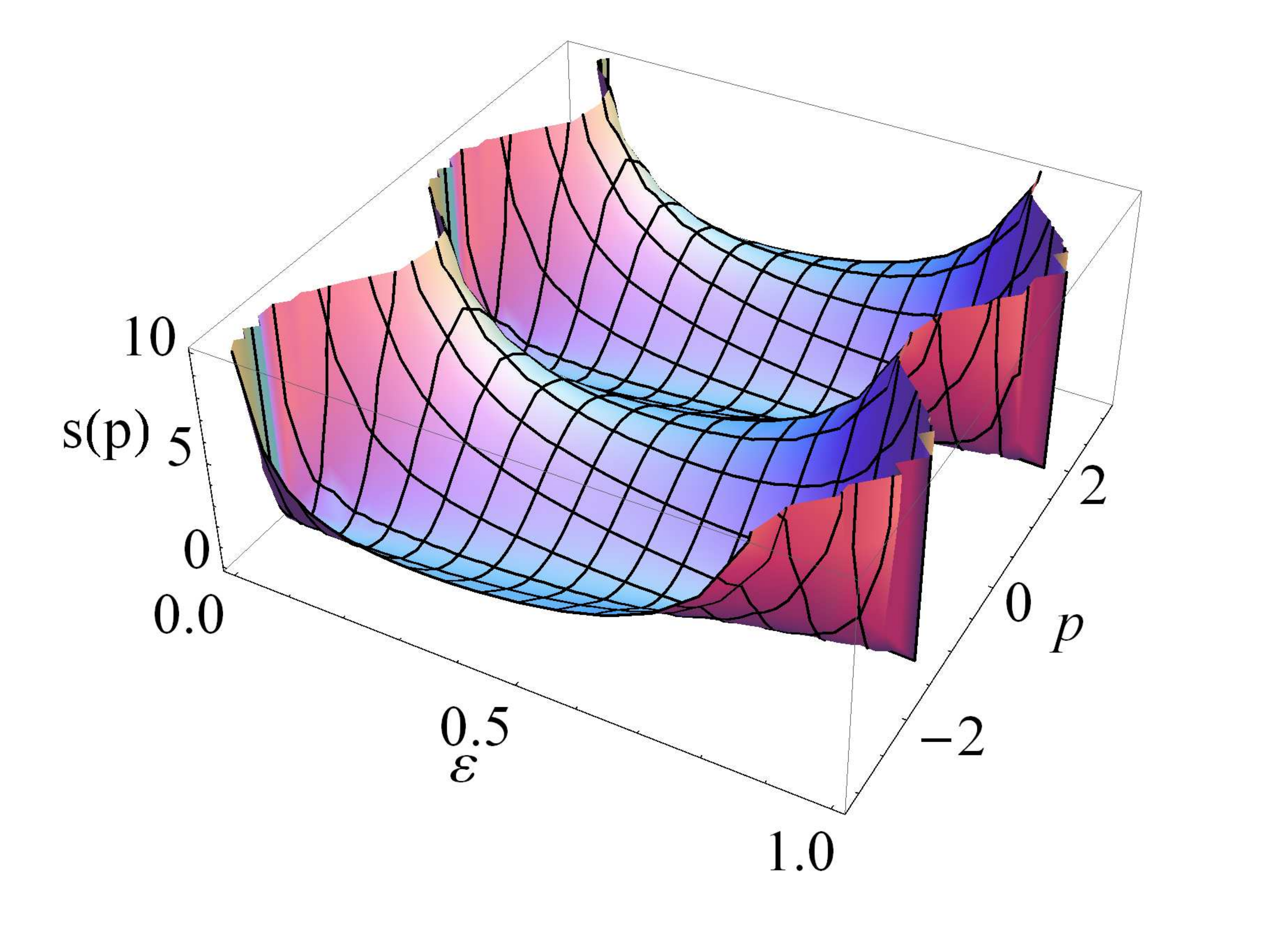}
    \caption{(color online) The plot shows the variance $s(p)$ for a dephased Hadamard Walk \eqref{DephasedHadamard} with $\theta = \frac{\pi}{3}$. The undisturbed Hadamard Walk is applied with probability $\veps\in[0,1]$ and the walk $W_{\pi/3}$ is applied with probability $(1-\veps)$.}\label{Plot:rotHad_diff}
\end{figure}

To compute the diffusion constant we follow the instructions of Sect. \ref{sec:Bernoulli} and calculate the operator $ A'=\sum_j a_j \sigma_j$, here decomposed in Pauli matrices. Choosing $ A'$ to have zero trace we find $a_0=0$ and are left with three parameters $a_j$ that can directly be inferred from Eq.~\eqref{XWOcontrol}, which after eliminating equivalent conditions reads
\begin{align}
   \left(
               \begin{array}{ccc}
                {\scriptstyle 1} & {\scriptstyle i(1+(1-\veps)e^{2i(p+\theta)}+\veps e^{2ip})} & {\scriptstyle -((1- \veps)e^{2i(p+\theta)}+\veps e^{2ip})} \\
                 {\scriptstyle 1} & {\scriptstyle -i(1 +(1-\veps)e^{-2i(p+\theta)}+\veps e^{-2ip}  )} & {\scriptstyle -((1- \veps)e^{-2i(p+\theta)}+\veps e^{-2ip})} \\
                 {\scriptstyle 1} & {\scriptstyle 0} & {\scriptstyle -1} \\
               \end{array}
             \right)\cdot
              \left(
               \begin{array}{c}
                a_1\\
                a_2\\
                a_3\\
               \end{array}\right) = \left(
                \begin{array}{c}
                0\\
                0\\
                -i\lambda\\
               \end{array}\right)\; .
\end{align}
Inferring $ A'$ from this equation, we can compute $\Vt_\gamma'( A')$ and since
\[
\tilde \Va''(\idty) = \sum_j  K_j^*(p)K_j''(p)= -\lambda\idty^2
\]
as in example \ref{ex:misslinks} we can determine the variance from Eq.~ \eqref{muppBernoulli}
\begin{align}
  s(p) = \frac{\lambda^2}{\sin(\theta)^2}\frac{\cos(p+\theta)^2 + \veps\sin(2p+\theta)\sin(\theta)}{\veps(1-\veps)} \; .
\end{align}
In contrast to example \ref{ex:misslinks} we have an explicit $p$ dependence of $s$. As one expects, in the limits $\veps=1$ and $\veps=0$, where only one coin is taken and the decoherence vanishes, the diffusion constant diverges. The same is true for the cases $\theta\in[0,\pi,2\pi]$, where the Kraus operators commute and the quantum walk once again exhibits ballistic spreading as will be shown in example \ref{ex:commKraus}. So, also in the coherent limit $\theta\,\mathrm{mod}\,2\pi\rightarrow 0 $, where the two quantum walks, the Markov process chooses from, become equal, the diffusion constant diverges.

In Fig.\ref{Plot:rotHad_diff} we plot the variance $s(p)$ for a z-Rotation of $\theta=\frac{\pi}{3}$. As explained in the last paragraph, we can observe divergence of $s(p)$ for the coherent limits $\veps=0$ and $\veps = 1$. In between these two regimes the $p$ dependence of $s(p)$ is recognizable.

\subsection{Non-unitary Kraus operators}
\label{subs:NonUnitKO}
The following example is a simple version of a quantum walk which satisfies Assumption \ref{ass:nomom} and \ref{ass:chaos} but not Assumption \ref{ass:rcoin}. This means there is no momentum transfer, and according to Proposition \ref{prob:simpev} the eigenvalue $1$ of $\Wa$ is non-degenerate for almost all $p$, but the Kraus operators are non-unitary. Hence, the results of Sect. \ref{sec:ballistic} are applicable, and in particular formula \eqref{vdecoh} determines the asymptotic behavior of the expectation value of $Q(t)/t$, but now with a momentum dependent velocity $v(p)$.

We consider a one dimensional lattice with no internal degree of freedom. The decoherence will be of Bernoulli type, i.e. the Markov chain is actually trivial. The Kraus operators of the quantum walk are defined by
\begin{equation}\label{nonUniKraus}
   (K_1 \psi) (x)=\frac{1}{2} (\psi (x)+\psi (x+1))\quad ,\quad (K_2 \psi )(x)=\frac{1}{2} (\psi (x)-\psi (x-1))\,.
\end{equation}
These operators satisfy the normalization condition $\sum_i K_i^*K_i=\idty$ and their Fourier transforms are given by
\[
 K_1 (p) =\frac{1}{2} (1+e^{ip})\quad ,\quad  K_2(p)=\frac{1}{2} (1-e^{-ip})\,.
\]
In fact, the Kraus operators $K_1$ and $K_2$ commute, hence, according to Sect. \ref{ex:commKraus}, the behavior of this quantum walk will be ballistic. The limit of the characteristic function of $Q(t)/t$ is according to the discussion in Sect. \ref{sec:AsymPosPertMeth} given by
\[
\lim\limits_{t\rightarrow \infty}C_{Q(t)/t}(\lambda )=\lim\limits_{t\rightarrow \infty} \int dp \tr \rho(p) \Wt_{1/t}^t (\idty)\,.
\]
Although our theory applies to this problem, it is instructive to calculate this limit using only Fourier methods. We need to compute the operator $\Wt_{1/t}^t (\idty)$, in momentum space this becomes a one dimensional problem, for we have
\[
\Wt_{\veps}(X)=\sum_i K_i^*(p) X(p) K_i(p+\lambda\veps) =X(p)\sum_i K_i^*(p) K_i(p+\lambda\veps)\,.
\]
But this means we can solve the equation $\Wt_\veps A_\veps=\mu_\veps A_\veps$ exactly. Clearly, $A_\veps =\idty$ and $\mu_\veps(p) = \sum_i K_i^*(p)K_i(p+\lambda\veps)$ solves the equation.
Now, in ballistic scaling $\veps=1/t$ the operator $\Wt_{1/t}^t(\idty)$ is just given by
\begin{eqnarray}\label{mutot}
\mu_{1/t}^t&=&
\frac{1}{4^t}\left( (1+e^{-ip})(1+e^{i(p+\lambda /t)})+(1-e^{ip})(1-e^{-i(p+\lambda /t)})\right)^t \\
&=& \left( \frac{1+\cos (\lambda / t) + i(\sin (p+\lambda /t)-\sin (p))}{2}\right)^t\,.
\end{eqnarray}
By a Taylor expansion of this equation to second order and using the formula $\lim_{t\to\infty}(1+x/t+\order(t^{-1}))^t=\exp(x)$ once again we obtain the limit
\[
\lim\limits_{t\rightarrow \infty}C_{Q(t)/t}(\lambda ) = \int dp \,\rho (p) e^{i\lambda\cos(p)/2}\,.
\]
As already pointed out, we would have obtained the same result by applying \eqref{vdecoh} and \eqref{v-example}. And indeed, these formulas can be used to determine the asymptotic behavior of more general quantum walks.
\begin{prop}
\label{NonUnitKrausOps}
Let the Kraus operators of a quantum walk in one lattice dimension with no internal degree of freedom be defined by
\[
(K_i\psi )(x)=\sum_{k}a_{ik}\psi (x-k) \quad \Leftrightarrow \quad  K_i (p)=\sum_k a_{ik}e^{ip k}\,.
\]
where the coefficients $a_{ik}$ have to satisfy the constraint $\sum_{i,k} \bar a_{i(k+x)}a_{ik} = c_0 \delta_{0,x} \, \forall \,x \in \Ir $
in order to guarantee the normalization condition $\sum_i K_i^*K_i=\idty$. Then the group velocity and hence the asymptotic behavior of the quantum walk is determined by
\[
v(p) = \sum_x |\gamma_x | \cos (px+\arg (\gamma_x))\,,\quad \gamma_x = \sum_{i,k}\bar a_{i(k-x)}a_{ik} k \,.
\]
\end{prop}
\begin{proof}
The statement follows directly from Eq.~\eqref{vdecoh}.
\end{proof}
Interestingly, we can infer from the characteristic function that the behavior of our example walk is truly ballistic. In fact, the second moment of $Q(t)/t$ in the asymptotic limit is non-zero:
\[
\lim\limits_{t\rightarrow \infty}\left\langle Q(t)^2/t^2\right\rangle = -\left.\frac{\partial^2}{\partial^2 \lambda}C(\lambda)\right\vert_{\lambda=0}= \int dp\,\frac{\cos (p)^2}{4} \rho(p)
\]

We would like to determine the asymptotic distribution for arbitrary initial states $\rho$ and according to Proposition \ref{locinv} this requires knowledge of the group velocity $v(p)$, which is for our example given by
\[
v(p)=\frac{\cos(p)}{2}\,.
\]
By decomposing momentum space into subsets $(-\pi,0]$ and $(0,\pi]$ on which the function $v(p)$ is invertible and exploiting the point symmetry of $v(p)$ we obtain the asymptotic distribution of $Q(t)/t$
\[
  \mathrm{P} (x) = \frac{2}{\sqrt{1-4x^2}}\left(\rho (v^{-1} (x))+\rho (-v^{-1} (x)) \right)\, ,
\]
where $v^{-1}$ denotes the inverse of $v$ restricted to $(-\pi,0]$. Again, we see that at the points where the derivative $\partial v/\partial p$ vanishes, i.e. at the caustic points, there are peaks in the limiting distribution. If we choose $\rho$ to be located at the origin, i.e. $\rho (p) = 1/2\pi$, we get
\[
\mathrm{P}(x) = \frac{2}{\pi \sqrt{1-4x^2}}\,.
\]
The Fig. \ref{Fig:NonUnitKO1} compares this asymptotic distribution with the scaled probability distributions for a finite number of time steps.
\begin{figure}[htb]

  \centering
  \subfigure[]{\includegraphics[width=0.45\textwidth]{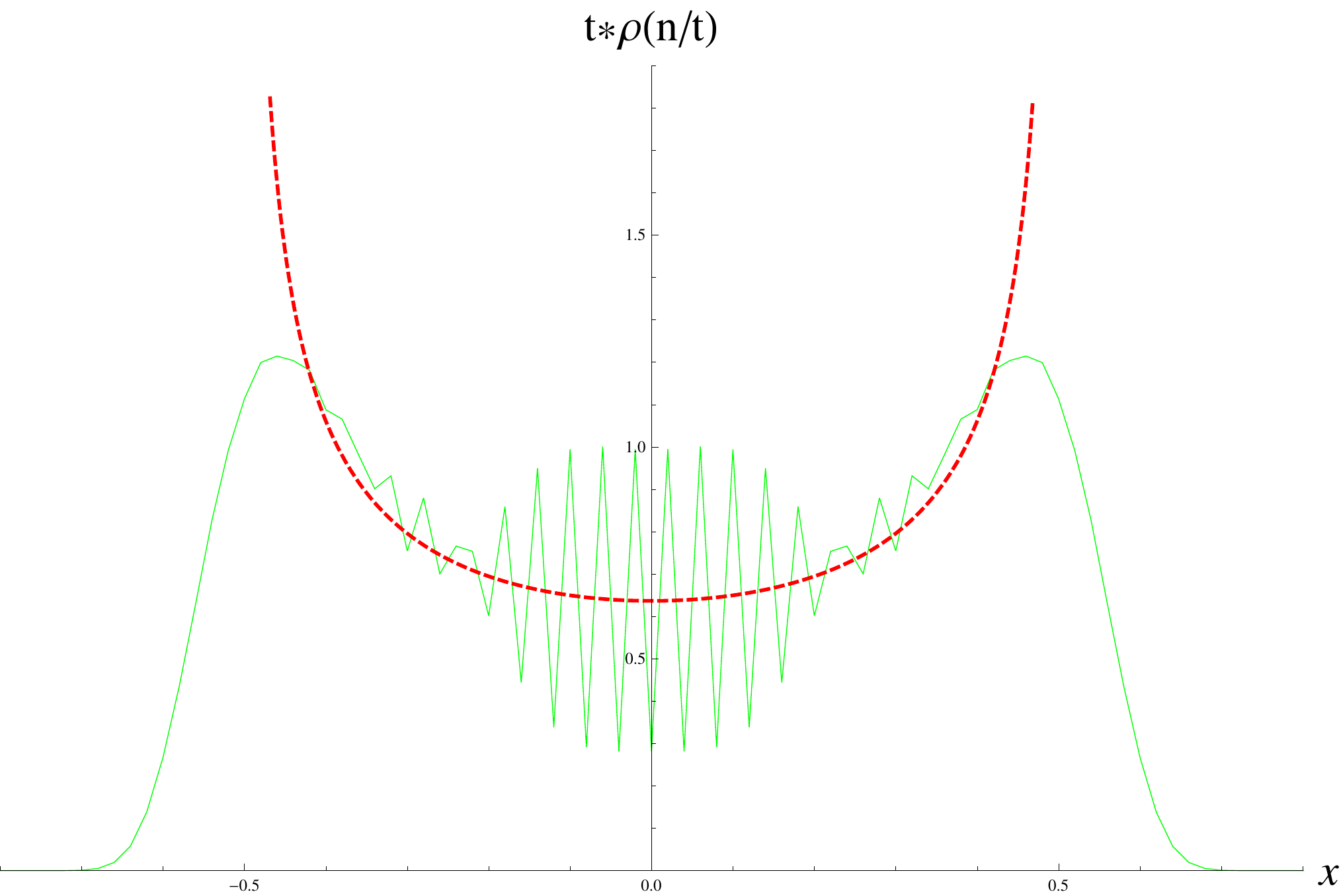}}
  \hspace{1cm}
  \subfigure[]{\includegraphics[width=0.45\textwidth]{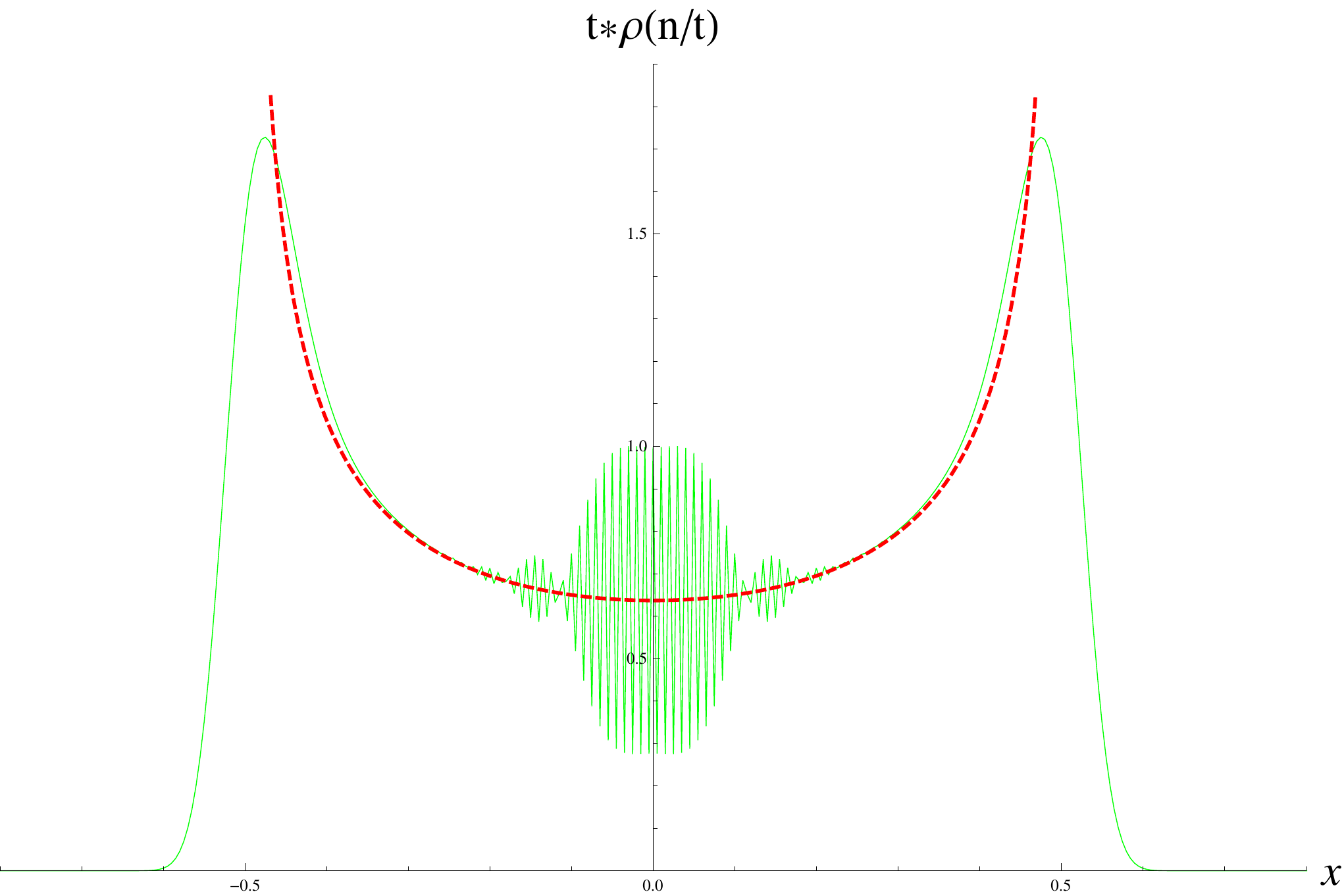}}
    \caption{(color online) Position distributions of the walk defined by \eqref{nonUniKraus} after (a) 50 and (b) 200 time steps (green/dotted lines) in ballistic scaling. The red/dashed line shows the asymptotic position distribution for comparison.}
\label{Fig:NonUnitKO1}
\end{figure}
Our knowledge of the exact eigenvalue $\mu_\veps$ allows us to analyze the asymptotic distribution in further detail. The Taylor expansion of $\mu_\veps$ to second order reads
\[
\mu_\veps = 1 + \veps\mu'  + \frac{\veps^2}{2}\mu''  +\Order (\veps^3) = 1+ i\frac{\veps\lambda}{2} \cos (p) - \frac{\veps^2\lambda^2}{4}(i\sin (p)+1)+\Order (\veps^3)\,.
\]
Hence, the variance $s(p)$ for this example is a complex valued function
\[
s(p)\lambda^2=-\mu''-(v\cdot \lambda)^2=\frac{2+i2\sin (p)-\cos^2 (p)}{4}\,\lambda^2\,.
\]
The imaginary part of this ``variance'' defies its interpretation as the variance of added Gaussian noise. However, the corrections to the probability distribution computed from it will be real, after integration over momenta. In a similar vein, \eqref{mutot} looks like the characteristic function of a sum of $t$ random variables, but with complex ``probabilities''. Only after integration over $p$ the expression gives a probability distribution.

Using the ideas from Section \ref{sec:HighOrdWOBall} we can determine the $1/t$ correction to the asymptotic distribution stemming from the second order of the perturbation expansion in $\veps$. Since we have $A'=0$ and we have chosen $\rho(p_1,p_2) =1/2\pi$ we get
\[
C_t(\lambda) = \frac{1}{2\pi}\int \!dp e^{i\lambda v(p)}\Bigl(1-\frac{2+i2\sin (p)-\cos^2 (p)}{8t}\,\lambda^2\Bigr)+\order(t^{-1})\,.
\]
With $v(p)=\cos (p)/2$ the integral over $p$ yields a sum of two Bessel functions for the first order approximation $C_1(\lambda,t)$ of $C_t(\lambda)$, i.e. $C_t(\lambda)=C_1(\lambda,t) + +\order(t^{-1})$ and
\[
C_1(\lambda,t) =\frac{1}{8t}\left((8t-\lambda^2)J_0(\lambda/2)-2\lambda J_1(\lambda/2)\right)\,.
\]
As already pointed out in Sect. \ref{sec:higherUnitary}, $C_1$ is not integrable over $\Rl$ (cf. Fig. \ref{CharFuncFirstOrd}) and therefore we need to introduce a cutoff in the integration for the inverse Fourier transform of $C_t(\lambda)$. In order to smoothen the resulting probability distribution we multiplied $C_1(\lambda,t)$ with a Gaussian $g(\lambda ,t)=e^{-(x/t)^2}$ and computed the inverse Fourier transform numerically. This smoothening was necessary because of the rapidly oscillating behaviour of $C_1(\lambda,t)$, which can be seen in Fig. \ref{CharFuncFirstOrd}, leading to a poor convergence of the numerical integration. The resulting correction to the asymptotic position distribution for 10 time steps, also shown in Fig. \ref{CharFuncFirstOrd}, is in good agreement with the exact position distribution.
\begin{figure}[htb]
  \centering
  \subfigure[]{\includegraphics[width=0.45\textwidth]{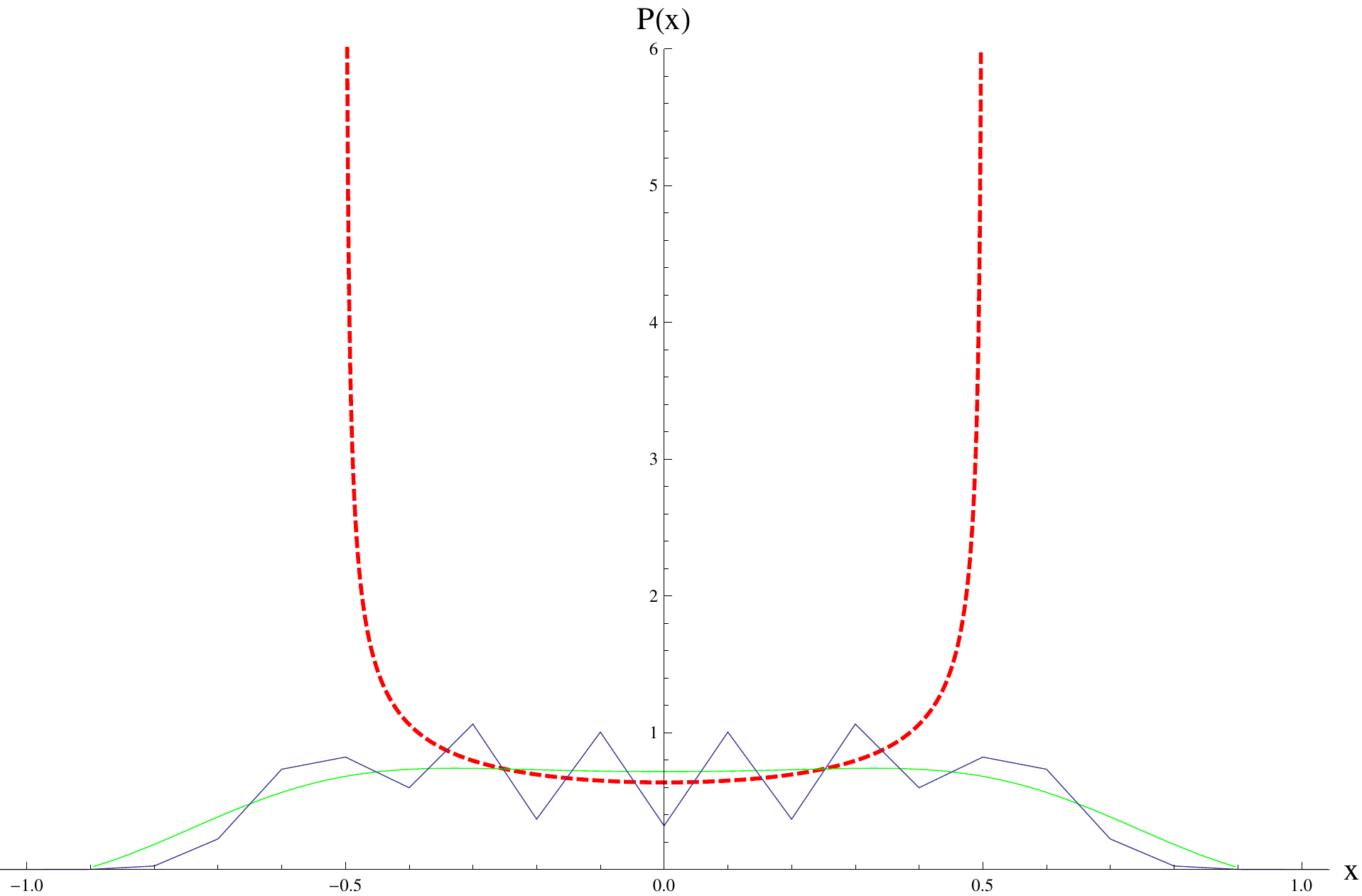}}
  \hspace{1cm}
  \subfigure[]{\includegraphics[width=0.45\textwidth]{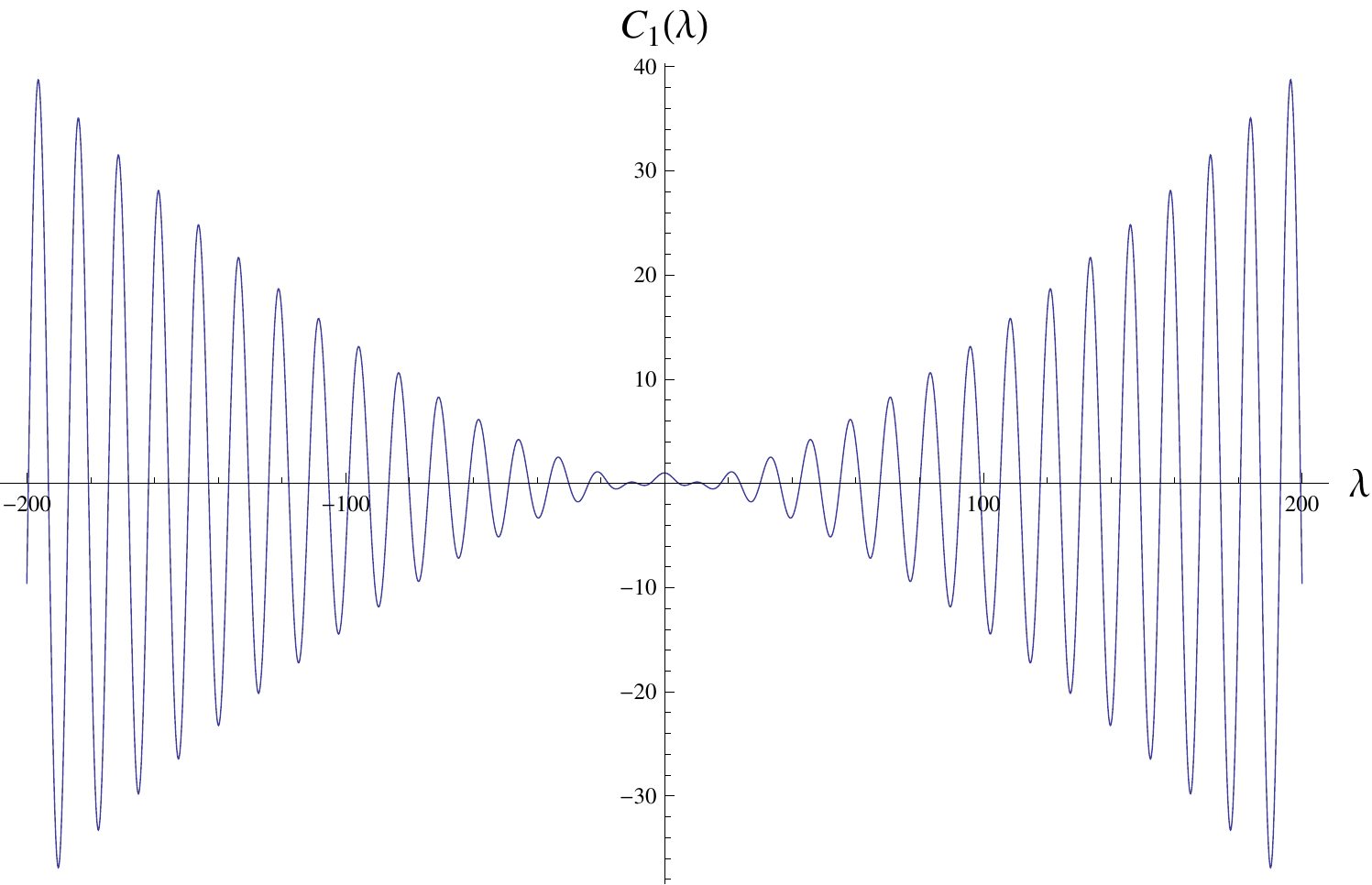}}
    \caption{(color online) Plot (a) shows the correction of order $1/t$ (green/solid) to the asymptotic distribution (red/dashed) for the walk \eqref{nonUniKraus}.
    The exact values (connected by a polygon) are shown for the same value ($t=10$) as the correction. In (b) the $1/t$ correction $C_1(\lambda,t)$ to the characteristic function $C_t(\lambda)$ is shown for $t=10$.}
\label{CharFuncFirstOrd}
\end{figure}

\subsection{Quantum walks with momentum shifts}
\label{subsec:QwMK}
Now we want to drop Assumption \ref{ass:nomom}, hence making the results of the preceding sections inapplicable. The quantum walk we are going to consider is closely related to the example of Sect. \ref{subs:NonUnitKO}, the decoherence will still be of Bernoulli type and it has Kraus operators
\begin{equation}
\label{eqn:MomShiftKO}
( K_1 \psi ) (p) =\frac{1}{2} (1+e^{i(p-q)}) \psi (p-q)\quad ,\quad  ( K_2 \psi )(p)=\frac{1}{2} (1-e^{-i(p-q)}) \psi (p-q)\,.
\end{equation}
The case $q=0$ recovers the quantum walk of the previous section which shows ballistic behavior. For $q\neq 0$ we will use Fourier methods to show that the behavior is diffusive. For this purpose we will assume that $q\in 2\pi\cdot \Rt$ , i.e. there are two numbers $m, n\in \Ir$ relatively prime, such that $q=2\pi\cdot n/m$. In position space these operators act in the following way
\[
(K_1\psi )(x)=\frac{1}{2}\left (e^{ixq}\psi (x) + e^{ixq}\psi (x+1)\right )\quad , \quad (K_2\psi )(x)=\frac{1}{2}\left (e^{ixq}\psi (x) - e^{ixq}\psi (x-1)\right )\,.
\]
Although these operators themselves are not translation invariant because of the position dependent phase factor $e^{iqx}$, the quantum walk constructed from them is translation invariant, as guaranteed by Corollary \ref{cor:KrausMomShift}. And indeed, for arbitrary operators $X$ the matrix elements of $K_i^*X K_i$ pick up phase factors which depend only on the difference of the row and column index of the operator, which is clearly a translation invariant operation.
\begin{prop}
Let the Kraus operators be given by \eqref{eqn:MomShiftKO} with $q\in 2\pi\cdot \Rt$. If $q\, \mathrm{mod} \,2\pi=0$ the behavior is ballistic according to Sect. \ref{subs:NonUnitKO}. For $q\, \mathrm{mod} \,2\pi \neq 0$ the ballistic order $Q(t)/t$ converges to a point measure at the origin. In the case $q\, \mathrm{mod} \,2\pi \neq \pi$ the diffusive order $Q(t)/\sqrt{t}$ converges to
\[
\mathrm{P}_1(x)=\frac{2}{\sqrt{3\pi}}e^{-4x^2/3} \,,
\]
if $q\, \mathrm{mod} \,2\pi = \pi$ the asymptotic distribution of $Q(t)/\sqrt{t}$ is given by
\[
\mathrm{P}_2(x)=\frac{1}{2\pi}\int d\lambda \,e^{-iy\lambda}e^{-\frac{3\lambda^2}{16}}\int dp\,\rho (p)\exp \left(\lambda^2\frac{\cos (2p)}{16}\right)\,.
\]
\end{prop}
\begin{proof}
Similarly to the preceding example, we will calculate the quantity $\tr \rho \Wt^t_\varepsilon (\idty ) $ in the limit $t\to \infty$ with appropriate scaling of $\varepsilon$. We begin by noting that for an arbitrary function $f(p,p' )$ the following identity holds
\begin{eqnarray}
\Wt_\varepsilon (f)(p,p') &=&\left(\sum_j  K_j^* f e^{i\varepsilon\lambda Q} K_j e^{-i \varepsilon \lambda  Q}\right)(p,p' )\nonumber\\
& = & \frac{1}{4}(1+e^{-ip})(1+e^{i(p'+\varepsilon \lambda)})f(p+q,p'+q)+\nonumber\\
&+&\frac{1}{4}(1-e^{ip})(1-e^{-i(p'+\varepsilon \lambda)})f(p+q,p'+q)\nonumber\\
&=&\frac{1}{2}(1+\cos (p'-p+\varepsilon\lambda )+i\cdot (\sin (p'+\varepsilon \lambda)-\sin (p)))f(p+q,p'+q)\,.\nonumber
\end{eqnarray}
Exploiting the condition $q =2\pi\cdot n/m$, with $n$ and $m$ relatively prime, together with this identity, we see that the characteristic function of the asymptotic distribution in $\varepsilon$-scaling is given by
\begin{equation}
\label{epsscal}
\lim\limits_{t\to\infty} \int dp\,\rho (p) \Wt^t_\varepsilon (\idty)(p,p)=
\int dp\,\rho (p)\lim\limits_{s\to\infty}\prod_{k=0}^{m-1} \frac{1}{2^s}(1+\cos {\scriptstyle (\varepsilon \lambda)}+i\cdot (\sin {\scriptstyle (p+kq+\varepsilon \lambda) }-\sin {\scriptstyle (p+kq)}))^s\,,
\end{equation}
with $t=m\cdot s$. Hence, in ballistic scaling $\varepsilon =1/t$ we get the result
\begin{eqnarray}
\lim\limits_{t\rightarrow \infty}C_{Q(t)/t}(\lambda )&=& \int dp\,\rho (p)\prod_{k=0}^{m-1} e^{\frac{i}{2m}\lambda\cos (p+kq)}\nonumber\\
&=&\int dp\,\rho (p) e^{\lambda\sum_{k=0}^{m-1}\frac{i}{2m}\cos (p+kq)}\nonumber\\
&=&\int dp\,\rho (p)\nonumber\\
&=&1\quad ,\nonumber
\end{eqnarray}
which follows from a Taylor expansion of \eqref{epsscal} in $\veps$ and $\lim_{t\to \infty}(1+x/t+\order (t^{-1}))^t=\exp (x)$.
Therefore, all the moments vanish and the random variable $Q(t)/t$ converges to a point measure at zero. On the other hand, some combinatorics tells us that
\begin{eqnarray}
&&\prod_{k=0}^{m-1} \frac{1}{2}(1+\cos (\varepsilon \lambda)+i\cdot (\sin (p+kq+\varepsilon \lambda)-\sin (p+kq)))\nonumber \\
&=&\prod_{k=0}^{m-1} \frac{1}{2}(2-\frac{1}{2}\varepsilon ^2\lambda ^2 + i\cdot \cos (p+kq)\varepsilon \lambda -\frac{i}{2} \sin (p+kq)\varepsilon ^2 \lambda ^2 + \order
(\varepsilon ^2))\nonumber \\
&=& \frac{1}{2^m}(2^m+2^{m-1}i\varepsilon \lambda\sum_k \cos (p+kq) - 2^{m-3}\varepsilon ^2\lambda^2 \left( \sum_k \cos (p+kq)\right)^2 +\nonumber \\
&&+2^{m-3}\varepsilon^2\lambda^2\sum_k \cos ^2(p+kq) -m2^{m-2}\varepsilon^2\lambda^2 - i2^{m-2}\varepsilon^2\lambda^2 \sum_k \sin (p+kq) + \order (\varepsilon^2)) \nonumber \\
&=& 1-m\frac{1}{4}\varepsilon^2\lambda^2 +\frac{1}{8}\varepsilon^2\lambda^2\sum_k \cos ^2 (p+kq) + \order (\varepsilon^2))\nonumber \\
&=& \left\{
\begin{array}{ll}
 1-\frac{3m}{16}\varepsilon^2\lambda^2 + \order (\varepsilon^2)&\mathrm{, if }\quad q\,\mathrm{mod}\, 2\pi \neq \pi \\
1-\left(3-\cos  (2p)\right)\frac{\varepsilon^2\lambda^2}{8} + \order (\varepsilon^2))&\mathrm{, if }\quad q\,\mathrm{mod}\, 2\pi = \pi
\end{array}\right. \, .\nonumber
\end{eqnarray}
The first equality is just a Taylor expansion in $\veps $, in the second equality we sorted terms according to $\veps $ and used the relation
\[
\sum_{k<l}\cos (p+kq)\cos (p+lq)=\frac{1}{2}\left(\sum_k \cos (p+kq)\right)^2 - \frac{1}{2}\sum_k \cos ^2(p+kq)\, ,
\]
and the third equality uses $\sum_{k=0}^{m-1} \cos (p+kq)=0$ for $m\cdot q \,\mathrm{mod}\, 2\pi=0 $. The last equality follows from the formula
\[
\sum_{k=0}^{m-1} \cos^2 (p+kq)=\frac{1}{2}\sum_{k=0}^{m-1} (1+\cos (2\cdot(p+kq)) = \frac{m}{2}+\delta_{m,2}\cos (2p)\,.
\]
Hence, for $q\,\mathrm{mod}\,2\pi\neq \pi$ and in diffusive scaling $\varepsilon =1/\sqrt{t}$ we have for arbitrary states $\rho$
\[
\lim\limits_{t\rightarrow \infty}C_{Q(t)/t}(\lambda )=\lim\limits_{s\rightarrow \infty} \int dp\,\rho (p)\left(1-\frac{3\lambda^2}{16s}+\order (s^{-1})\right)^s=\int dp\,\rho (p)e^{-\frac{3\lambda^2}{16}}=e^{-\frac{3\lambda^2}{16}}\, ,
\]
and the asymptotic probability distribution obtained via inverse Fourier transform reads
\[
\mathrm{P}_1(x)=\frac{2}{\sqrt{3\pi}}e^{- \frac{4x^2}{3}}\,.
\]
If $q\,\mathrm{mod}\,2\pi=\pi$ we get the result
\[
\lim\limits_{t\rightarrow \infty}C_{Q(t)/t}(\lambda )=e^{-\frac{3\lambda^2}{16}}\int dp\,\rho (p)\exp \left(\lambda^2\frac{\cos (2p)}{16}\right)\,.
\]
\end{proof}
In contrast to the case $q=0$, the quantum walk for $q\neq 0$ shows diffusive behavior. The asymptotic distribution together with the probability distribution for a finite number of time steps is shown in Fig. \ref{PlotMomShifts}.

\begin{figure}[htb]
  \centering
  \subfigure[]{\includegraphics[width=0.45\textwidth]{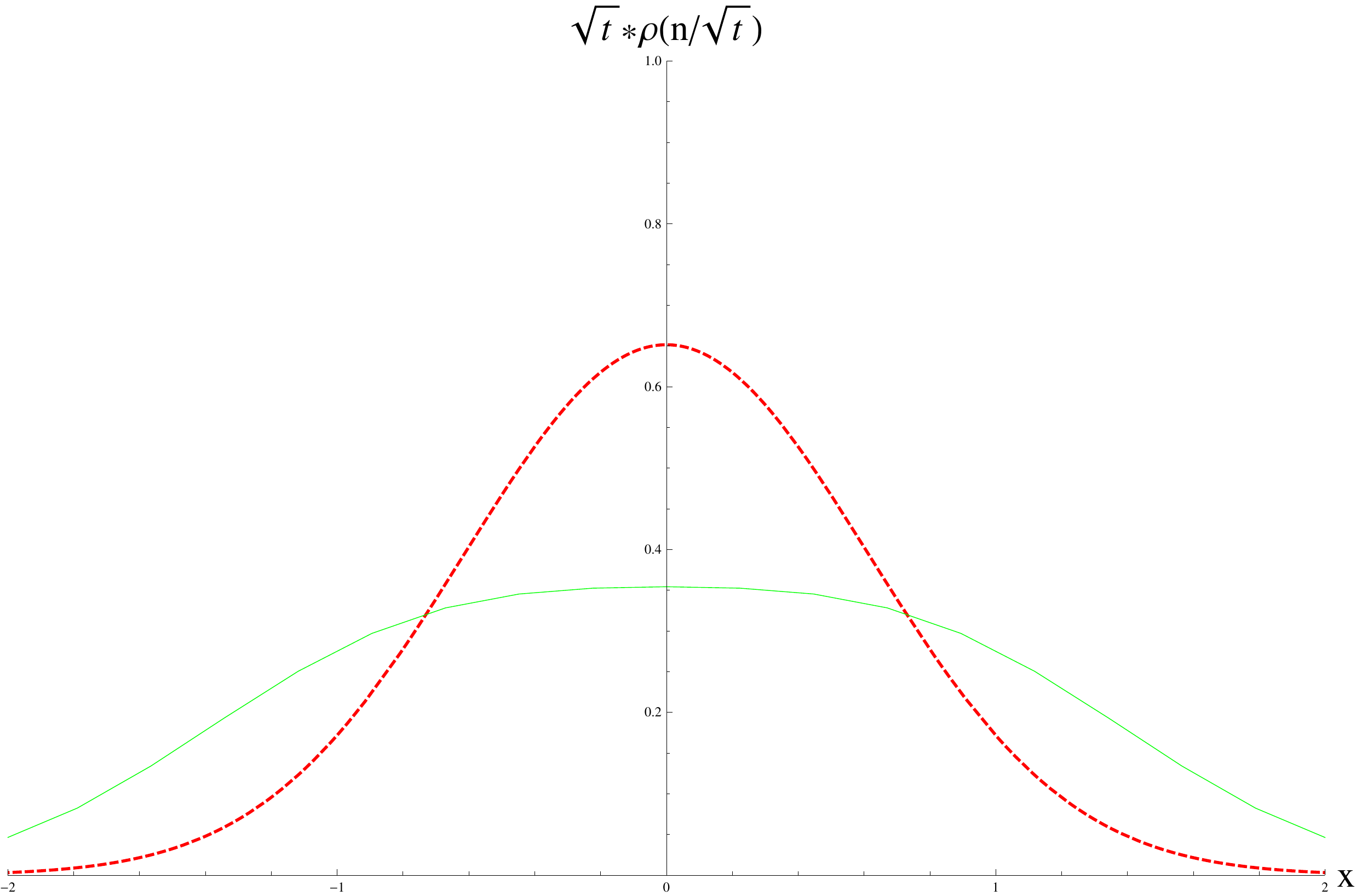}}
  \hspace{1cm}
  \subfigure[]{\includegraphics[width=0.45\textwidth]{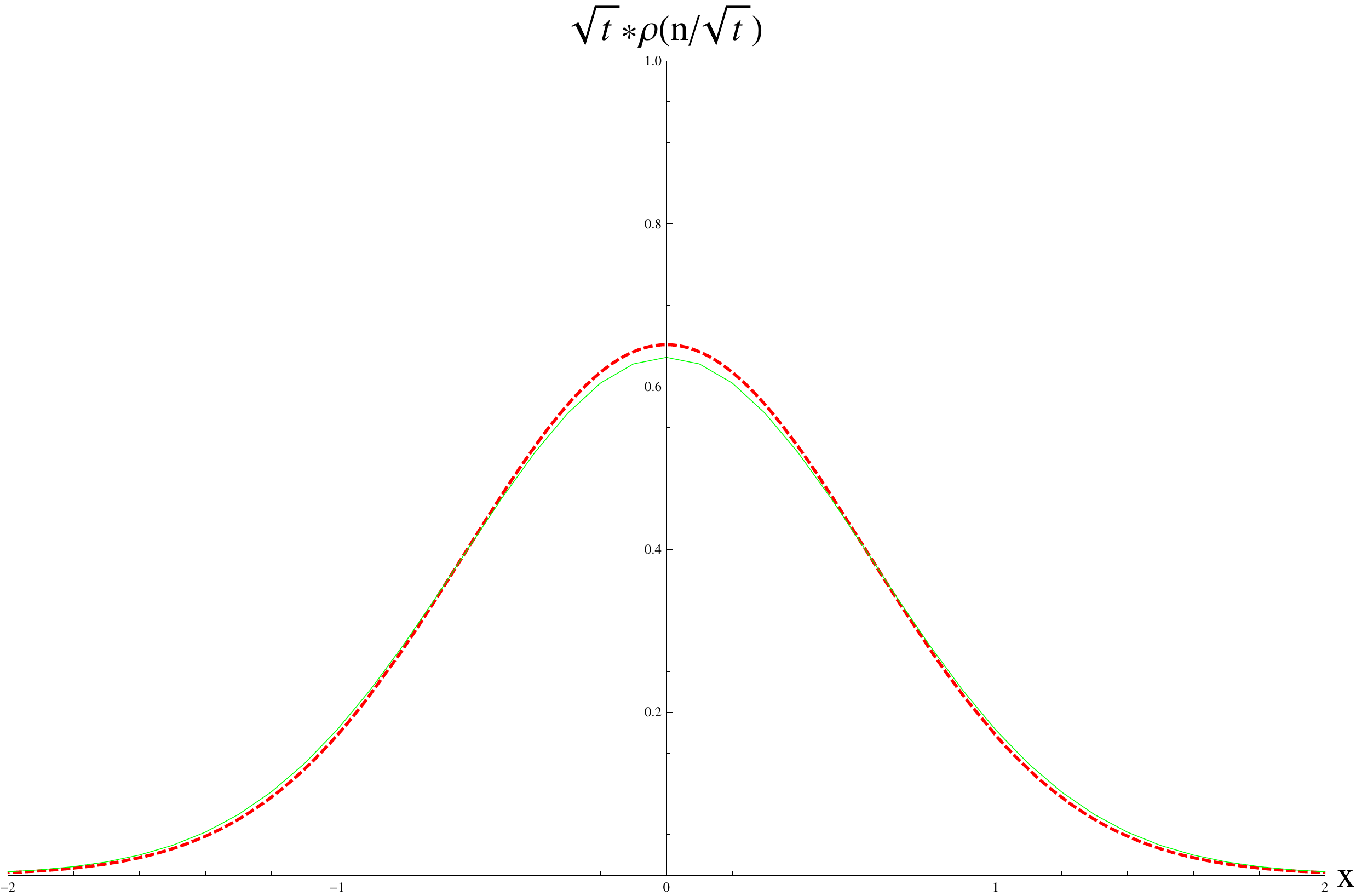}}
    \caption{(color online) Asymptotic position distribution for a walk \eqref{eqn:MomShiftKO} with momentum shift $q=\pi/16$ (green/dotted line) after (a) $20$ steps and (b) $100$ steps. The red/dashed line  is the asymptotic distribution.}
\label{PlotMomShifts}
\end{figure}

\subsection{Commuting Kraus operators\label{ex:commKraus}}

In this subsection we want to give an example that violates the non-degeneracy condition of \assref{chaos} and show that we can still compute the asymptotic position. We consider a decoherent quantum walk with Bernoulli type control process acting on a finite control space. The degeneracy of the eigenvalue $1$ will be exactly the dimension of the coin space due to the assumption that all the Kraus operators $K_j$ commute, but they must not have a common degenerate eigenspace. Since we will study uncorrelated noise in this example, we can forget about the classical control space of the Bernoulli process and model the decoherence by a quantum channel  without explicit classical control.
\begin{prop}\label{prob:comKraus}
  Suppose that a quantum walk $\Wa$ with Bernoulli type decoherence on a finite control space is implemented with normal and commuting Kraus operators $ K_j(p)$ without degenerate eigenvalues. The Kraus operators can then be decomposed in a common eigenbasis $ K_j(p) = \sum_\alpha k_{j,\alpha}(p) \ketbra{\psi_\alpha(p)}{\psi_\alpha(p)}$. If $\sum_j\overline{k_{j,\alpha}}k_{j,\beta}\neq 1$ for $\alpha\neq\beta$
  the asymptotic position distribution in ballistic scaling is given by
  \begin{align*}
    \lim_{t\rightarrow\infty} C_{Q(t)/t}(\lambda) =\int dp \; \tr \left(\rho(p) e^{i \lambda \cdot\tilde V}\right)
  \end{align*}
  with the vector operator
  \begin{align*}
  \tilde V_\tau=-i \sum_{j,\alpha} \overline{k_{j,\alpha}(p)}\ \frac{\partial k_{j,\alpha}(p)}{\partial p_\tau} \ketbra{\psi_\alpha(p)}{\psi_\alpha(p)}\;.
  \end{align*}
\end{prop}

\begin{proof}
  Since all the Kraus operators $ K_\gamma(p)$ commute and are normal by assumption, they can certainly be diagonalized in a common eigenbasis
  \begin{align*}
     K_j(p) = \sum_\alpha k_{j,\alpha}(p)  \ketbra{\psi_\alpha(p)}{\psi_\alpha(p)} \; .
  \end{align*}
One can check that a basis of invariant operators of $\Wa$ is given by the one-dimensional projections $\ketbra{\psi_\alpha(p)}{\psi_\beta(p)}$. Since $\Wa$ is a quantum channel it follows that $\sum_j \abs{k_{j,\alpha}}^2=1$ and together with the condition on the sums of eigenvalues for $\alpha\neq\beta$ this implies that $\Wa(A)=A$ holds only for operators in the span of $\{\ketbra{\psi_\alpha(p)}{\psi_\alpha(p)}=P_\alpha\}$.

In order to compute the ballistic scaling we are interested in the corrections to the eigenvalue one. Evaluation of the first perturbation order on an arbitrary operator $ R=\sum_\alpha r_\alpha P_\alpha$ in the span of the eigenoperators to the eigenvalue one  yields
\begin{align*}
 \Wa( R^\prime)(p)- R^\prime = \mu^\prime  R - \Wa^\prime( R)
\end{align*}
If evaluated with respect to the state $ \rho_\beta = \ketbra{\psi_\beta(p)}{\psi_\beta(p)}$ the left hand side of this expression vanishes and we have
\begin{align*}
\mu^\prime_\beta &= \frac{1}{r_\beta}\sum_{\alpha,j}r_\alpha \tr (   \rho_\beta   K_j^*(p)(\ketbra{\psi_\alpha}{\psi_\alpha}  K^\prime_j(p)) = \sum_j \overline{ k_{j,\beta}(p)}\tr\left(\rho_\beta\left. \frac{d}{d\veps} K_j(p+\veps\lambda)\right\vert_{\veps=0}\right)\\
&=  \sum_j \overline{ k_{j,\beta}(p)} \nabla k_{j,\beta}(p)\cdot\lambda +  \sum_{j,\alpha} \overline{k_{j,\beta}(p)} k_{j,\alpha}(p)\tr\left( \rho_\beta \left.\frac{d}{d\veps} P_\alpha(p+\veps\lambda)\right\vert_{\veps=0}\right) \; ,
\end{align*}
where we just used the definition of $ R$ and $ K_j^\prime(p)$ and the fact that $\{\ket{\psi_\alpha}\}$ is an eigenbasis of the Kraus operators.
Since $P_\beta (\partial_{p}P_\alpha) P_\beta=0$ holds for orthogonal projectors, the second summand on the right-hand side vanishes and we get
\begin{align*}
  \mu^\prime_\beta(p) = \sum_j \overline{ k_{j,\beta}(p)} \nabla k_{j,\beta}(p)\cdot\lambda
\end{align*}
Substituting these results into $\Wt_{1/t}(\idty)$ we find
\begin{align*}
 \lim_{t\rightarrow\infty} \Wt^t_{1 / t}(\idty)(p) = \lim_{t\rightarrow\infty}\sum_\alpha (1+\frac{1}{t}\mu^\prime_\alpha(p))^t P_{\alpha} = \sum_\alpha e^{\mu^\prime_\alpha(p)} P_\alpha(p)\; .
\end{align*}
Together with \eqref{charwballist} and \eqref{Wpowers} we have
\begin{align*}
  \lim_{t\rightarrow\infty}  C_{Q(t)}(\lambda) = \sum_\alpha \int dp \; \tr \left(\rho(p) e^{\mu^\prime_\alpha(p)} P_\alpha(p)\right)
\end{align*}
which ends the proof.
\end{proof}

\begin{cor}
For a unitary implemented walk, fulfilling the conditions of Proposition \ref{prob:comKraus},
the components of the operator $\tilde V$ are given as the weighted sums of the components of the  group velocity operators $V_j$ of the single walk operators $W_j$
\begin{align*}
  \tilde V_\tau = \sum_j \eta_j V_{\tau,j} = i \sum_{j} \eta_j\left(\sum_{\alpha}\frac{\partial\omega_{j,\alpha}(p)}{\partial p_\tau} P_\alpha\right)\;\;,
\end{align*}
where $\eta_j$ is the probability that the quantum walk $W_j$ is applied in a time step.
\end{cor}
\begin{proof}
  The unitarity of all the Kraus operators $ K_j$ implies that the eigenvalues $k_{j,\alpha}(p)$ are given by phases $\sqrt{\eta_j}e^{i \omega_{j,\alpha}(p)}$, $\eta_j$ being the probability to apply walk operator $ K_j$. This implies $\overline{k_{j,\alpha}(p)}=\eta_j k^{-1}_{j,\alpha}(p)$, and therefore we get
  \begin{align*}
  \overline{k_{j,\alpha}(p)}\!\;\frac{\partial k_{j,\alpha}(p)}{\partial p_\tau}=i \eta_j  \frac{\partial \omega_{j,\alpha}}{\partial p_\tau}\,.
  \end{align*}
Inserting this result into the definition of $\tilde V_\tau$ finishes the proof.
\end{proof}

To conclude this subsection we will look at a one dimensional example. For a given unitary one dimensional quantum walk $ W(p)$ we chose the two Kraus operators
\begin{align*}
  & K_1(p)= \frac{1}{2}(\idty +  W(p)) & &  K_2(p)= \frac{1}{2}(\idty -  W(p))\;.
\end{align*}
Since both $ K_j$ are just functions of the original walk operator $ W(p)$ they will certainly commute, and since $ W$ is unitary and therefore diagonalizable, so are the $ K_j$ with eigenvalues
\begin{align*}
  & k_{1,\pm}(p)= \frac{1}{2}(1 + e^{i\omega_\pm(p)}) & & k_{2,\pm}(p)= \frac{1}{2}(1 -e^{i\omega_\pm(p)})\;.
\end{align*}
One can calculate that for  $\omega_+(p)\neq\omega_-(p)$ the $k_{j,\pm}$ satisfy the sum-condition of Proposition \ref{prob:comKraus}, and therefore we can compute $\tilde V$ directly from the eigenvalues
\begin{align*}
\tilde V = -i\sum_{j,\alpha} \overline{k_{j,\alpha}(p)}\ \frac{\partial k_{j,\alpha}(p)}{\partial p} P_\alpha = \sum_\alpha \frac{\partial \omega_{\alpha}(p)}{\partial p} P_\alpha(p) = \frac{V_{ W}(p)}{2} \;,
\end{align*}
where  $P_\alpha$ are the eigenprojections and $V_{ W}$ the group velocity operator of the original quantum walk $ W$. So this decoherence model halves the velocities $\frac{\partial \omega_{\alpha}(p)}{\partial p}$ of the undisturbed quantum walk.